\documentclass{article}
\usepackage[a4paper]{geometry}
\usepackage[utf8]{inputenc}

\usepackage{enumerate}
\usepackage{mathtools}
\mathtoolsset{showonlyrefs=true}
\usepackage{todonotes}

\usepackage{url}
\usepackage{hyperref}
\usepackage{amsfonts}
\usepackage{amsmath,amssymb,color}
\usepackage{graphicx}
\usepackage{epstopdf}
\usepackage{multirow}
\usepackage{bm}
\usepackage{natbib}
\usepackage{verbatim}
\usepackage{float}
\usepackage[toc,page]{appendix}
\usepackage[onehalfspacing]{setspace}
\usepackage{setspace}
\usepackage{adjustbox}
\usepackage{rotating}
\usepackage{subcaption}
\usepackage{amsthm}
\usepackage{booktabs,dcolumn,caption}

\usepackage{color}
\usepackage{ulem}
\usepackage{tikz}
\usepackage{multirow}
\usepackage{graphicx}
\usepackage{adjustbox}
\usepackage{float}
\usepackage{footnote}
\usepackage{array}
\usepackage{longtable}
\usepackage{threeparttablex}
\usepackage{booktabs,dcolumn,caption}
\usepackage{amsmath,amssymb,color}
\usepackage{fnbreak}
\usetikzlibrary{arrows.meta,shapes.multipart,positioning}
\usepackage{mathtools}

\usepackage{mathabx}
\usepackage{array}
\usepackage{longtable}
\usepackage{threeparttablex}

\usepackage{chngcntr}

\newcommand{\zx}[1]{{{\color{red}  #1}}}
\newcommand{\aaa}{\boldsymbol \alpha}
\newcommand{\bbb}{\boldsymbol \beta}
\newcommand{\rrr}{\boldsymbol \gamma}
\newcommand{\XX}{\mathbf X}
\newcommand{\YY}{\mathbf Y}
\newcommand{\QQ}{\mathbb Q}
\newcommand{\WW}{\mathbf W}
\newcommand{\xx}{\mathbf x}
\newcommand{\yy}{\mathbf y}
\newcommand{\xxk}{\tilde{\mathbf x}}
\newcommand{\zz}{\mathbf z}
\newcommand{\ww}{\mathbf w}
\newcommand{\zzk}{\tilde{\mathbf z}}
\newcommand{\XXk}{\tilde{\mathbf{X}}}
\newcommand{\ZZ}{\mathbf Z}
\newcommand{\ZZk}{\tilde{\mathbf{Z}}}
\newcommand{\WWk}{\tilde{\mathbf{W}}}

\newcommand{\YYm}{\mathbf{Y}^{\operatorname{mis}}}
\newcommand{\YYo}{\mathbf{Y}^{\operatorname{obs}}}

\newcommand{\yyo}{\mathbf{y}^{\operatorname{obs}}}

\newcommand{\XXo}{\mathbf{X}^{\operatorname{obs}}}
\newcommand{\ZZo}{\mathbf{Z}^{\operatorname{obs}}}
\newcommand{\zzo}{\mathbf{z}^{\operatorname{obs}}}
\newcommand{\xxo}{\mathbf{x}^{\operatorname{obs}}}
\newcommand{\ZZm}{\mathbf{Z}^{\operatorname{mis}}}
\newcommand{\ZZko}{\tilde{\mathbf{Z}}^{\operatorname{obs}}}
\newcommand{\zzko}{\tilde{\mathbf{z}}^{\operatorname{obs}}}
\newcommand{\xxko}{\tilde{\mathbf{x}}^{\operatorname{obs}}}
\newcommand{\ZZalg}{\mathbf{Z}^{\operatorname{Alg1}}}
\newcommand{\Zalg}{{Z}^{\operatorname{Alg1}}}

\newcommand{\XXko}{\tilde{\mathbf{X}}^{\text{obs}}}

\newcommand{\SSigma}{\boldsymbol{\Sigma}}
\newcommand{\AAA}{\mathcal A}
\newcommand{\BBB}{\mathcal B}
\newcommand{\CC}{\mathcal C}
\newcommand{\DD}{\mathcal D}
\newcommand{\GG}{\mathbf G}
\renewcommand{\SS}{\mathbf S}
\newcommand{\SSS}{\mathcal S}
\newcommand{\PP}{\mathbb{P}}
\newcommand{\III}{\mathbb I}
\newcommand{\argmax}{\operatornamewithlimits{arg\,max}}

\newcommand{\TT}{\mathbf T}
\newcommand{\EE}{\mathbb E}
\newcommand{\LL}{\mathbf L}
\newcommand{\uu}{\mbox{$\mathbf u$}}
\newcommand{\vv}{\mbox{$\mathbf v$}}
\newcommand{\ttt}{\boldsymbol \theta}

\newtheorem{theorem}{Theorem}
\newtheorem{algorithm}{Algorithm}
\newtheorem{defi}{Definition}
\newtheorem{remark}{Remark}
\newtheorem{pro}{Proposition}
\newtheorem{lemma}{Lemma}

\providecommand{\keywords}[1]
{
  \small	
  \textbf{\textbf{Keywords---}} #1
}

\title{Variable Selection in Latent Variable Models via Knockoffs: An Application to International Large-scale Assessment in Education}

\title{\Large  Variable Selection in Latent Regression IRT Models via Knockoffs: An Application to International Large-scale Assessment in Education}
\author{ 
Zilong Xie\thanks{School of Mathematical Sciences, Fudan University, Shanghai, China}, Yunxiao Chen\thanks{Department of Statistics, London School of Economics and Political Science, London, UK}, 
Matthias von Davier\thanks{Lynch School of Education, Boston College, Chestnut Hill, Massachusetts, USA}, 
Haolei Weng\thanks{Department of Statistics and Probability, Michigan State University, East Lansing, Michigan, USA}
}
\date{} 

\begin{document}

\maketitle

\begin{abstract}
International large-scale assessments (ILSAs) play an important role in educational research and policy making. They collect valuable data on education quality and performance development across many education systems, giving countries the opportunity to share techniques, organizational structures, and policies that have proven efficient and successful. To gain insights from ILSA data, we identify non-cognitive variables associated with students' academic performance. This problem has three analytical challenges: 1) academic performance is measured by cognitive items under a matrix sampling design; 2) there are many missing values in the non-cognitive variables; and 3) multiple comparisons due to a large number of non-cognitive variables. We consider an application to the Programme for International Student Assessment (PISA), aiming to identify non-cognitive variables associated with students' performance in science. We formulate it as a variable selection problem under a general latent variable model framework and further propose a knockoff method that conducts variable selection with a  controlled error rate for false selections. 
\end{abstract}


\section{Introduction}
International large-scale assessments (ILSAs), including the Programme for International Student Assessment (PISA), Programme for the International Assessment of Adult Competencies (PIAAC), Progress in International Reading Literacy Study (PIRLS), and Trends in International Mathematics and Science Study (TIMSS), play an important role in educational research and policy making. They collect valuable data on education quality and performance development across many education systems in the world, giving countries the opportunity to share techniques, organizational structures, and policies that have proven efficient and successful \citep{singer2018international,von2012role}.

PISA is a worldwide study by the Organisation for Economic Co-operation and Development (OECD) in member and non-member nations intended to evaluate educational systems by measuring 15-year-old school students' scholastic performance in the subjects of mathematics, science, and reading, as well as a large number of	non-cognitive variables, such as students’ socioeconomic status, family background, and learning experiences. Students' scholastic performance is measured by response data from cognitive items that measure ability/proficiency in each of the three subjects, and non-cognitive variables are collected through non-cognitive questionnaires for students, school principals, teachers, and parents.  In this study, we focus on the knowledge domain of science in PISA 2015, where science was the assessment focus in this survey. Given the importance of science education 	\citep{national2012framework, uk2015national}, it is of particular interest for educators and policymakers to understand what non-cognitive variables (e.g., socioeconomic status, family background, learning experiences) are significantly associated with 
student's knowledge of science. Naturally, 
one would consider a regression model with students' performance in science as the response variable and the non-cognitive variables as predictors and identify the predictors with non-zero regression coefficients. Seemly straightforward, constructing such a regression model and then selecting the non-null variables is nontrivial due to three challenges brought by the complexity of the current problem. First, students' performance in science is not directly observed but instead measured by a set of test items.  The measurement is further complicated by a matrix sampling design adopted by PISA \citep{gonzalez2010principles}. That is, each student is administered a small subset of available cognitive items in order to cover an extensive content domain while not overburdening students and schools in terms of their time and administration costs. Consequently, one cannot simply calculate a total score as a surrogate for student science performance. We note that OECD provides plausible values, which are obtained using a multiple imputation procedure \citep{von2009plausible}, as a summary of each student's overall performance in each subject domain. However, it is not suitable to use a plausible value as the response variable when performing the current variable selection task. This is because the multiple imputation procedure for producing the plausible values involves the predictors through a principal component analysis step \citep[chapter 9,][]{organisation2016pisa}, due to which all the predictors are associated with the plausible values and thus, performing variable selection is not sensible. Second, students' non-cognitive variables are collected via survey questions, which contain many missing values. In fact, in the US sample considered in the current study, {around 6\% of the entries are missing, and the proportion of sample points that are fully observed is less than 26\%.}
Consequently, it is virtually impossible to conduct the regression analysis without a proper treatment of the missing values. Finally, PISA collects a large number of non-cognitive variables. In the current study of PISA 2015 data, we have 62 predictors, even though careful pre-processing is performed that substantially reduces the number of variables.  Due to the multiple comparison issues, it is a challenge to control for a reasonable error metric when conducting variable selection. 
 
We tackle these challenges through several methodological contributions. We introduce a latent construct for science knowledge and use an Item Response Theory (IRT) model \citep{chen2023IRT} to measure this latent construct based on students' responses to science items. The relationship between the latent construct and non-cognitive variables is further modeled through a structural model that regresses the latent construct onto the non-cognitive variables.
This structural equation model is often known as the latent regression IRT model, or simply the latent regression model \citep{mislevy1984estimating,von2010stochastic}. When there are many missing values in the non-cognitive variables, estimating the latent regression model is a challenge. To tackle this problem, we propose to model the predictors using a Gaussian copula model \citep{fan2017high,han2012composite}, which allows the predictors to be of mixed types (e.g., continuous, binary, ordinal). Thanks to the Gaussian copula model, we can estimate the latent regression model with a likelihood-based estimator. In dealing with multiple comparisons, we consider 
    the knockoff framework for controlled variable selection \citep{barber2015controlling,candes2018panning}. More specifically, we adapt the derandomized knockoffs method \citep{ren2021derandomizing} to the current latent regression model with missing values. This approach allows us to control the  Per Family Error Rate (PFER), i.e., the expected number of false positives among the detections. We choose the derandomized knockoff method instead of the Model-X knockoff method because the latter is a randomized procedure that may suffer from a high Monte Carlo error. The derandomized knockoff method leverages the Model-X knockoff method by aggregating the results from multiple knockoff realizations. 
    To our best knowledge,  this is the first time that missing data is considered in a knockoff approach with theoretical guarantees.

In real-world applications, especially in social sciences, missing data are commonly encountered. In addition, many variables of interest, such as individuals’ attitudes,  personality traits, and abilities, are latent constructs that are not directly observable.
They are often defined by multiple indicators and play the role of a response variable or predictors in a model \citep[Chapter 4,][]{skrondal2004generalized}. For example, 
the latent construct for students' science knowledge
is such a variable,  and it serves as the response variable in the latent regression analysis of PISA data. While we focus on data from an education survey and a tailored latent regression model, we describe the proposed knockoff method under a general latent variable model framework so that the proposed method can be applied to variable selection problems involving missing values, latent constructs or both. {\color{black}
Model selection of latent variable models is usually performed based on information criteria, such as the Akaike information criterion \citep[AIC;][]{akaike1974new} and the Bayesian information criterion \citep[BIC;][]{schwarz1978estimating}. These methods suffer from several issues under the current complex data setting. First, when the latent regression model involves many predictors,
an information criterion needs to be combined with a search method,  such as a stepwise selection method or a Lasso-type regularised estimation method. The search method is used to retain a smaller number of candidate models 
from the original model space that is exponentially large so that the information criterion can be computed. Even so, this approach may be computationally infeasible
 when there exist many latent variables or missing values, 
as a search method  needs to optimize 
many marginal likelihoods that involve high-dimensional integrals. For regularised estimation methods, the optimization additionally involves non-smooth regularisation terms and thus can be computationally even more time-consuming. Moreover, stepwise selection methods are greedy algorithms that lack a theoretical guarantee for identifying the true model. Second, the computational burden with the information-criteria-based methods mostly comes from handling missing data and latent variables. One may naturally wonder whether we can use a two-step procedure that first handles the missing data and latent variables using an off-the-shell missing data handling methods, such as imputation methods \citep{little2019statistical,van2018flexible} and missing indicator methods \citep{cohen1975,dardanoni2011regression,dardanoni2015model},  and then applies an information criterion to the imputed or augmented data. Unfortunately, such a procedure lacks theoretical justification and is often practically infeasible or inaccurate. For example, the missing indicator method cannot be performed when some predictors are latent variables. In addition, a small simulation study in Section F.2 of the supplementary material shows that the BIC performs poorly when calculated based on imputed data. Finally,  the proposed method is more flexible, as it allows the users to choose the
threshold for the PFER, allowing a trade-off between type
I and type II error rates. This is an advantage that model selection based on an information criterion does not offer.
 }

The remainder of the paper is structured as follows. Section \ref{sec:data} provides the background on the central substantive question -- how students' knowledge of science is associated with their non-cognitive variables -- and a description of the PISA 2015 data. In Section~\ref{sec:review}, we introduce the latent regression IRT model for studying the relationship between a latent construct of science knowledge and non-cognitive variables and a Gaussian copula model for handling missing predictors, which are of mixed types. Section~\ref{sec:varsec} proposes knockoff methods for 
 controlled variable selection under the latent regression IRT model with missing data.  The proposed method is evaluated via a simulation study in Section~\ref{sec:sim} and then applied to data from PISA 2015 in Section~\ref{sec:real}. 
Finally, we discuss the implications of our results and possible directions for future research in Section \ref{sec:dis}. Proof of theoretical results, details of computation, additional simulation studies, and further information about the PISA data are given in the supplementary material. 

\section{Background and Overview of PISA 2015 Data}\label{sec:data}
	
\subsection{Academic Achievement and Non-cognitive Predictors}

The term ``non-cognitive" typically refers to a broad range of personal attributes, skills, and characteristics representing one's attitudinal, behavioral, emotional, motivational, and other psychosocial dispositions. It is often used as a catch-all phrase encompassing variables that are potentially important for academic achievement but not measured by typical achievement or cognitive tests \citep{farkas2003cognitive}. Social science researchers have devoted considerable research effort towards identifying non-cognitive predictors of students' academic achievement \citep[e.g.,][]{duckworth2015measurement,richardson2012psychological,lee2018non}.

Science has changed our lives and is vital to the future prosperity of society. Thus, science education plays an important role in the modern education system \citep{national2012framework, uk2015national}. Identifying the predictors of science education helps educators, policymakers, and other stakeholders understand the psychosocial factors behind science education, which may lead to better policies and practices of science education. PISA, which collects both students' science achievement and non-cognitive variables, provides a great opportunity for identifying the key non-cognitive predictors of science achievement. 
 
	\subsection{PISA 2015 Data}
	
	PISA is conducted in a three-year cycle, with each cycle focusing on one of the three subjects, i.e., mathematics, science, and reading. PISA 2015 is the most recent cycle that focused on science. It collected data from 72 participating countries and economics. Computer-based tests were used, with assessments lasting a total of two hours for each student.
	Following a matrix sampling design, different students took different combinations of test items on science, reading, mathematics, and collaborative problem-solving. Test items involved a mixture of multiple-choice and 
	constructive-response questions. See \cite{organisation2016pisa} for the summary of the design and results of PISA 2015.
	
This study considers a subset of the PISA 2015 dataset. Specifically, to avoid modeling country heterogeneity, we considered data from a single country, the United States (US). After some data pre-processing which excluded observations with poor-quality data, the sample size is 5,685.  PISA 2015 contained 184 items in the science domain that were dichotomously or polytomously scored. Due to 
	the matrix sampling design of PISA,  on average, each student was only assigned {16.25\%} of the items. 
	
In addition, we consider non-cognitive variables collected by the student survey, which provides information about the students themselves, their homes, and their school and learning experiences. We constructed 62 variables as candidates in variable selection. These variables include {11} raw responses to questionnaire items (e.g., {GENDER (gender), LANGAH (language at home)}), {34} indices that OECD constructed (e.g., {CULTPO (cultural possession), HEDRES (home educational resources)}), and {17} composites that we constructed {based on students' responses to questionnaire items} (e.g., {OUT.GAM (play games out of school), OUT.REA (reading out of school)}). We decided to include these constructed variables rather than the corresponding raw responses for better substantive interpretations. {For some ordinal variables, certain adjacent categories were merged due to sample size considerations.}
Details of these 62 candidate variables are given in Section~\ref{sec:real} and the supplementary material. 
Unlike the cognitive items, students were supposed to answer all the items in the student survey. However, there are still many missing responses in the student survey data. Among the candidate variables,  20 variables have more than 5\% of their data missing, and the variable {DUECEC (duration in early childhood education and care)} has the largest missing rate, {37.17\%}. 
	

\section{Model Framework}\label{sec:review}
		
In this section, we describe a general latent variable model framework, which includes the latent regression model for analyzing PISA data as a special case. The model is defined through (1) a structural model, (2) a measurement model, and (3) a data missingness mechanism.

\begin{figure}[ht]
  \centering
  \includegraphics[scale=0.8]{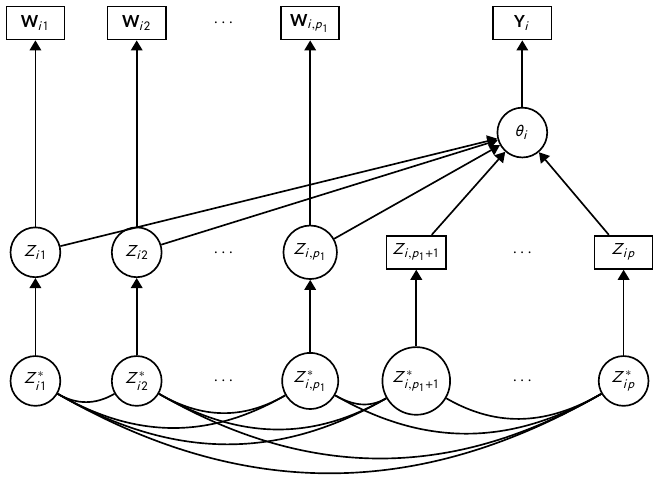}

\caption{{Path diagram for the general latent variable model framework. Variables within a circle represent unobserved or latent variables, while those within a rectangle represent observed variables. The measurement models are represented by the directed edges from $\theta_i$ to $\YY_{i}$ and those from $Z_{ik}$ to $\WW_{ik}$, $k = 1, 2, \ldots, p_1$. The structural model is represented by the directed edges from $Z_{ij}$s to $\theta_i$. The predictor model is represented by the directed edges from $Z^*_{ij}$ to $Z_{ij}$ and the undirected edges between $Z^*_{ij}$s.}}

\label{fig:latentreg}
\end{figure}

\subsection{Structural Model}

We consider data collected from $N$ observations. For each observation, there is a response variable $\theta_i$ and predictors $\ZZ_i = (Z_{i1}, \ldots, Z_{ip})^\top$, where $\theta_i$ and some or all entries of $\ZZ_i$ can be latent constructs measured by observed indicators. 
We allow the variables in $\ZZ_i$ to be binary, ordinal, continuous, or
a mixture of them. 
{\color{black}Without loss of generality, we assume that $Z_{i1}$, ..., $Z_{ip_1}$ are continuous latent constructs, with $p_1 = 0$ when all entries of $\ZZ_i$ are observable.}
In the PISA application, each observation is a student, $\theta_i$ represents the student's latent construct on science knowledge, and $\ZZ_i$ contains observable non-cognitive predictors.   
 
The structural model defines the joint distribution of $(\theta_i, \ZZ_i)$ through two steps -- (1) the conditional distribution of $\theta_i$ given $\ZZ_i$ and (2) the marginal distribution of $\ZZ_i$. 
A linear regression model is assumed for $\theta_i$ given $Z_{i1}$, ..., $Z_{ip}$. More specifically, 
for each variable $j$, we introduce a transformation $g_j(Z_j)$. When $Z_j$ is an ordinal variable with categories 
$\{0, ..., K_j\}$, the transformation function $g_j$ creates $K_j$ dummy variables, i.e., $g_j(Z_j) = (\III(\{Z_j\geq 1\}), ..., \III(\{Z_j \geq K_j\}))^\top$. For continuous and binary variables, $g_j$ is an identity link, i.e., $g_j(Z_j) = Z_j$. We assume
	\begin{equation}\label{eq:struc}
\theta_i \vert \ZZ_i \sim N(\beta_0 + \bbb_1^\top g_1(Z_{i1}) + \cdots + \bbb_p^\top g_p(Z_{ip}), \sigma^2),
	\end{equation}
where $\beta_0$ is the intercept, $\bbb_1$, ...,  $\bbb_p$ are the slope parameters, and $\sigma^2$ is the residual variance. Note that $\bbb_j$ is a scalar when predictor $j$ is continuous or binary and is a vector when the predictor is ordinal. Here, $\beta_0$, $\bbb_1$, ..., $\bbb_p$, and $\sigma^2$ are unknown and will be estimated from the model. The main goal of our analysis is to find predictors for which $\Vert \bbb_j\Vert \neq 0$. {The structural model is depicted in Figure~\ref{fig:latentreg}.}

Since $\ZZ_i$ may contain variables of mixed types, one cannot simply adopt a Gaussian assumption. Here, we consider a Gaussian copula model. This model introduces underlying random variables $\ZZ_i^* = (Z_{i1}^*, ..., Z_{ip}^*)^\top$, for which $\ZZ_1^*$, ..., $\ZZ_N^*$ are independent and identically distributed, following a $p$-variate normal distribution $N(\mathbf 0, \SSigma)$. We assume that the normal distribution is non-degenerate, i.e.,  $\text{rank}(\SSigma) = p$. Each underlying variable $Z_{ij}^*$ is assumed to marginally follow a standard normal distribution, i.e., the diagonal entries of $\SSigma$ are 1. Each predictor $Z_{ij}$ is assumed to be a transformation of its underlying variable $Z_{ij}^*$, denoted by $Z_{ij} = F_j(Z_{ij}^*)$.  	
	    For a continuous predictor $j$, let
	    $F_j(Z_{ij}^*) = c_j + d_jZ_{ij}^*,$
	where $c_j$ and $d_j$ are unknown parameters. 
	{For the latent constructs, we let $Z_{ij} = Z_{ij}^*$ as their location and scale need to be fixed for identification, i.e., $c_j = 0$ and $d_j = 1$, $j=1, ...,p_1$.}
	   For a binary or ordinal predictor $j$, let 
	   $F_j(Z_{ij}^*) =k$ if $Z_{ij}^*  \in (c_{jk}, c_{j,k+1}], k = 0, \ldots, K_j$, 
where $c_{j1}$, ..., $c_{jK_j}$ are unknown parameters, and $c_{j0} = -\infty$ and $c_{j,K_j+1} =\infty$. 
Note that $K_j=1$ for a binary variable and $K_j >1$ for an ordinal variable.
{The predictor model is also illustrated in Figure~\ref{fig:latentreg}.} 
 
We note that the above model specifies a joint distribution for $Z_{i1}, ..., Z_{ip}$. More specifically, 
let $\DD \subset \{1, ..., p\}$ be the set of dichotomous and polytomous predictors. We use $\Xi$ as  generic notation for the unknown parameters in the Gaussian copula model, including $\SSigma$ and the parameters in the transformations between $Z_{ij}^*$
and $Z_{ij}$. We further use $\phi(\cdot|\SSigma)$ to denote the density function of the multivariate normal distribution $N(\mathbf{0}, \SSigma)$. Then the density function of $\ZZ_i$ takes the form 
\begin{equation}\label{eq:density}
\begin{aligned}
    f(\zz\vert \Xi)&= \int\ldots\int \left\{\left(\prod_{j\in \DD }~dz_{j}^{*}\right) \right.\times \\ 
    &\left.\left[\phi(\zz^*\vert \SSigma) \times \left(\prod_{j\notin\DD} d_j^{-1}\right) \times  \left( \mathop{\prod}_{j\in \DD }\III (z_{j}^* \in  (c_{j, z_{j}-1},c_{j, z_{j}} ])\right) \right]\Bigg\vert_{z_{j}^* = \frac{z_{j} - c_j}{d_j}, j\notin\DD}.\right\} 
\end{aligned}
\end{equation}

\subsection{Measurement Model}\label{subsec:measurement}

We note that $\theta_i$ is a latent construct that is 
not directly observable. In addition, some variables in $\ZZ_i$ may also be latent constructs. The latent constructs are defined by observable data through a measurement model. We now specify this measurement model based on complete data (i.e., no data is missing). The treatment of missing values is left to Section~\ref{subsec:infer}. 
Let $\YY_i = (Y_{i1}, ..., Y_{iJ})^\top$ be the indicators for $\theta_i$. And let $\mathbf W_{ij} = (W_{ij1}, ..., W_{ijl_j})^\top$ be the indicators for $Z_{ij}$, $j=1, ..., p_1$.  The measurement model defines the conditional distribution of $(\YY_i^\top, \mathbf W_{i1}^\top, ..., \mathbf W_{ip_1}^\top)^\top$ given $(\theta_i, Z_{i1}, ..., Z_{ip_1})^\top$. 
This conditional distribution is specified by assuming  
(1) $\YY_i$ is conditionally independent of all the other variables given $\theta_i$, (2) $\mathbf W_{ij}$ is conditionally independent of all the other variables given $Z_{ij}$, for all $j = 1, ..., p$, and (3) the conditional model of $\YY_i$ given $\theta_i$ and those of $\mathbf W_{ij}$  given $Z_{ij}$. {These conditional models are visualized in Figure~\ref{fig:latentreg}.} We now elaborate on these conditional models. 

\medskip
\noindent
{\bf Conditional model of $\YY_i$ given $\theta_i$.} If $\theta_i$ is observable, then we just let $\YY_i = \theta_i$, and in this case, the conditional model of  $\YY_i$ given $\theta_i$ is degenerate. Otherwise, when $\theta_i$ is a latent construct, a unidimensional linear factor model or IRT model \citep[Chapter 3,][]{skrondal2004generalized} can be used for this conditional distribution, depending on the variable types in $\YY_i$. 
We assume that this measurement model satisfies the standard identifiability conditions. In the application to PISA data, students' science performance in science, $\theta_i$, is measured by cognitive items. In this application, $\YY_i$ contains students' responses to cognitive items, where the responses are either binary (correct/incorrect) or ordinal. In what follows, we describe the measurement model used in the scaling of PISA 2015 data \citep[chapter 9,][]{organisation2016pisa}. This model will be used in our simulation studies and application to PISA data.

		More specifically, this model assumes local independence, an assumption that is commonly adopted in IRT models \citep{embretson2000item}. That is, $Y_{ij}$, $j =1, ..., J$, are conditionally independent given $\theta_i$. For a dichotomous item $j$, the conditional distribution of $Y_{ij}$ given $\theta_i$ is assumed to follow a two-parameter logistic model (2PL, \cite{birnbaum1968some})  
		\begin{equation}\label{eq:2pl}
			\mathbb P(Y_{ij} = 1 \vert \theta_i) = \frac{\exp(a_j \theta_i +b_j)}{1+\exp(a_j \theta_i +b_j)},
		\end{equation}
		where $a_j$ and $b_j$ are two item-specific parameters. For a polytomous item $j$ {with $K_j+1$ categories},  $Y_{ij}$ given $\theta_i$ is assumed to follow  a generalized partial credit model (GPCM, \cite{muraki1992generalized}), for which  
		\begin{equation}\label{eq:gpcm}
		 \PP(Y_{ij} = k\vert \theta_i) = \frac{\exp\left[\sum\limits_{r = 1}^k  (a_j \theta_i +b_{jr})\right]}{1 + \sum\limits_{k' = 1}^{K_j}\exp\left[\sum\limits_{r = 1}^{k'}  (a_j \theta_i +b_{jr})\right]}, ~k = 1, 
		 \ldots, K_j,
		\end{equation}
		where {$a_j, b_{j1}, b_{j2}, \ldots, b_{j,K_j}$} are item-specific parameters.
	In OECD's analysis of PISA data,  the item-specific parameters are first calibrated based on item response data from all the countries and then treated as known when inferring the proficiency level of students or the proficiency distributions of countries  \citep[chapter 9,][]{organisation2016pisa}. We follow this routine when analyzing PISA data. Specifically, the item-specific parameters are fixed to the values used by OECD for scaling PISA 2015 data\footnote{The item parameters can be found from: \url{https://www.oecd.org/pisa/data/2015-technical-report/PISA2015_TechRep_Final-AnnexA.pdf}}.

	In the rest, we denote the conditional probability density/mass function of $\YY_i$ given $\theta_i$ at $\YY_i = \yy_i$ as $h(\yy_i\vert\theta_i; \Delta)$, where $\Delta$ denotes the unknown parameters in this conditional model. In the PISA application, $h(\yy_i\vert\theta_i; \Delta) = \prod_{j=1}^J P(Y_{ij} = y_{ij} \vert \theta_i)$, where $P(Y_{ij} = y_{ij} \vert \theta_i)$ follows \eqref{eq:2pl} or \eqref{eq:gpcm} depending on whether item $j$ is dichotomous or polytomous. As all the item parameters are pre-calibrated in this application, $\Delta$ becomes an empty vector and will not be estimated. In situations where the item parameters are unknown, $\Delta$ can be estimated from data; see Section F.1 in the supplementary material for a simulation study under this setting.


\medskip
\noindent
{\bf Conditional model of $\mathbf W_{ij}$ given $Z_{ij}$.}  
When $Z_{ij}$ is a latent construct, a unidimensional linear factor model or IRT model can be used for this conditional distribution, depending on the variable types in $\mathbf W_{ij}$. We assume that these measurement models satisfy the standard identifiability conditions. In the rest, we denote the conditional probability density/mass function of $\mathbf W_{ij}$ given $Z_{ij}$ at $\mathbf W_{ij} = \mathbf w_{ij}$ as 
\begin{equation}\label{eq:WcondZ}
    q_j(\mathbf w_{ij}\vert Z_{ij}; \Lambda_j),~j = 1, ..., p_1,
\end{equation}
where $\Lambda_j$ denotes the unknown parameters in this conditional model if they exist.  

	\subsection{Data Missingness and Statistical Inference}\label{subsec:infer}
	

In PISA data, as well as many other multivariate data in the social and behavioral sciences, there are often a substantial proportion of missing values. Here, we impose assumptions for data missingness. 
First, we assume that entries of $\YY_i$ are missing completely at random. This assumption is sensible in the PISA application, where the missing responses to cognitive items are due to the matrix sampling design of PISA. {\color{black}Second, we assume that $W_{i1}$, ..., $W_{ip_1}$ do not have missing values. This assumption is made for simplicity and can be easily relaxed. }
Finally, we assume that the missing data in 
$(Z_{i,p_1+1}, ..., Z_{ip})^\top$ are Missing At Random (MAR), which is a quite strong but commonly adopted assumption in missing data analysis  \citep{little2019statistical, van2018flexible}. 

More specifically,  let $\ww_i$ be the realisation of $\WW_i = (\WW_{i1}^\top, ..., \WW_{ip_1}^\top)^\top$, $i=1, ..., N$, and recall that $\Xi$ denotes the unknown parameters of the Gaussian copula model and $\Lambda_1$, ..., $\Lambda_{p_1}$ denote the parameters in the measurement models for $Z_{ij}, j = 1$. Let  $\BBB_i$ be an index set containing all the indicators $j$ such that $Y_{ij}$ is not missing. We let $\YYo_i = \{Y_{ij}: j \in \BBB_i\}$ be the observed indicators for $\theta_i$ and let 
$\YYm_i = \{Y_{ij}: j \notin \BBB_i\}$ be the missing ones. 
Similarly, we let $\AAA_{i}$ be the set indicating all the observed variables in $(Z_{i,p_1+1}, ..., Z_{i,p})^\top$, and let $\ZZo_i = \{Z_{ij}: j \in \AAA_i\}$ and  $\ZZm_i = \{Z_{ij}: j \notin \AAA_i\}$. 
Under the MAR assumption, the log-likelihood function for $\Xi$, $\Lambda_1$, ..., and $\Lambda_{p_1}$ takes the form
$l_1(\Xi, \Lambda_1, ..., \Lambda_{p_1})= \sum_{i=1}^N \log f_i(\ww_i, \zzo_i {\vert \Xi, \Lambda_1, ..., \Lambda_{p_1}})$, 
where 
\[f_i(\ww_i, \zzo_i\vert \Xi, \Lambda_1, ..., \Lambda_{p_1}) = \int \cdots\int f(\zz_i\vert \Xi) (\prod_{j=1}^{p_1} q_j(\mathbf w_{ij}\vert z_{ij}; \Lambda_j)) (\prod_{j\notin \AAA_i}d z_{ij}).\] 
Note that the  integrals in $f_i(\ww_i, \zzo_i\vert \Xi, \Lambda_1, ..., \Lambda_{p_1})$  are
with respect to $\ZZm_i$.

The maximum likelihood estimator for $\Xi, \Lambda_1, ..., \Lambda_{p_1}$ is given by 
\begin{equation}\label{eq:mmlcopula}
\begin{aligned}
(\hat\Xi, \hat\Lambda_1, ..., \hat \Lambda_{p_1}) = \argmax_{\Xi, \Lambda_1, ..., \Lambda_{p_1}}\quad & l_1(\Xi, \Lambda_1, ..., \Lambda_{p_1})\\
 \text{ subject to }\quad&\Sigma_{jj} = 1,~j = 1, \ldots, p,\\
  & {d_j > 0,  j \notin \DD,} ~ c_{j1} < c_{j2} < \ldots < c_{jK_j}, j\in \DD. 
\end{aligned} 
\end{equation}
We note that this optimization problem involves high-dimensional integrals and constraints.  We adopt a stochastic proximal gradient algorithm proposed in \cite{zhang2022computation}. 
In this algorithm, the integrals are handled by Monte Carlo sampling of the missing values, and the unknown parameters are updated by stochastic proximal gradient descent, in which constraints are handled. The details of this algorithm can be found in the supplementary material.

Given $\hat\Xi, \hat\Lambda_1, ..., \hat \Lambda_{p_1}$
from \eqref{eq:mmlcopula}, one can estimate the rest of the unknown parameters, including the regression coefficients in \eqref{eq:struc} that are of major interest.  
We denote  {$\bbb = (\bbb_1^\top, ...,\bbb_p^\top)^\top$}.  Let $\yyo_i$ 
be the realisation of $\YYo_i$, $i=1, ..., N$. The log-likelihood for $\bbb$, {$\beta_0$,} $\sigma^2$ and $\Delta$
takes the form
\small{\begin{equation}\label{eq:loglik}
\begin{aligned}
    &l_2(\bbb, {\beta_0, } \sigma^2, \Delta) = \\
    &\sum_{i=1}^N \log \left[\int \cdots \int \left(\prod_{j\notin \AAA_i }~dz_{ij}\right) f(\zz_i\vert \hat \Xi) \left(\prod_{j=1}^{p_1} q_j(\mathbf w_{ij}\vert z_{ij}; \hat \Lambda_j)\right) f_i(\yyo_i \vert \zz_i; \bbb, {\beta_0, } \sigma^2, \Delta)\right],
\end{aligned}
\end{equation}}
where $f_i(\yyo_i \vert \zz_i; \bbb, {\beta_0,} \sigma^2, \Delta)$ is the conditional density function of  $\YYo_i$ given $\ZZ_i = \zz_i$ 
\small{\begin{equation}\label{eq:Yobsmodel}
    \begin{aligned}
        &f_i(\yyo_i \vert \zz_i; \bbb, {\beta_0,} \sigma^2, \Delta) = \\
        & \frac{1}{\sqrt{2\pi\sigma^2}} \int \cdots \int d\theta_i (\prod_{j \notin \BBB_i}dy_{ij}) h(\yy_i\vert \theta_i; \Delta) \exp\left(-\frac{(\theta_i - (\beta_0 + \bbb_1^\top g_1(z_{i1}) + \cdots + \bbb_p^\top g_p(z_{ip}))^2}{2\sigma^2}\right) .
    \end{aligned}
\end{equation}}
We estimate $\bbb$, ${\beta_0}$ ,$\sigma^2$ and $\Delta$ by maximising $l_2(\bbb, {\beta_0,} \sigma^2, \Delta)$.
Similar to the optimisation \eqref{eq:mmlcopula}, the maximisation of $l_2(\bbb, {\beta_0}, \sigma^2, \Delta)$ also involves high-dimensional integrals. 
We carry out this optimization using a stochastic Expectation-Maximisation (EM) algorithm\footnote{The stochastic proximal gradient algorithm used for the optimization problem \eqref{eq:mmlcopula} can also be used to solve the current optimization problem. The stochastic EM algorithm is chosen as it tends to converge empirically faster for the current problem.} \citep{nielsen2000stochastic, zhang2020improved}. The details are given in the supplementary material. 


\section{Variable Selection via Knockoffs}\label{sec:varsec}

\subsection{Problem Setup and Knockoffs} 

As mentioned previously, our goal is to solve a model selection problem, i.e., to find the non-null predictors for which $\Vert \bbb_j\Vert \neq 0$. We hope to control the statistical error in the model selection to assure that most of the discoveries are indeed true and replicable. This is typically achieved by controlling for a certain risk function, such as the false discovery rate, the $k$-familywise error rate, and the per familywise error (PFER); see  \cite{janson2016familywise} and \cite{candes2018panning}. 
Let $\hat {\mathcal S}$ and $\mathcal S^*\subset \{1, ..., p\}$ be the selected and true non-null predictors, respectively.  
 The current study concerns the control of PFER, defined as $\mathbb E |\hat{\mathcal S}\setminus \mathcal S^*|$, where $|\cdot\vert$
denotes the number of elements in a set. 

The knockoff method is a general framework for controlled variable selection.  The key to a knockoff method is the construction of knockoff variables, where the knockoff variables mimic the dependence structure within the original variables but are null variables (i.e., not associated with the response variable). They serve as negative controls in the variable selection procedure that help identify the truly important predictors while controlling for a certain risk function, such as the PFER. Many knockoff methods have been developed \citep{barber2015controlling, barber2019knockoff,candes2018panning,fan2019rank,fan2020ipad, janson2016familywise, sesia2019gene,romano2020deep}. Many knockoff methods are based on the model-X knockoff framework \citep{candes2018panning}, which is very flexible and can be extended to the current setting involving missing data and mixed-type predictors. However, one drawback of the model-X knockoffs is that it only takes one draw of the knockoff variables through Monte Carlo sampling. As a result, this procedure often suffers from high uncertainty brought about by the Monte Carlo error, even though the risk function is controlled. To alleviate this uncertainty, which has important implications on the interpretability of the variable selection results, we adopt the derandomized knockoff method \citep{ren2021derandomizing}. This method can substantially reduce the Monte Carlo error by 
aggregating the selection results across multiple runs of a knockoff algorithm. In what follows, we first introduce the way of constructing knockoff variables under the joint model described in the above section and then introduce a derandomized knockoff procedure for controlling PFER.

\subsection{Constructing Knockoffs with Missing Data}\label{subsec:construction}

We extend the concept of knockoffs to the missing data setting. 
To control the variable selection error with the knockoff procedure introduced below, a Stronger MAR condition is needed. It is called the SMAR condition as introduced in Definition~\ref{def:smar} below. 
{\begin{defi}[SMAR condition]\label{def:smar}
Let $\XX_i = (\XX_{i1}^\top, ..., \XX_{ip}^\top)^\top$, such that $\XX_{ij} = \WW_{ij}$ if $j = 1, ..., p_1$ and $\XX_{ij} = Z_{ij}$ otherwise. 
Consider the conditional distribution of  $\AAA_i$ given $\XX_i$.  Let $q(\aaa \vert \xx_i)$ denote the conditional probability mass function of $\AAA_i$ given $\XX_i$. We say the SMAR condition holds with respect to the non-null variables $\mathcal S^*$, if $q(\aaa \vert \xx) = q(\aaa \vert \xx')$ holds, for any $\aaa$, $\xx = (\xx_{1}^\top, ..., \xx_{p_1}^\top)^\top$, and $\xx' = ({\xx_{1}'}^\top, ..., {\xx_p'}^\top)^\top$ satisfying $\{\xx_j: j\in \{1, ..., p\} \cap \mathcal S^*\} = \{\xx'_j: j\in \{1, ..., p\} \cap \mathcal S^*\}$.
\end{defi}}
{This SMAR condition says that the 
probability of being missing is the same within groups defined by the observed non-null variables. 
It is stronger than MAR because MAR only requires   
 $q(\aaa \vert \xx) = q(\aaa \vert \xx')$ to hold, for any $\aaa$, $\xx = (\xx_{1}^\top, ..., \xx_{p_1}^\top)^\top$, and $\xx' = ({\xx_{1}'}^\top, ..., {\xx_p'}^\top)^\top$ satisfying
$\{\xx_i: i\in \aaa\} = \{\xx_i': i\in \aaa \}$, i.e.,   the 
probability of being missing is the same within groups defined by the observed variables, regardless of whether they are in $\mathcal S^*$ or not. On the other hand, the SMAR condition is weaker than  Completely Missing at Random (MCAR), as MCAR implies that $q(\aaa \vert \xx) = q(\aaa \vert \xx')$ for all  $\aaa$, $\xx$, and $\xx'$. }
{Throughout the rest, the SMAR condition is assumed.}

{\begin{defi}[Knockoffs]\label{def:knockoff}
     Suppose that the SMAR condition in Definition~\ref{def:smar} holds and {the true values of $\Xi$, $\Lambda_1$, ..., $\Lambda_{p_1}$, $\bbb$, $\beta_0$, $\sigma^2$, and $\Delta$ are known}. 
     Under the setting in Section~\ref{sec:review}, we say that $(\WWk_i, \ZZko_i)$ is a knockoff copy of $(\WW_i, \ZZo_i)$, if there exists a random vector $\WWk_i = (\WWk_{i1}^\top, ..., \WWk_{ip_1}^\top)^\top$, where $\WWk_{ij} = (\WWk_{ij1}, \ldots, \WWk_{ijl_j})^\top$ for each $j = 1, \ldots, p_1$, as well as underlying variables $\ZZ^*_i = (Z^*_{i1}, ..., Z^*_{ip})^\top$ and $\tilde \ZZ_i^* = (\tilde Z_{i1}^*, ..., \tilde Z_{ip}^*)^\top$, such that
    \begin{enumerate}
    
        \item $\ZZo_i = \{F_j(Z_{ij}^*): j \in \AAA_i\}$ and $\tilde\ZZ^{obs}_i = \{F_j(\tilde Z_{ij}^*): j \in \AAA_i\}$;  

        \item { $\YYo_i, \WW_i$, and $\ZZk^*_{i}$ are conditionally independent given $\ZZ_i^*$;

        { 
        \item Let $\ZZ_i = (F_1(Z_{i1}^*, ..., F_p(Z_{ip}^*))^\top$. Then {$\ZZ_i$ follows the Gaussian copula model~\eqref{eq:density}, $\WW_{ij}$ given $Z_{ij}$ follows the conditional model \eqref{eq:WcondZ} for each $j=1, \ldots, p_1$, and $\YYo_i$ given $\ZZ_i$ follows the conditional model \eqref{eq:Yobsmodel}.}

        \item Let $\ZZk_i = (F_1(\tilde Z_{i1}^*), ..., F_p(\tilde Z_{ip}^*))^\top$. Then for each $j=1, \ldots, p_1$, 
        the conditional probability density/mass function of $\WWk_{ij}$ given $\tilde{Z}_{ij}$ is the same as \eqref{eq:WcondZ}
        }
        .
        
        \item For any subset $\mathcal S \subset \{1,...,p\}$, $(\ZZ^*_i, \tilde \ZZ^*_i)_{\text{swap}(\mathcal S)}$ and $(\ZZ^*_i, \tilde \ZZ^*_i)$ are identically distributed. }
    \end{enumerate}
\end{defi}}
Recall that $F_j(\cdot)$ is the transformation between each predictor and its underlying variable. In addition, the vector {$(\ZZ^*_i, \tilde \ZZ^*_i)_{\text{swap}(\mathcal S)}$} is obtained from {$(\ZZ_i^*, \tilde \ZZ_i^*)$  }
by swapping the entries {$Z_{ij}^*$} and {$\tilde Z_{ij}^*$} for each $j \in \mathcal S$; for example, with $p = 3$ and $\mathcal S = \{1,3\}$, {$(Z^*_{i1}, Z^*_{i1}, Z^*_{i3}, \tilde Z^*_{i1}, \tilde Z^*_{i2}, \tilde Z^*_{i3})_{\text{swap}(\{1,3\})} = (\tilde  Z^*_{i1}, Z^*_{i2}, \tilde  Z^*_{i3}, Z^*_{i1}, \tilde Z^*_{i2}, Z^*_{i3})$}. 
We compare the current definition of knockoffs under a missing data setting with the standard definition for model-X knockoffs in \cite{candes2018panning}.  
The model-X knockoff framework assumes no missing data in predictors $\ZZ_i$, and $\ZZ_1$, ..., $\ZZ_N$ are independent and identically distributed. Therefore, the definition of model-X knockoff omits the subscript $i$. On the other hand, the current analysis depends on $\AAA_i$, which differs across observations. Consequently, knockoffs are defined for each $\ZZo_i$. When there {are no unobservable predictors and} no missing data, i.e., {$p_1 = 0$ and } $\AAA_i = \{1, ..., p\}$, $i = 1, ..., N$, the current definition coincides with the definition in \cite{candes2018panning}.  Note that stronger conditions are needed for the construction of knockoffs when there exist missing data. These conditions (e.g., SMAR) are needed to ensure that the joint distribution of 
{ $\YYo_i$, $\AAA_i$, $\WW_i$, $\ZZo_i$, $\WWk_i$, and $\ZZko_i$ 
remains identical when swapping the null indices}, which is essential for establishing 
the exchangeability property \citep{candes2018panning}
for controlling variable selection error. 
{Specifically, 
under Definition~\ref{def:knockoff}, $(\WWk_i, \ZZko_i)$ and $\YYo_i$ are likely not conditionally independent given $(\WW_i, \ZZo_i)$.}
Consequently, 
when constructing the knockoff variables {$(\WWk_i, \ZZko_i)$}, one needs information  from not only {$(\WW_i, \ZZo_i)$} but also {$\YYo_i$}, to compensate for the missing information. In other words, 
the joint distribution of {$\YYo_i, \WW_i$, and $\ZZo_i$} is needed to construct {$(\WWk_i, \ZZko_i)$}. 

In what follows, we present an algorithm for constructing knockoffs {$(\WWk_i, \ZZko_i)$} under Definition~\ref{def:knockoff}. 
  To ensure the exact satisfaction of 
Definition~\ref{def:knockoff}, we assume that the true model parameters are known. In practice, we plug an estimate of the parameters into the algorithm; see Section~\ref{subsec:robust} for theoretical justifications and further discussions.

\begin{algorithm}[Constructing knockoff copies]\label{alg:knockcons2}~
	\begin{itemize}
		\item[]{\bf Input:} Observed data {$\YYo_i$, $\WW_i,$} and $\ZZo_i$, $i = 1, ..., N$, the {true model parameters $\Xi$} of the Gaussian copula model, {the true parameters $\Lambda_1$, ..., $\Lambda_{p_1}$ in the measurement models for $Z_{ij}$,} and the {true parameters $\bbb, \beta_0, \sigma^2, \Delta$ in the conditional model of $\YY_i$ given $\ZZ_i$}. 
		
		\item[]{\bf Step 1:} Sample underlying variables $\ZZ_i^*$
		from  their conditional distribution given {$\YYo_i$, $\WW_i,$ and $\ZZo_i$}.
    \item[]{\bf Step 2:} Sample $\ZZk_i^*$ given $\ZZ_i^*$, where $(\ZZ_i^*,\ZZk_i^*)$ jointly follows a multivariate normal distribution with mean zero and covariance matrix
		\begin{equation}\label{eq:jointcov}
		\GG = 	\begin{pmatrix} \SSigma &  \SSigma-\SS\\  \SSigma-\SS & \SSigma\end{pmatrix},
		\end{equation}
		where $\SSigma$ is the correlation matrix in the Gaussian copula model, and $\SS$ is a diagonal matrix specified in such a way that the joint covariance matrix $\GG$ is positive semidefinite. The construction of $\SS$ is based on the   Minimize the Reconstructability (MVR) procedure   \citep{spector2022powerful}.
		
		\item[]{\bf Step 3:} Obtain $\ZZk_i$ from $\ZZk^*_i$, where $\tilde Z_{ij} = F_j(\tilde Z^*_{ij})$ for each $j = 1,\ldots, p$.  

            {\item[]{\bf Step 4:} Sample $\WWk_{ij}$ from the conditional distribution $q_j(\cdot\vert\tilde{Z}_{ij}; \Lambda_j)$ for each $j = 1,\ldots, p_1$.}
		
		\item[]
		{\bf Output:} Knockoff copy {$(\WWk_{i}, \ZZko_i)$}, where $\ZZko_i = \{\tilde Z_{ij}^{obs}: j\in \AAA_i\}$.
	\end{itemize}	
\end{algorithm}

\begin{pro}\label{prop:knockZmiss}
	The output {$(\WWk_{i}, \ZZko_i)$} from  Algorithm~\ref{alg:knockcons2} satisfies Definition~\ref{def:knockoff}. 
\end{pro}

The proof of this proposition is given in the supplementary material. 
Figure~\ref{fig:LSGCknockoff} below gives the path diagram for the generation of knockoff copies.  We provide several remarks on the algorithm.
This algorithm allows for mixed types of predictors 
under a Gaussian copula model, which
extends the multivariate Gaussian model for knockoff construction considered in \cite{barber2015controlling} and \cite{candes2018panning}. 
When all the predictors are continuous, the Gaussian copula model degenerates to the multivariate Gaussian model. In that case, and if there is no missing data, then Algorithm~\ref{alg:knockcons2} coincides with the knockoff construction method in \cite{candes2018panning}, except that \cite{candes2018panning} uses the Mean Absolute Correlation (MAC) procedure to construct the $\SS$ matrix. 
 
\begin{figure}[ht]
	\centering
	\adjustbox{max height = 0.15\textheight}{
\includegraphics{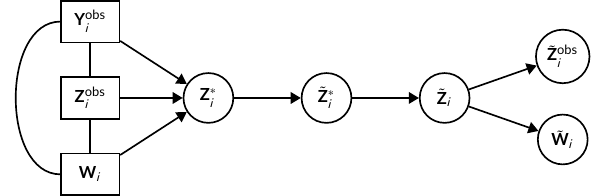}
}
	
\caption{Path diagram for constructing {$(\tilde{\WW}_i, \ZZko_i)$.}}\label{fig:LSGCknockoff}	   
\end{figure}

  

When {$\WW_i$ or $\ZZo_i$} contains binary or ordinal variables, the sampling of $\ZZ_i^*$ is not straightforward.  However, we can obtain approximate samples via Gibbs sampling. Thanks to the underlying multivariate normality assumption, each step of the Gibbs sampler only involves sampling from univariate normal or truncated normal distributions.  Details of the Gibbs sampler are given in the supplementary material. 
	We compute the diagonal matrix $\SS$ in Step~2 of the algorithm using the  MVR procedure   \citep{spector2022powerful}, which tends to be more powerful than the MAC procedure adopted in 
	 \cite{barber2015controlling} and \cite{candes2018panning}.  

\subsection{Variable Selection via Derandomised Knockoffs}

We now describe a knockoff procedure for variable selection with a controlled PFER. Suppose that knockoff copies {$(\WWk_{i}, \ZZko_i)$}, $i = 1, ..., N$, have been obtained using Algorithm~\ref{alg:knockcons2}. For ease of exposition, we denote 
$\ZZo = \{\ZZo_i\}_{i=1}^N$, $\ZZko = \{\ZZko_i\}_{i=1}^N$, {$\WW = \{\WW_i\}_{i=1}^N$, $\WWk = \{\WWk_i\}_{i=1}^N$,  and $\YYo = \{\YYo_i\}_{i=1}^N$.} We define a knockoff statistic that measures the importance of each predictor. 
	\begin{defi}[Knockoff statistic]\label{def:feature}
		Consider a statistic {$T_j$} taking the form 
  $T_j = t_j\left({(\WW, \ZZo)}, {(\WWk, \ZZko)}, \YYo\right)$
		for some function {$t_j$}, where {$(\WWk, \ZZko)$} are knockoffs satisfying Definition~\ref{def:knockoff}.
		This statistic is called a knockoff statistic for the $j$th predictor
		if it satisfies the \textit{flip-sign property}; that is for any subset $\mathcal S \subset\{1, \ldots, p\}$,
            \begin{equation}
			t_{j}\left({\{(\WW, \ZZo), (\WWk, \ZZko)\}_{\mathrm{swap}(\mathcal S)}}, \YYo\right) = \begin{cases}t_{j}\left({(\WW, \ZZo), (\WWk, \ZZko)}, \YYo\right), & j \notin \mathcal S, \\ -t_{j}\left({(\WW, \ZZo), (\WWk, \ZZko)}, \YYo\right), & j \in \mathcal S,\end{cases}
		\end{equation}
        where {$\left\{(\WW, \ZZo), (\WWk, \ZZko)\right\}_{\mathrm{swap}(\mathcal S)}$ }is obtained by swapping 
        \begin{itemize}
            \item[]{(1)} the entries of $\WW_{ij}$ and $\WWk_{ij}$  for each~$j \in \mathcal S \cap \{1, 2, \ldots, p_1\}$, $i=1, ..., N$;
            \item[]{(2)} the entries $Z_{ij}$ and $\tilde Z_{ij}$ for each~$j \in \mathcal S \cap \AAA_i$, $i=1, ..., N$.
        \end{itemize}
        
	\end{defi}
	
	The flip-sign property in Definition~\ref{def:feature} 
	is key to guaranteeing valid statistical inference from finite samples. However, to achieve a good power, {$T_j$} should also provide evidence regarding whether $\|\bbb_j\| = 0$. See Section~3 of \cite{candes2018panning} for a generic method of constructing {$T_j$} and specific examples. 
	In this study, we will focus on knockoff statistics constructed based on the likelihood function. 
	More specifically, we incorporate the knockoff variables into {the general latent variable model} defined in Section~\ref{sec:review}. That is, the measurement model remains the same, while the structural model becomes 
\begin{equation}\label{eq:struc2}
	\theta_i \vert \ZZ_i, \ZZk_i \sim N(\beta_0 + \bbb_1^\top g_1(Z_{i1}) + \cdots + \bbb_p^\top g_p(Z_{ip}) + \rrr_1^\top g_1(\tilde Z_{i1}) + \cdots + \rrr_p^\top g_p(\tilde Z_{ip}), \sigma^2),
	\end{equation}
	where $\ZZ_i$ and  $\ZZk_i$ are defined in Definition~\ref{def:knockoff}. {Since $\tilde \ZZ_{i}^*$ and $\YYo_i$ are conditionally independent given $\ZZ^*_i$, the true value of $\rrr_j$ is $\mathbf 0$, $j=1, ..., p$}, though these parameters will be estimated when constructing the knockoff statistics.  {The general latent variable model corresponding to this is depicted by the path diagram shown in Figure~\ref{fig:latentreg2} .}
 
\begin{figure}
  \centering
  \adjustbox{max height = 0.3\textheight}{
\includegraphics{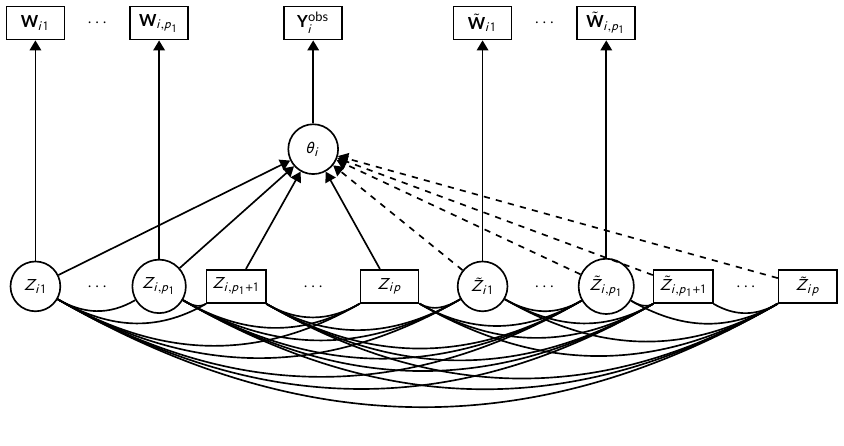}
}

\caption{
    Path diagram for {the general latent variable model involving knockoff variables.} 
    The interpretation is similar to that of Figure~1. 
     The directed edges from the knockoff variables {$\tilde{Z}_{ij}$s} to $\theta_i$ are drawn with dashed lines, as the true values of the corresponding coefficients are zero.
    }

    \label{fig:latentreg2}
\end{figure}

{Suppose that the values of $\Xi$, $\Lambda_1$, ..., $\Lambda_{p_1}$, $\beta_0$, $\sigma^2$, and $\Delta$ are known.} 
The likelihood function for $(\boldsymbol{\beta}, \boldsymbol{\gamma})$ under
this extended latent regression model takes the form
\small{\begin{equation}\label{eq:likelihood}
\begin{aligned}
    {\tilde l_2  (\bbb, \rrr)} =&\sum_{i=1}^N \log \left\{\int \cdots \int \left[ \left(\prod_{j\notin \AAA_i }~dz_{ij}\right) \left(\prod_{j\notin \AAA_i }~d\tilde z_{ij}\right) \right.\right.\times \\
    &\left.\left. f(\zz_i, \zzk_i \vert \Xi) { \left(\prod_{j=1}^{p_1} q_j(\mathbf w_{ij}\vert z_{ij};  \Lambda_j)\right)\left(\prod_{j=1}^{p_1} q_j(\tilde{\mathbf w}_{ij}\vert \tilde{z}_{ij};  \Lambda_j)\right)} f_i(\yyo_i \vert \zz_i, \zzk_i; \bbb,\rrr, {\beta_0,} \sigma^2, \Delta)\right]\right\}.
\end{aligned}
\end{equation}}
Here, {$f_i(\yyo_i \vert \zz_i, \zzk_i; \bbb,\rrr, \sigma^2, \Delta)$} is the density function of the conditional distribution of {$\YYo_i$ } given $\ZZ_i = \zz_i$ and $\ZZk_i = \zzk_i$; that is, 
\small{\begin{equation*}
    \begin{aligned}
      &f(\yyo_i \vert \zz_i, \zzk_i; \bbb,\rrr, {\beta_0,} \sigma^2,\Delta) =\\
      & \frac{1}{\sqrt{2\pi\sigma^2}} \int \cdots \int \left[d\theta_i (\prod_{j \notin \BBB_i}dy_{ij})  \times\right.
      \left.h(\yy_i\vert \theta_i; \Delta) \exp\left(-\frac{(\theta_i - (\beta_0 + \bbb_1^\top g_1(z_{i1}) + ... + \rrr_1^\top g_1(\tilde z_{i1}) + ... + \rrr_p^\top g_p(\tilde z_{ip}))^2}{2\sigma^2}\right)\right]. 
    \end{aligned}
\end{equation*}}
In addition, $f(\zz_i, \zzk_i \vert \Xi)$ denotes the density function of the Gaussian copula model for $(\ZZ_i, \ZZk_i)$, noting that this density function is completely determined by the parameters $\Xi$ of the Gaussian copula model for $\ZZ_i$; see the supplementary material for the specific form of $f(\zz_i, \zzk_i \vert \Xi)$.

A knockoff statistic {$T_j$} is constructed based on {$\tilde l_2(\bbb, \rrr)$}. 	Specifically, consider the maximum likelihood estimator based on {$\tilde l_2(\bbb, \rrr)$
\begin{equation}\label{eq:latregest}
(\tilde \bbb, \tilde \rrr) = \argmax_{\bbb, \rrr} ~ \tilde l_2(\bbb, \rrr).
\end{equation}
}	
Then a knockoff statistic can be constructed as 
	\begin{equation}\label{eq:knockoff}
		{T_j}  =  \operatorname{sign}(\|\tilde \bbb_j^\dagger\|  - \|\tilde \rrr_j^\dagger\|)\max\left\{\|\tilde \bbb_j^\dagger\|/\sqrt{p_j}, \|\tilde \rrr_j^\dagger\|/\sqrt{p_j}\right\},
	\end{equation}
	where $p_j$ is the dimension of $\bbb_j$ (or equivalently that of $\rrr_j$), and $\tilde \bbb_j^\dagger = \operatorname{Cov}(g_j(Z_{ij}))^{\frac{1}{2}}\tilde\bbb_j$  and
	 $\tilde \rrr_j^\dagger = \operatorname{Cov}(g_j(Z_{ij}))^{\frac{1}{2}}\tilde\rrr_j$ are standarised coefficients.	
	\begin{pro}\label{prop:knockffstat}
	 {Assume that the values of $\Xi$, $\Lambda_1$, ..., $\Lambda_{p_1}$, $\beta_0$, $\sigma^2$, and $\Delta$ are known and the knockoffs satisfy Definition~\ref{def:knockoff}. Then}   
  {$T_j$}  given by  \eqref{eq:knockoff} satisfies Definition~\ref{def:feature}.
	\end{pro}
	
The proof of this proposition is given in the supplementary material. 
Similar to the estimation of the latent regression model without knockoffs, the optimization problem \eqref{eq:latregest} can be solved using a stochastic EM algorithm. We remark that the statistic \eqref{eq:knockoff} is a special case of the Lasso coefficient-difference statistic given in \cite{candes2018panning} when the Lasso penalty is set to zero. Since the sample size $N$ is often much larger than $p$ in ILSA applications, this likelihood-based knockoff statistic performs well in our simulation study and real data analysis.  For higher-dimensional settings,  a Lasso coefficient-difference statistic may be preferred; see \cite{candes2018panning}. 
	  
	We now adapt the derandomized knockoff method \citep{ren2021derandomizing} to the current problem. This method achieves PFER control by aggregating	the results from multiple runs of a baseline algorithm proposed in \cite{janson2016familywise}. This baseline algorithm is summarised in Algorithm~\ref{alg:pfer2} below. 

{
    \begin{algorithm}[Baseline algorithm for PFER control \citep{janson2016familywise}]\label{alg:pfer2}~
	\begin{itemize}
		\item[]{\bf Input:} Observed data {$\YYo, \WW,$ and $\ZZo$}, a PFER level $\nu \in \mathbb Z_+$,  the {true model parameters $\Xi$} of the Gaussian copula model, {the true parameters $\Lambda_1$, ..., $\Lambda_{p_1}$ in the measurement models for $Z_{ij}$,} and the {true parameters $\bbb, \beta_0, \sigma^2, \Delta$ in the conditional model of $\YY_i$ given $\ZZ_i$}.

        \item[]{\bf Step 1:} Generate knockoffs $(\WWk, \ZZko)$ using Algorithm~\ref{alg:knockcons2}.

        \item[]{\bf Step 2:} 
        Compute a set of knockoff statistics {$T_1, ..., T_p$} 
         {using equations \eqref{eq:latregest} and \eqref{eq:knockoff}}.
         
        \item[]{\bf Step 3:} Compute the threshold 
		$\tau = \inf \left\{t>0: 1+|\left\{j: {T_{j}} < -t \right\}|=v\right\}.$
		We let $\tau = -\infty$  if the set on the right-hand side is an empty set. 
		
		
		\item[]
		{\bf Output:} $\hat{\mathcal S} = \{j: {T_{j}} > \tau\}$.
		
	\end{itemize}
	
\end{algorithm}
}

\begin{pro}\label{prop:pfer0}
$\hat{\mathcal S}$ given by Algorithm~\ref{alg:pfer2} satisfies $\mathbb E|\hat{\mathcal S}\setminus \mathcal S^*| \leq   \nu$, 
i.e.,  the PFER can be controlled at level $\nu$.  

\end{pro} 

{
\begin{algorithm}[Derandomised knockoffs \citep{ren2021derandomizing}]\label{alg:derandom2}~
	\begin{itemize}
		\item[]{\bf Input:}  Observed data {$\YYo, \WW,$ and $\ZZo$}, the number of runs $M$ of the baseline algorithm,  a selection threshold $\eta$, a PFER level $\nu \in \mathbb Z_+$, the {true model parameters $\Xi$} of the Gaussian copula model, {the true parameters $\Lambda_1$, ..., $\Lambda_{p_1}$ in the measurement models for $Z_{ij}$,} and the {true parameters $\bbb, \beta_0, \sigma^2, \Delta$ in the conditional model of $\YY_i$ given $\ZZ_i$}.
				
		\item[]
		{\bf Step 1:} For each $m = 1, \ldots, M$, run Algorithm \ref{alg:pfer2} independently and obtain the selection set $\hat{\mathcal S}^{(m)}$.
		
		\item[]
		{\bf Step 2:} For each $j = 1, \ldots, p$, compute the selection frequency 
		$$\Pi_{j}=\frac{1}{M} \sum_{m=1}^{M} \mathbb{I}\left(j \in \hat{\mathcal S}^{(m)}\right).$$
		
		
		\item[]
		{\bf Output:} $\hat{\mathcal S}=\left\{j \in \{1,\ldots,p\} :\Pi_{j} \geq \eta\right\}.$ 
		
	\end{itemize}
	
\end{algorithm}
}

Following the theoretical result in \cite{ren2021derandomizing} when the threshold $\eta$ is chosen properly, Algorithm~\ref{alg:derandom2} guarantees to control PFER at level $\nu$. We provide a simplified version of this result in Proposition~\ref{prop:pfer} below.


\begin{pro}\label{prop:pfer}

If for any  $\eta \in (0,1)$, the condition $\mathbb P(\Pi_j \geq \eta) \leq \mathbb E[\Pi_j]$ holds for every $j \notin \mathcal S^*$, then $\hat{\mathcal S}$ given by Algorithm~\ref{alg:derandom2} satisfies $\mathbb E|\hat{\mathcal S}\setminus \mathcal S^*| \leq   \nu$, i.e., the PFER can be controlled at level $\nu$. In particular, assuming that the probability mass function  of $\Pi_j$ is monotonically non-increasing for each $j \notin \mathcal S^*$, $\mathbb P(\Pi_j \geq \eta) \leq \mathbb E[\Pi_j]$ holds for  $M = 31$ and $ \eta = 1/2$. 
\end{pro} 

While noting that other choices are possible, we set $M=31$ and $\eta = 1/2$, which is also the default choice in 
\cite{ren2021derandomizing}. We also note that the statistics $\Pi_j$, $j=1, ..., p$, rank the importance of the predictors. The predictors with $\Pi_j \geq \eta$ are selected as the non-null variables.

\subsection{A Three-step Procedure When Model Parameters are Unknown and Its Robustness}\label{subsec:robust}

The knockoff procedure described previously requires the true joint model for {$\YY_i$, $\WW_i$, and $\ZZ_i$}, which is infeasible in practice. 
When the true model is known,  the variable selection problem becomes trivial since the null and non-null variables can be directly identified from the true model. In practice, we first estimate the model parameters and then conduct variable selection based on the estimated model. This procedure involves three steps.
    First, estimate the parameters $\Xi$ in the Gaussian copula model {as well as the parameters $\Lambda_1, \ldots, \Lambda_{p_1}$ in the measurement models for $Z_{ij}$.} This is done by the maximum likelihood estimator \eqref{eq:mmlcopula}. 
Second, estimate the parameters {$\bbb$, $\beta_0$, $\sigma^2$, and $\Delta$} based on the log-likelihood \eqref{eq:loglik}, where the estimated Gaussian copula model $\hat \Xi$ {and the estimated measurement models $\hat \Lambda_1, ... ,\hat \Lambda_{p_1}$ are }
plugged in. 
{Third, select variables by plugging the estimated parameters $\hat \Xi$, $\hat \Lambda_j$s, $\hat \bbb$, $\hat \beta_0$, $\hat \sigma^2$, and $\hat{\Delta}$ into {Algorithm \ref{alg:pfer2}} or \ref{alg:derandom2}.}



Empirically,  simulation results in Section~\ref{sec:sim} show that  PFER is well controlled when we apply the above three-step procedure. 
Theoretically, by plugging into the estimated model rather than the true model, the PFER can no longer be exactly controlled as described in Propositions~\ref{prop:pfer0} and \ref{prop:pfer}. Following a similar proof strategy as in {\cite{barber2020robust}}, we show that this procedure is robust, in the sense that the resulting PFER is controlled near $\nu$ if the plug-in model is sufficiently accurate. Note that {\cite{barber2020robust}} only consider the robustness of model-X knockoffs for controlling false discovery rate and does not cover PFER. 

{More precisely, we use $\mathbb P$ and $\mathbb Q$ to denote the true and plug-in models, respectively. 
 Consider a pair of  $i$ and $j$, satisfying $j\in (\{1, ..., p_1\}\cup\AAA_i)$.
  {We consider $\XX_i$ in Definition~\ref{def:smar}.
  We further let $\XXo_{i, -j} = \{\XX_{ik}: k \in (\{1, ..., p_1\}\cup\AAA_i)\backslash\{j\}\}$.
  We also define $\XXk_i$ and $\XXko_i$ similar to $\XX_i$ and $\XXo_i$, respectively, but with  $\WW_{ij}$ replaced by $\WWk_{ij}$ and with $Z_{ij}$ replaced by $\tilde Z_{ij}$.} 
  Let
$\PP_{ij}(\xx_{ij} \vert {\xxo_{i,-j}}, \yyo_i)$  denote the conditional density function of $\XX_{ij}$ given $\XXo_{i, -j} = {\xxo_{i,-j}}$ and $\YYo_i = \yyo_i$ under the true model $\PP$.
Let $\QQ_{ij}(\xxko_{i,-j}, \xxk_{ij}   \vert  {\xxo_{i, -j}}, \xx_{ij}, \yyo_i)$ denote the conditional density function of $(\XXko_{i, -j}, \XXk_{ij})$ given $\XXo_{i, -j} = {\xxo_{i,-j}}$, $\XX_{ij} = \xx_{ij}$ and $\YYo_i = \yyo_i$ under the plug-in model $\QQ$.
We define 
$$\hat{\text{KL}}_j = \sum\limits_{i:j\in(\{1, ..., p_1\}\cup\AAA_i)} \log\left(\frac{\PP_{ij}(\XX_{ij} \vert {\XXo_{i, -j}}, \YYo_i )\cdot\QQ_{ij}(\XXko_{i,-j},\tilde \XX_{ij}   \vert  {\XXo_{i, -j}}, \XX_{ij}, \YYo_i )}{\PP_{ij}(\XXk_{ij} \vert {\XXo_{i, -j}}, \YYo_i ) \cdot \QQ_{ij}(\XXko_{i,-j}, \XX_{ij}   \vert  {\XXo_{i, -j}}, \tilde \XX_{ij}, \YYo_i )}\right).$$
Here $\{i:j\in(\{1, ..., p_1\}\cup\AAA_i)\} = \{1, ..., p\}$ if $j \in \{1, ..., p_1\}$, and $\{i:j\in(\{1, ..., p_1\}\cup\AAA_i)\} = \{i:j\in\AAA_i\}$ otherwise. Note that the numerator inside of the logarithm corresponds to the true data generation mechanism for $(\XX_{ij}, \XXko_i)$, and the denominator corresponds to that when switching the roles of $\XX_{ij}$ and $\XXk_{ij}$.
 }
$\hat{\text{KL}}_j$ 
 can be viewed as an observed Kullback–Leibler (KL) divergence 
that measures the discrepancy between the true model  $\mathbb P$ and its approximation $\mathbb Q$, with $\hat{\text{KL}}_j = 0$ when $\QQ = \PP$. We remark that this definition of $\hat{\text{KL}}_j$ is consistent with that in \cite{barber2020robust}. However, the  $\hat{\text{KL}}_j$ in \cite{barber2020robust}  can be further simplified with a pairwise exchangeable property of their procedure under a model-X knockoff setting without missing data, while this pairwise exchangeable property does not always hold for the current procedure due to the involvement of {$\YYo$} and thus, the current $\hat{\text{KL}}_j$  cannot be further simplified.

\begin{theorem}\label{thm:robust}
Under the definitions above,  for any $\epsilon \geq 0$, consider the null variables for which $\hat{\text{KL}}_j \leq \epsilon$. If we use a modified Algorithm~\ref{alg:pfer2} that generates knockoffs under the plug-in model $\mathbb Q$ {which is assumed to be independent of data}, 
then the expected number of rejections that correspond to such nulls obeys
$\mathbb E |\{j: j\in \hat{\mathcal S}\setminus \mathcal S^* \mbox{~and~} \hat{\mbox{KL}}_j \leq \epsilon\} | \leq \nu e^{\epsilon}.$
In particular, $\hat{\text{KL}}_j=0$, when $\QQ=\PP$. \end{theorem}  
When $\PP = \QQ$, we can set $\epsilon = 0$, and thus,   Theorem~\ref{thm:robust} implies Proposition~\ref{prop:pfer0}. 
This property of robustness carries over to the derandomized procedure. We define 
$\Pi_j^\dagger =  \left(\sum_{m=1}^{M} \mathbb{I}\left(j \in \hat{\mathcal S}^{(m)} \mbox{~and~} \hat{\text{KL}}^{(m)}_j \leq \epsilon\right)\right)/M,$
where $\hat{\mathcal S}^{(m)}$ is the selection in the $m$th run of   modified Algorithm~\ref{alg:derandom2} that generates knockoffs under the plug-in model $\mathbb Q$, and 
$\hat{\text{KL}}^{(m)}_j$
is the corresponding observed KL divergence based on the knockoffs from the $m$th run. 

\begin{theorem}\label{thm:derandpfer}
Under the definitions above, for any $\epsilon \geq 0$, consider the null variables for which $\hat{\text{KL}}_j^{(m)} \leq \epsilon$ for all $m=1, ..., M$. We use a modified Algorithm~\ref{alg:derandom2} where knockoffs are generated under the plug-in model $\mathbb Q$ which is assumed to be independent of data, and obtain selections $\hat{\mathcal S}$.  
If the condition 
$\mathbb P(\Pi_j^\dagger \geq \eta) \leq \mathbb E[\Pi_j^\dagger]$
holds for every $j \notin \mathcal S^*$,  then
$\mathbb E |\{j: j\in \hat{\mathcal S}\setminus \mathcal S^* \mbox{~and~} \hat{\text{KL}}_j^{(m)}  \leq \epsilon,  m =1, ..., M\} | \leq \nu e^{\epsilon}.$

If the probability mass function  of $\Pi_j^\dagger$ is monotonically non-increasing for each $j \notin \mathcal S^*$, 
$\mathbb P(\Pi_j^\dagger \geq \eta) \leq \mathbb E[\Pi_j^\dagger]$
holds for  $M = 31$ and $ \eta = 1/2$. 
\end{theorem}

\section{Simulation Study}\label{sec:sim}

In this section, we conduct a simulation study to evaluate the performance of the proposed knockoff method. We check if the PFER can be controlled at the targeted level when the three-step procedure described in Section~\ref{subsec:robust} is applied. The power of variable selection will also be assessed.  

We set $p = 100, J = 60$, and consider $N \in \{1000, 2000, 4000\}$ for comparing power under different sample sizes. It leads to three settings. For each setting, we generate 100 independent replications. The data are generated as follows.  We divide the predictors into five blocks, each containing  10 continuous variables and 10 binary variables. {Ordinal variables or unobservable variables are not included in this study for simplicity. In Section F.3 of the supplementary material, we present a simulation study that includes unobservable variables.}

We consider the following design for the correlation matrix $\SSigma$ of the underlying variables $\ZZ_{i}^*$,
which is similar to the one used in \cite{grund2021treatment} that concerns analyzing missing data in ILSAs. This correlation matrix mimics the correlation structure in ILSA data.
	{(a)} 
	Within block 1, the correlation between every pair of variables is 0.6. 
	{(b)} 
	Within block 2, the correlation between every pair of variables is 0.6. For the $10$-by-10 submatrix recording the correlations between variables in blocks 1 and 2, the diagonal entries are set to be 0.3, and the off-diagonal entries are set to be 0.15. 
	{(c)} Within block 3, the correlation between every pair of variables is 0.6. The variables in block 3 have a correlation of 0.15 with each variable in blocks 1 and 2. 
	{(d)} Within block 4, the correlation between every pair of variables is 0.3. For the 10-by-30 submatrix recording the correlations between variables in block 4 and those in blocks 1 to 3, all the entries take a value of 0.15, except that the diagonal entries of the 10-by-10 submatrix corresponding to blocks 4 and 1 are set to 0.3. 
	{(e)} Within block 5, the correlation between every pair of variables is 0.3. For the 10-by-40 submatrix recording the correlations between variables in block 5 and those in blocks 1 to 4, the entries are generated independently from a uniform distribution over the interval $[0.1, 0.2]$.  
The same correlation matrix is used in all 100 replications. The heat map of this correlation matrix is given in {Figure~\ref{fig:heat}}. This correlation matrix has a maximal eigenvalue of 22.73 and a minimal eigenvalue of 0.11.

\begin{figure}[ht]
    \centering
    \includegraphics[scale=0.5]{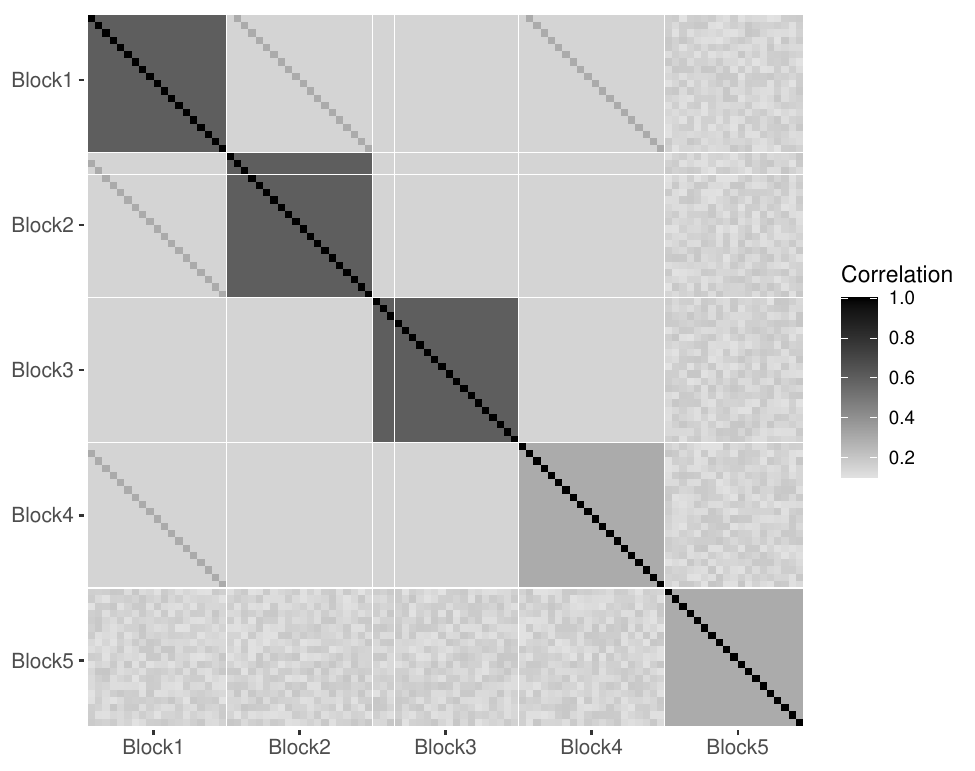}
    \caption{Heatmap of the designed correlation matrix in the simulation study.}
    \label{fig:heat}
\end{figure}

The rest of the Gaussian copula model is set as follows. For continuous variables, we set $c_j=0$ and $d_j=1$. For the binary variables, we set their threshold parameters $c_{j1}$ to take one of the values in $(-1.2, -0.3, 0, 0.3, 1.2)$ iteratively (i.e., $c_{11,1} = -1.2, c_{12,1} = -0.3$ and so on). 
Regarding the parameters in the structural model, {we set the intercept $\beta_0 = 0$, $\beta_j = 0.5$ for $j = 1, 22, 43, 64, 85$, $-0.5$ for $j = 11, 32, 53, 74, 95$, and 0 for the rest of the variables}. Under this setting, 
the non-zero coefficients are distributed uniformly among the variables.  We further set $\sigma^2 = 1$ for the residual variance.

Data missingness is generated following the SMAR condition. For each observation $i$, we generate a random variable $R_i$ from a categorical distribution with support $\{1, 2, ..., 5\}$, satisfying $P(R_i = k) = 0.2$, for all $k=1, ..., 5$. The data missingness is determined by $R_i$ and the non-null variables. Let $\mathcal S_k^*$ denote the set of non-null variables in the $k$th block. For observation $i$, when $R_i = k$, we let all the variables in $\mathcal S_k^*$ be observed. For each of the rest of the variables $j$, its probability of being missing is given by
$(1+\exp(1-(\sum_{j'\in \mathcal{S}_k^*} Z_{ij'})/2))^{-1}.$
Under this setting, around 33\% of the entries of the data matrix for predictors are missing. 

Finally, we generate the parameters in the measurement model with only dichotomous items. 
We sample  $a_j$'s from a uniform distribution $U[0.5, 1.5]$, and $b_j$'s from uniform distribution $U[-2, 0]$, where the range of these distributions is chosen to guarantee that $a_j\theta_i + b_j$ to be in a suitable range. 
When generating the responses, a matrix sampling design is adopted. Here, all the items are divided into three equal-sized blocks. Each observation is randomly assigned one of the three blocks, and the responses to the rest of the two blocks are missing completely at random. 


We apply the three-step procedure described in Section~\ref{subsec:robust}, including both the baseline procedure based on Algorithm~\ref{alg:pfer2} and the derandomized procedure based on Algorithm~\ref{alg:derandom2}. {In this simulation study, we assume that all item parameters are known and fix them to their true values in both \eqref{eq:loglik} and \eqref{eq:likelihood}. In addition, we include an $l_2$-penalty on $\bbb$ in equation \eqref{eq:loglik}, as well as an $l_2$-penalty on $(\bbb^\top, \rrr^\top)^\top$ in equation \eqref{eq:likelihood} during the estimation to mitigate the problem of overfitting. More details are given in the supplementary material.}
Different target levels are considered, including $\nu \in \{1, 2, ..., 5\}$.  
Our results are given in Table~\ref{tab:sim}. Two performance metrics are reported, including (1) the average PFER, which is calculated by averaging $|\hat{\mathcal S} \setminus {\mathcal{S}}^*|$ over 100 replications, and (2) the average True Positive Rate (TPR), which is calculated by averaging $|\hat{\mathcal S} \cap {\mathcal{S}}^*|/|{\mathcal{S}}^*|$. As we can see, the baseline algorithm controls the PFER around the nominal level, while the derandomized knockoff method tends to be more conservative, which gives an average PFER much smaller than the nominal level. On the other hand, the derandomized method tends to be more powerful than the baseline algorithm in the sense that it typically achieves a higher average TPR. This phenomenon is consistent with the findings in \cite{ren2021derandomizing} under linear and logistic regression settings. 


\begin{table}
	\caption{Simulation results. Here, ``Baseline" refers to the baseline algorithm, Algorithm 2, and ``DRM" refers to derandomized knockoffs, Algorithm 3. $\nu$ refers to the nominal PFER level.}
	\label{tab:sim}
	\centering
	\adjustbox{max height=0.2\textheight}
	{
	{\begin{tabular}{cccccccc}			
		\toprule
		&&& $\nu=1$  &  $\nu=2$ & $\nu=3$ & $\nu=4$ & $\nu=5$ \\
		\midrule
		&
		&  Baseline & 0.68 & 1.72 & 2.78 & 4.01 & 5.04\\
		\cmidrule{3-8}
		&\multirow{-2}{*}{PFER} &  DRM & 0.01 & 0.10 & 0.33 & 0.55 & 0.83  \\			
		\cmidrule{2-8}
		& & Baseline & 54.0\% & 65.5\% & 70.9\% & 74.2\% & 77.2\%\\
		\cmidrule{3-8} 
		\multirow{-5}{*}{N = 1000}&\multirow{-2}{*}{TPR}& DRM & 59.9\% & 69.1\% & 74.6\% & 78.3\% & 80.7\%\\
		\midrule
		&
		 &  Baseline &  1.28  & 2.52  &  3.63 & 4.98  &  5.85 \\
		\cmidrule{3-8}
		&\multirow{-2}{*}{PFER}&  DRM &  0.06 & 0.33  & 0.69  & 1.15  &  1.60\\		
		\cmidrule{2-8}
		& & Baseline &  81.7\% & 88.9\% & 91.1\% &  93.5\% &  94.0\%\\
		\cmidrule{3-8} 
		\multirow{-5}{*}{N = 2000}&\multirow{-2}{*}{TPR}& DRM & 83.8\% & 90.3\% & 93.1\% &  95.2\% &  95.8\%\\
		\midrule
		&
		 &  Baseline &   0.76&    1.84&   3.06 &  4.17 & 5.41  \\
		\cmidrule{3-8}
		&\multirow{-2}{*}{PFER}&  DRM & 0.14  & 0.53   &  0.95 &  1.51  & 1.95  \\			
		\cmidrule{2-8}
		& & Baseline &  95.3\% & 98.9\% &  99.2\% &  99.4\% &    99.5\%\\
		\cmidrule{3-8} 
		\multirow{-5}{*}{N = 4000}& \multirow{-2}{*}{TPR}& DRM &   97.7\% & 99.4\% &  99.6\% &  99.6\% &  99.7\%\\
		\bottomrule
	\end{tabular}}
 }
\end{table}

\section{Application to PISA 2015}\label{sec:real}

We now apply the proposed method to the PISA 2015 dataset described in Section~\ref{sec:data}.  
Our results are given in Table~\ref{tab:real}. In this table, the predictors are ranked according to the value of $\Pi_j$ when $\nu=1$, from the largest to the smallest. For each predictor, we give the variable name, the variable type (continuous, binary, or ordinal), and a brief explanation of the variable. Further details about these variables are given in the supplementary material. In addition, we present the estimated coefficients of these variables under the full model (i.e., the model with all the predictors) and their standard errors based on a non-parametric bootstrap procedure with 200 replications. For each continuous variable, the standardized estimated coefficient is given, which is the estimated coefficient multiplied by the standard deviation of the corresponding variable.  
Variable selection results with nominal PFER levels $\nu = 1, 2, 3$ are given in Table~\ref{tab:real}, for which 
36, 45, and 48
predictors are selected, respectively. Note that by the construction of the derandomized knockoff method, these selection results are nested,  in the sense that the variables selected with $\nu = t$ are also selected with $\nu = t+1$, $t=1, 2, ...$. We also point out that for the first 
{20 variables (ANXTES to UNFAIR)}
, $\Pi_j = 1$, i.e., the variables are always selected by the baseline algorithm, and for the last
  11 variables (FISCED to WEALTH), 
$\Pi_j = 0$, for any $\nu = 1, 2, 3$, i.e., they are never selected by the baseline algorithm.

\begin{center}
\begin{longtable}{>{\raggedright\arraybackslash}p{\dimexpr(\textwidth-4\tabcolsep)/8}>{\centering}p{\dimexpr(\textwidth-4\tabcolsep)/16}>
{\raggedright\arraybackslash}p{\dimexpr(\textwidth-4\tabcolsep)/2}>{\raggedleft\arraybackslash}p{\dimexpr(\textwidth-4\tabcolsep)/10}>{\centering\arraybackslash}p{\dimexpr(\textwidth-4\tabcolsep)/10}}
 \caption{Results from applying Algorithm~3 to PISA data. The variables are ordered according to the value of $\Pi_j$ when $\nu = 1$, from the largest to the smallest. For variables with the same $\Pi_j$ values, they are ordered alphabetically. Continuous, binary, and ordinal variables are indicated by C, B, and O, respectively. For an ordinal variable $Z_j$, a coefficient corresponds to a dummy variable $\III(Z\geq k)$, for each non-baseline category $k=1, ..., K_j$.}
	\label{tab:real}\\ 

  \toprule
	Name & Type &Description & Estimate & SE\\
			\midrule
			\multicolumn{5}{c}{$\nu = 1$}\\
            \hline
            ANXTES	 & C &  Personality: test anxiety.  &  $-0.0542$ &  $0.0097$\\            
			\hline									
BELONG & C & Subjective well-being: sense of belonging to school. &	$-0.0454$ & $0.0108$\\
            	\hline
             DISCLI & C&  Disciplinary climate in science classes. &	$0.0657$ &  $0.0104$\\
			\hline	
   CPSVAL & C &  Collaboration and teamwork dispositions: value cooperation.  &  $-0.0862$ & $0.0110$\\
			\hline
			EBSCIT	& C &  Enquiry-based science teaching and learning practices. &  $-0.0561$ &  $0.0117$\\ 
            \hline	
            EISCED & O &  ISCED (International Standard Classification of Education) level student expects to complete. (0/1/2 = [level 2 or 3A]/[level 4 or 5B]/[level 5A or 6]) &	{$0.1733$  $0.0615$} & {$0.0350$  $0.0299$} \\
		\hline	
            ENVAWA & C & Environmental awareness.  & $0.0640$ & $0.0109$ \\
            \hline
			ENVOPT & C & Environmental optimism.  &  $-0.0886$   &  $0.0090$ \\
			\hline
			EPIST & C & Epistemological beliefs.  &	 $0.0887$ & $ 0.0101$\\
			\hline			
			GENDER & B &  Student's gender. ($0/1 =$ female/male) & $0.1884$ & $0.0211$\\
			\hline
			JOYSCI	& C & Enjoyment of science. & $0.0889$ & $0.0124$\\
			\hline			
                 OUT.JOB	&B & Whether work for pay outside the school.(0/1 = no/yes)  & $0.2076$ &  $0.0290$\\
			\hline
                OUT.PAR & B& Whether talk to parents outside the school. (0/1 = no/yes)  & $-0.1373$ &  $0.0307$\\
            \hline
			OUT.SPO & B& Whether exercise or do a sport outside the school. (0/1 = no/yes) & $0.1966$ &  $0.0223$\\
			\hline
			OUT.STU &B & Whether study for school or homework outside the school. (0/1 = no/yes) & $0.1188$ & $0.0202$\\
			\hline
						
			PERFEE	 & C & Perceived feedback.  & $-0.1373  $ & $0.0125$\\
			\hline
			REPEAT  & B & Whether the student has ever repeated a grade. (0/1 = no/yes) & $-0.2391$ & $0.0350$ \\
                \hline
                SCI.CHE &B& Whether attended chemistry courses in this or last school year. (0/1 = no/yes) &	$0.1109$ & $0.0207$\\	
  		\hline
   			TMINS & C & Learning time in class per week (minutes). & $0.0949$ & $0.0098$\\
            \hline   
            UNFAIR & C & Teacher unfairness. &  $-0.0542$ & $0.0108$\\	
      \hline
   
   LANGAH	& B & Whether language at home different from the test language. (0/1 = no/yes) & $-0.1022 $ & $0.0280 $\\
            \hline	
    TDSCIT &C& Teacher-directed science instruction.  &  $0.0504$ & $0.0112$\\ 
    \hline	      
EISEIO & C &  Student's expected  International Socio-economic Index of occupational status. & $0.0448$ & $0.0109$\\        		
    \hline
    OUTHOU & C &  Out-of-school study time per week (hours). &	$-0.0448$ & $0.0106$ \\
    \hline
    COOPER & C &  Collaboration and teamwork dispositions: enjoy cooperation. &   $0.0537$  & $0.0113$\\
    \hline		
    INSTSC & C & Instrumental motivation. & $-0.0388$ &  $0.0097$\\	
    \hline            
    SCIEEF	& C & Science self-efficacy. & $0.0408   $ & $0.0105$\\
   \hline

   FISEIO & C & ISEI (International Socio-economic Index) of occupational status of father. & $0.0436$ & $ 0.0122$\\
    \hline
    MISEIO &C& ISEI (International Socio-economic Index) of occupational status of mother.  & $0.0373$ & $0.0113$\\
    \hline
   CULTPO & C & Cultural possessions at home.  & $0.0388$ & $0.0114$ \\
    \hline
   CHONUM &O& Whether can choose the number of school science course(s) they study. (0/1/2 = no, not at all/ yes, to a certain degree/yes, can choose freely) &	 {$0.0995$ $-0.0288$} &  {$0.0230$ $0.0361$} \\ 
   \hline
               
    SCI.PHY	&B&  Whether attended physics courses in this or last school year. (0/1 = no/yes) & $-0.0718$ & $ 0.0201$\\   
    \hline
    SKIDAY &O& The frequency student skipped a whole school day in the last two full weeks of school. (0/1/2 = [none]/[one or two times]/[three or more times]) &	 {$-0.0398$  $-0.1312$} &  {$0.0198$  $0.0452$} \\
    \hline
    CHODIF &O&  Whether can choose the level of difficulty for school science course(s). (0/1/2 = no, not at all/ yes, to a certain degree/yes, can choose freely) &	 {$0.0677$  $0.0459$} &  {$0.0227$ $0.0305$} \\
    \hline
    DAYPEC & O & Averaged days that student attends physical education classes each week. (0/1/2/3 = [0]/[1 or 2]/[3 or 4]/[5 or more]) &  {$-0.0561$ $0.0150$ $-0.0740$} &  {$0.0361$  $0.0409$  $0.0278$}\\
    \hline
    ADINST &C& Adaption of instruction.  & $0.0243$ &  $0.0140$\\ 
  \hline	
  
   \multicolumn{5}{c}{$\nu = 2$}\\
    \hline
    SCI.EAR &B& Whether attended earth and space courses in this or last school year. (0/1 = no/yes) & $-0.0549$ &$ 0.0221$ \\
			\hline
    OUT.NET  & B & Whether use Internet outside the school. (0/1 = no/yes)  & $0.0663$ & $0.0254$ \\
    			\hline
    ARRLAT &O& The frequency of arriving late for school in the last two full weeks of school. (0/1/2 = [none]/[one or two times]/[three or more times]) &
     {$-0.0731$  $-0.0118$} &  {$0.0211$  $0.0349$} \\		
    \hline
GRADE & O & Student's grade. (0/1/2 = lower than modal grade/not lower than modal grade/higher than modal grade.) &	 {$0.1125$ $-0.0090$} &  {$0.0368$ $0.0242$}\\
	\hline            
    {HEDRES} &C& Home educational resources.  &  $-0.0216$  & $0.0108$\\

                \hline

   INTBRS &C& Interest in broad science topics. & $0.0232$ & $0.0120$\\ 
			\hline
   {CHOCOU} &O&  Whether can choose the school science course(s) they study. (0/1/2 = no, not at all/yes, to a certain degree/yes, can choose freely)
			&	 {$0.0516$ $ -0.00612$} &  {$0.0223$ $0.0291$} \\
  \hline
   {OUT.VED} &B & Whether watch TV/DVD/Video outside the school. (0/1 = no/yes) & $0.0400$ & $0.0204$\\
			\hline
   {DUECEC} & O & Duration in early childhood education and care of students. (0/1/2/3 = [less than two years]/[at least two but less than three years]/[at least three but less than four years]/[at least four years])
	    	&	 {$0.0496$  $-0.0313$  $-0.0797$} &  {$0.0268$  $0.0313$  $0.0416$}\\			
	    	\hline
  \multicolumn{5}{c}{$\nu = 3$}\\
			\hline
   SCI.GEN &B& Whether attended general, integrated, or comprehensive science courses in this or last school year. (0/1 = no/yes) &	$0.0372$ & $0.0210$ \\			
			\hline
   EMOSUP &C& Parents' emotional support. &$-0.0180$ & $0.0114$\\
    \hline
    DAYMPA &O& Number of days with moderate physical activities for a total of at least 60 minutes per week. ($0/1/2/3/4/5/6/7 = 0/1/2/3/4/5/6/7$)  &  {$0.0070$ $0.0453$  $0.0457$  $-0.0218$  $-0.0048$  $0.0555$  $0.0006$}  &   {$0.0446$ $0.0457$  $0.0405$ $0.0370$  $0.0344$  $0.0363$  $0.0340$}\\
    \hline
  \multicolumn{5}{c}{Unselected}\\
   \hline
   OUT.MEA & B& Whether have meals before school or after school. (0/1 = no/yes)  & $0.0373$ &   $0.0195$\\
    \hline
    OUT.GAM  & B & Whether play video-games outside the school. (0/1 = no/yes)  & $0.0222$ & $0.0230$ \\
    \hline
TEASUP & C& Teacher support in science classes of students' choice. & $0.0112$ &  $0.0126$\\
			\hline

    FISCED & O & Father's education in ISCED level.  (0/1/2/3/4 = [none or ISCED 1]/[ISCED 2]/[ISCED 3B or 3C]/[ISCED 3A or 4]/[ISCED 5B]/[ISCED 5A or ISCED 6])
	    	&	 {$0.0057$  $0.0051$  $-0.0145$  $0.0445$} &  {$0.0433$ $ 0.0343$  $0.0304$  $ 0.0360$}\\

      \hline
    MISCED & O & Mother's education in ISCED level.  (0/1/2/3/4 = [none or ISCED 1]/[ISCED 2]/[ISCED 3B or 3C]/[ISCED 3A or 4]/[ISCED 5B]/[ISCED 5A or ISCED 6])
	    	&	 {$-0.0376$   $0.0314$   $-0.0418$   $0.0382$} &  {$0.0520$  $ 0.0355$   $0.0270$   $0.0280$}\\        	
			\hline
               MOTIVA & C & {Achievement motivation.}  & $0.0082$ & $0.0100$\\
			\hline
                OUT.FRI & B& Whether meet or talk to friends on the phone outside the school. (0/1 = no/yes)  & $0.0105$ &  $0.0216$\\
			\hline
            
               OUT.HOL &B & Whether work in the household outside the school. (0/1 = no/yes)  & $-0.0040$ &  $0.0222$\\
			\hline
             
               OUT.REA  &B & Whether read a book/newspaper/magazine outside the school. (0/1 = no/yes) & $0.0159$ & $0.0244$\\
            			\hline

			SCI.APP &B& Whether attended applied sciences and technology courses in this or last school year. (0/1 = no/yes) & $0.0024$ & $0.0293$ \\
			\hline
			SCI.BIO  &B& Whether attended biology courses in this or last school year. (0/1 = no/yes) & $-0.0148$ & $ 0.0248$\\
			\hline
			
			SCIACT  & C & Index science activities.  & $0.0060$ & $0.0114$ \\
			\hline				
			SKICAL &O& The frequency of skipping some classes in the last two full weeks of school. (0/1/2 = [none]/[one or two times]/[three or more times]) &	 {$-0.0192$  $-0.0090$} &  {$0.0204$  $0.0390$} \\
			\hline
			
            WEALTH &C& Family wealth.  & $ 0.0053$ & $0.0108$\\	\bottomrule

\end{longtable}
\end{center}

We comment on some of the variable selection results. Several variables in the data concern the socioeconomic status of students' families, including the parents' occupational statuses (FISEIO, MISEIO), cultural possessions at home (CULTPO; e.g., books), parents' education levels (FISCED, MISCED), home educational resources (HEDRES), and family wealth (WEALTH), where FISEIO, MISEIO, FISCED, and MISCED are ordinal variables, and HEDRES, CULTPO and WEALTH are continuous variables. These variables are positively correlated with each other {(correlations/polyserial correlations between 0.22 and 0.69)}. It is interesting that parents' occupational statuses, cultural possessions, and home educational resources seem to be important in explaining students' performance in science (statistically significant and selected when {$\nu \leq 2$}).
{Given the rest of the variables, it is found that the higher occupational status of the father/mother or the more cultural possessions is associated with better science performance. However, HEDRES has a negative coefficient, which seems to be counter-intuitive and is worth further investigation.
}
 On the other hand, parents' education levels and family wealth seem to be less important (statistically insignificant and not selected even when $\nu = 3$). These results may be interpreted by a hypothetical mediation model as shown in Figure~\ref{fig:path}, which remains to be validated using additional data and statistical path analysis. 
 That is, WEALTH naturally has direct effects on CULTPO and HEDRES, which may have direct effects on students' science achievement. 
 Moreover, FISCED and MISCED naturally have a direct effect on FISEIO and MISEIO, respectively, and also possibly have direct effects on 
CULTPO, HEDRES, and WEALTH. 
However, there may not be direct paths from 
 WEALTH, FISCED or MISCED 
to students' science achievement. Students' science achievement may be largely influenced by genetic factors (e.g., intelligence) and environmental factors 
(e.g., education resources inside and outside home).
It is possible that 
 FISEIO, MISEIO, CULTPO, HEDRES  
and the other variables in the current analysis have provided good proxies to these genetic and environmental factors. 
Given these variables, FISCED, MISCED, and WEALTH tend to be conditionally independent of students' science achievement. 

\begin{figure}
    \centering
    \adjustbox{max height = 0.3\textheight}{

\includegraphics{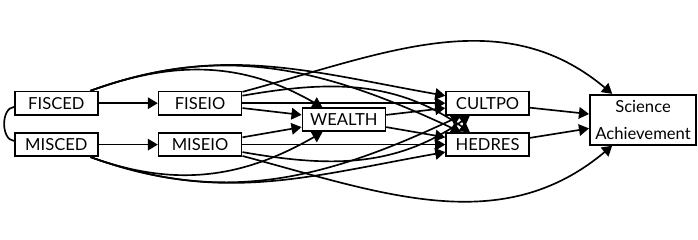}
}
    \caption{A hypothetical path diagram for several socioeconomic variables and science achievement.}
    \label{fig:path}
\end{figure}

Several variables consider students' behaviors attending school, including whether the student has ever repeated a grade (REPEAT), the frequency of a student skipping a whole school day in the last
two full weeks of school (SKIDAY), the frequency of the student arriving late for school in the last
two full weeks of school (ARRLAT), and the frequency of the student skipping some classes in the last
two full weeks of school (SKICLA), where REPEAT is a binary variable, and the other three are ordinal variables. 
These variables are positively correlated with each other {(tetrachoric/polychoric correlations between 0.07 and 0.53)}. The signs of the estimated coefficients are all consistent with our intuition. For instance, a student tended to perform worse on the test if they had ever repeated a grade or if they often arrived at school late. Among these variables, REPEAT, ARRLAT, and SKIDAY seem to be important variables in the sense that they are all selected with {$\nu \leq 2$}. On the other hand, given these variables as well as the rest of the variables, the variable SKICAL seems to be irrelevant (not selected even with $\nu = 3$, and the coefficients are not significant).

A few variables are related to teachers and their teaching style, including enquiry-based teaching and learning (EBSCIT), teacher-directed science instruction (TDSCIT), perceived feedback (PERFEE),  teacher unfairness (UNFAIR), adaptive instruction (ADINST), and teacher support in science classes of students' choice (TEASUP), all of which are continuous variables. 
Among these variables, ADINST, EBSCIT, TDSCIT, PERFEE, and UNFAIR are selected by our procedure with $\nu=1$,  while TEASUP is not selected. 
Variable UNFAIR has a negative coefficient, suggesting that teacher unfairness is associated with poor student performance after controlling for the other variables.
ADINST has a positive coefficient, suggesting that teachers’ flexibility with their lessons -- tailoring the lessons to the students in their classes -- tends to improve students' science performance.
In addition, it is interesting to see that TDSCIT has a positive coefficient while EBSCIT has a negative coefficient, which suggests that enquiry-based teaching and learning seem to have a negative effect on students' science achievement while teacher-directed instruction has a positive effect. It is possible that enquiry-based teaching and learning can broaden students' interests and increase their enjoyment of science {(correlation between EBSCIT and JOYSCI is 0.16 and that between EBSCIT and INTBRS is 0.13)}, but may be less efficient in developing students' science knowledge than teacher-directed instruction. Thus, a blended instruction model that combines the two teaching modes may be preferred. Finally, PERFEE has a negative coefficient, which may seem counter-intuitive at first glance, as providing informative and encouraging feedback is essential for improving student outcomes. This result may be due to the confounding of school types, which are not included in the current analysis. That is, students in disadvantaged schools may be more likely to report that their teachers provide them with feedback {\citep[Chapter 2,][]{pena2016pisa}}. These students also tended to perform worse on the test, which resulted in a negative coefficient estimate. 

Several variables concern students' attending of science courses in this or last school year, including chemistry (SCI.CHE), physics (SCI.PHY), earth and space (SCI.EAR), biology (SCI.BIO), general, integrated, or comprehensive science (SCI.GEN), and applied sciences and technology (SCI.APP). All these variables are binary. Among these variables, 
{SCI.CHE and SCI.PHY are selected with $\nu = 1$,  SCI.EAR is selected with $\nu = 2$, SCI.GEN is selected with $\nu = 3$, and the rest are not selected even with $\nu=3$. }
For the selected variables, SCI.CHE and SCI.GEN have positive coefficients, while SCI.PHY and SCI.EAR have negative coefficients. We suspect that these results may be due to the different curriculum settings at different types of schools, which are not included in the current model. Besides, there are also variables that measure students' opportunity to learn science at school. In particular, data are available on whether students can choose the number (CHONUM) and level of difficulty (CHODIF) of science courses, {and whether they can choose specific science courses (CHOCOU) at school.} It turns out that
 CHONUM and CHODIF are selected with $\nu = 1$, while CHOCOU is {selected with $\nu = 2$}. More specifically, the estimated coefficients for CHOCOU, CHONUM and CHODIF suggest that students with some freedom to choose the {subject}, number and level of difficulty of science courses tended to perform better in the test. 

Students' science achievement may also be related to their activities and received support outside of school. The current analysis includes variables on whether a student studies for school or homework (OUT.STU), talks to parents outside the school (OUT.PAR), works for pay (OUT.JOB), exercises or does sports (OUT.SPO), uses internet (OUT.NET), watches TV/DVD/Video (OUT.VED), plays video games (OUT.GAM), has meals (OUT.MEA), meets or talks to friends (OUT.FRI), works in the household (OUT.HOL), and reads a book/newspaper/magazine (OUT.REA) outside the school, and whether they receive emotional support from their parents (EMOSUP). 
All these variables are binary. Among these variables, 
OUT.STU, OUT.PAR, OUT.JOB, OUT.SPO, and  are selected with $\nu=1$, 
OUT.NET and 
OUT.VED are selected with $\nu=2$, EMOSUP is selected with $\nu=3$,  and the rest are not selected. 
Among the selected variables,
 variables OUT.STU, OUT.JOB, OUT.SPO, OUT.NET, and {OUT.VED} have positive coefficients, 
 suggesting that students with these outside-of-school activities also tended to perform better in the test after controlling for the rest of the variables. On the other hand, it is counter-intuitive that 
 OUT.PAR and EMOSUP have negative coefficients, though the coefficient for EMOSUP is not statistically insignificant.  
 This is worth future investigation.

Furthermore, the data contain variables that concern students' perceptions or attitudes towards science and related topics. They include the level of enjoying cooperation (COOPER), the level of valuing cooperation (CPSVAL), environmental awareness (ENVAWA), environmental optimism (ENVOPT),  epistemological beliefs about science (EPIST), enjoyment of science (JOYSCI), instrumental motivation (INSTSC), science self-efficacy (SCIEEF), and interest in broad science topics (INTBRS). All these variables are continuous. They are all selected. 
Specifically, INTBRS is selected with $\nu=2$, and the rest are selected with $\nu=1$. 
The correlation between CPSVAL and COOPER is 0.45. It is interesting that CPSVAL has a negative coefficient, suggesting that controlling for the other variables, students who more appreciate the value of cooperation and teamwork tended to perform worse in the test. In contrast, COOPER has a positive coefficient, implying that controlling for the other variables, students who enjoy cooperation and teamwork tended to perform better. Variables ENVAWA and ENVOPT have a correlation {$-0.14$}. Interestingly, ENVAWA has a positive coefficient, and ENVOPT has a negative coefficient. Moreover, INSTSC, which measures students' perception that studying science in school is useful to their future lives and careers, has a negative coefficient. It seems slightly counter-intuitive.  However, such a result is possible, given that variables like JOYSCI and INTBRS have been included in the regression model {(correlation between INSTSC and JOYSCI is 0.34 and correlation between INSTSC and INTBRS is 0.24)}. It may be explained by a mediation model in Figure~\ref{fig:instsc}, where INSTSC has positive direct effects on  JOYSCI and INTBRS, both of which further have positive effects on students' science achievement. However, given JOYSCI and INTBRS, the direct effect of INSTSC on science achievement is negative, possibly due to that INSTSC also brings pressure and stress to students when they learn science.

Finally, science achievement may also be related to other psychological factors. Specifically, the current analysis considers students' test anxiety level (ANXTES), sense of belonging to school (BELONG), expected education level (EISCED),
expected occupational status (EISEIO), and motivation to achieve (MOTIVA), all of which are continuous except that EISCED is ordinal.  
All these variables are selected with $\nu=1$, except for MOTIVA, which is not selected even when $\nu=3$.
For most of these variables, the signs of the estimated coefficients are consistent with our intuition. Specifically, ANXTES has a negative coefficient, suggesting a higher level of test anxiety is associated with poorer performance, controlling for the rest of the variables. EISCED and EISEIO have positive coefficients, which suggests that higher anticipation of the future is associated with high science achievement. However, it is less intuitive that BELONG has a negative coefficient, which may be due to not accounting for the school effect.

\begin{figure}
    \centering
    \adjustbox{max height = 0.3\textheight}{

\includegraphics{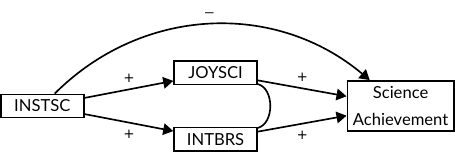}
}
    \caption{A hypothetical path diagram for INSTSC, JOYSCI, INTBRS, and science achievement.}
    \label{fig:instsc}
\end{figure}

\section{Discussions}\label{sec:dis}

In this paper, we considered identifying non-cognitive predictors of students' academic performance based on complex data from ILSAs that involve many missing values, mixed data types, and measurement errors. This problem can naturally be formulated as a variable selection problem. However, existing statistical methods are not applicable due to the complex data structure. For instance, variable selection methods for linear regression do not solve the current problem due to that (1) the response variable -- students' academic achievement -- is not directly observable but measured by cognitive items, and (2) there are  
many missing values in the predictors. We addressed these challenges by proposing a new model which combines a latent regression model and a Gaussian copula model. 
Furthermore, we proposed a derandomized knockoff method under a general latent variable model framework which includes the proposed latent regression model as a special case. 
This method tackles the multiple comparison issues of variable selection by controlling the PFER, a familywise error rate for variable selection. {Theoretical properties of the proposed method were established.} We focused on an application to PISA 2015 data, with the response variable being students' proficiency in the science domain. This analysis involved 5,685 students, 184 science items, and 62 non-cognitive variables that are of mixed types and contain many missing values. To our best knowledge, this is the first variable selection study of ILSAs that involves a dataset as large as the current one. With PFER level set to be $\nu = 1, 2, 3$, the proposed procedure selected
36, 45, and 48 variables, respectively. The model selection results are sensible, and signs of the parameter estimates for most of the selected variables are consistent with our intuition. The variable selection and parameter estimation results were examined from the perspectives of family socioeconomic status, school attending behaviors, teacher-related factors, science course resources and choices at school, out-of-school activities, perception and attitude towards science and related topics, and other psychological factors. These results provided 
insights into non-cognitive factors that are likely associated with students' science achievement, which can be useful to educators, policymakers, and other stakeholders.

The current analysis has several limitations that will be addressed in future research. First, the current application only considers the US sample and the science domain in PISA. It is of interest to investigate how the result of model selection varies across countries and knowledge domains. In particular, we expect the selection results to be substantially different across different countries due to cultural and socio-economic differences. In addition, the non-null predictors for different knowledge domains may also differ, which can suggest tailored education strategies for different domains. Second, it is also of interest to extend the current analysis to other ILSAs, such as the TIMSS and PIRLS, to see how the results change with slightly different test designs and different student age groups. 

The proposed method may be very useful for the scaling and reporting of ILSAs. First, the variable selection results establish a pathway between students' achievement in each subject domain and its possible influencing factors/causes. These results provide evidence that assists educators, policymakers, and related stakeholders to make informed education decisions. Second, it may improve the scaling methodology of ILSAs. Currently, a latent regression model is used in most ILSAs to estimate the performance distributions of populations (e.g., countries). This model, which is similar to the latent regression model in the current study, borrows information from non-cognitive background variables to compensate for the shortage of cognitive information. 
However, unlike the current model,  the latent regression model adopted in ILSAs does not directly regress on the background variables. Instead, it first conducts a PCA step to
reduce the background variables' dimensionality and then incorporates the derived PCA scores as predictors in latent regression.  This approach is often criticized for lacking interpretability, as the principal components often lack substantive meanings. Instead of performing PCA, 
we recommend reducing the dimensionality of the background variables by variable selection and then 
fitting the latent regression model with the selected predictors. With the theoretical guarantee of our variable selection method and by reporting the selected variables, the estimation and reporting of performance distributions become more transparent and interpretable.

 While we focus on an application to ILSAs, the proposed method also receives many other applications. For example, the method can also be used to identify neural determinants of visual short-term memory and to identify demographic correlates of psycho-pathological traits \citep{jacobucci2019practical}. Moreover, the proposed Gaussian copula model can be used with other regression models, such as linear and generalized linear regression models, for solving estimation and variable selection problems involving massive missing data and mixed types of variables. It is thus widely applicable to real-world problems involving missing data, which are commonly encountered in the social sciences, such as social surveys, marketing, and public health.

From the methodological perspective, there are several directions worth future development and investigation. First, 
the current model fails to account for possible multilevel structures in the data; for example, students are nested within schools. From the analysis of PISA data,  the signs of some estimated coefficients are not consistent with our intuition, which is likely due to not accounting for the school effect. Therefore, we believe that it is important to extend the current model by introducing random effects to model multilevel structures. New computation methods need to be developed accordingly. Second, the current analysis requires a relatively strong condition on data missingness, which is weaker than MCAR but stronger than MAR.
In social science, data may often be missing not at random. In that case, 
one may simultaneously model the complete data distribution and the missing data mechanism  \citep[e.g.,][]{kuha2018latent}.  
Such a joint model can be incorporated into the current analysis framework for generating knockoffs and further controlling variable selection errors. 
Third, the knockoff method may be coupled with the multiple imputation method for missing data analysis, as the knockoff variables can naturally be viewed as missing data. Thus, one may extend the state-of-the-art multiple imputation methods \citep{liu2014stationary,van2018flexible} to simultaneously impute missing data and knockoff copies and then use the imputed data for solving the variable selection problem. 
Finally, this paper focuses on the PFER as the performance metric for variable selection. Other performance metrics may be explored, such as false discovery rate and $k$ family-wise error rate, which may be more sensible in other applications. Making use of recent developments on knockoff methods \citep{ren2021derandomizing,ren2022derandomized}, we believe that it is not difficult to extend the current method to these error metrics.  


\appendix
\section*{Appendix}
\section{Proof of Theorems}
\begin{proof}[Proof of Theorem \ref{thm:robust}]
First of all, we claim that $\SSS^*$ is exactly the index set of non-null variables under the full rank assumption on $\SSigma$. In other words, $j \notin \SSS^*$ if and only if $Z_j$ is conditionally independent of $\YY$ given $\{Z_k\}_{k\neq j}$. This observation is summarized in the following Lemma, whose proof can be found in Appendix~\ref{app:lemproof}. 
\begin{lemma}\label{lem:null}
    Suppose the underlying correlation matrix $\SSigma$ has full rank, then $j \notin \SSS^*$ if and only if $Z_j$ is conditionally independent of $\YY$ given $\{Z_k\}_{k\neq j}$.
\end{lemma}

Next, for any index $j\notin\SSS^*$, we can establish the following result in a similar manner to \cite{barber2020robust}. The proof of this Lemma is also given in Appendix \ref{app:lemproof}.
\begin{lemma}\label{lem:Wcontrol}
    For any $j\notin\SSS^*$ and $\epsilon \geq 0$,
    \begin{equation}\label{eq:KLcontrol}
    \begin{aligned}
       &P\left({T_j} > 0, \hat{\text{KL}}_j \leq \epsilon \big\vert {|T_j|,\TT_{-j}},  |T_j|, \TT_{-j}, {\{\AAA_i\}_{i=1}^N, \{\BBB_i\}_{i=1}^N}\right) \\
        &\quad\quad\quad\quad\quad
        \leq e^{\epsilon}\cdot P\left({T_j}< 0 \big\vert {|T_j|,\TT_{-j}}, |T_j|, \TT_{-j}, {\{\AAA_i\}_{i=1}^N, \{\BBB_i\}_{i=1}^N}\right).
    \end{aligned}
\end{equation}
where {$\TT_{-j} = \{T_k\}_{k\neq j}$}.
\end{lemma}

For the given  value of $\nu$, we can  express the threshold $\tau$ as a function of {$\TT$} as follows:
\[{\tau = \tau(\TT)} = \inf \left\{t>0: 1 + |\left\{j: {T_j} < -t \right\}|=v\right\},\]
with $\inf \varnothing = -\infty.$ Note that by right continuity with respect to $t$, we have
\[1 + |\left\{j: {T_j} < -\tau \right\}|=v.\]
With this notation, we further define $${\tau_j = \tau(T_1, \ldots, T_{j-1}, |T_j|, T_{j+1}, \ldots, T_p)}$$ 
for each $j = 1, 2, ..., p$. Some useful properties about the sequence $\{\tau_j\}_{j=1}^p$ are summarised in the following lemma:
\begin{lemma}\label{lem:Tjk}
$\{\tau_j\}_{j=1}^p$ satisfies the following properties:
\begin{itemize}
    \item[]{(1)} For each $j=1,..,p$,  $\tau = \tau_j$ if $T_j > \tau$.
    \item[]{(2)} For any $j\neq k$, if $\max\{T_j, T_k\} < -\min\{{\tau_j, \tau_k}\}$, then ${\tau_j = \tau_k}.$ Consequently, if $T_j < -\tau_j$, then 
    \[T_k < -\tau_j \text{ if and only if } T_k < -\tau_k.\]
\end{itemize}    
\end{lemma}
\noindent See Appendix \ref{app:lemproof} for the proof of this lemma.

Now consider the number of false discoveries $|\{j: j\in \hat{\SSS}\cap\SSS^* \text{ and } \hat{\text{KL}}_j\leq \epsilon\}|$. We have
\begin{align}
    |\{j: j\in \hat{\SSS}\backslash\SSS^* \text{ and } \hat{\text{KL}}_j\leq \epsilon\}| &= \sum\limits_{j\notin\SSS^*} \III({T_j} > \tau, \hat{\text{KL}}_j\leq \epsilon) \\
    &= \left[1 + \sum\limits_{j\notin\SSS^*} \III({T_j} <  -\tau )\right] \cdot \frac{ \sum_{j\notin\SSS^*}\III({T_j}> \tau, \hat{\text{KL}}_j\leq \epsilon)}{1 + \sum_{j\notin\SSS^*} \III({T_j} < -\tau)} \\
    &\leq \left[1 + \sum\limits_{j=1}^p \III({T_j} < -\tau)\right] \cdot \frac{ \sum_{j\notin\SSS^*}\III({T_j}> \tau, \hat{\text{KL}}_j\leq \epsilon)}{1 + \sum_{j\notin\SSS^*} \III({T_j} < -\tau)} \\
    &\leq \nu\cdot R_{\epsilon}(\tau),
\end{align}
where we define $$R_{\epsilon}(\tau) = \frac{\sum_{j\notin\SSS^*}\III({T_j}> \tau, \hat{\text{KL}}_j\leq \epsilon)}{1 + \sum_{j\notin\SSS^*} \III({T_j} < -\tau)}.$$
Therefore, it suffices to prove that 
\[\EE[R_{\epsilon}(\tau) {\vert \{\AAA_i\}_{i=1}^N, \{\BBB_i\}_{i=1}^N}] \leq e^\epsilon.\]

We begin by observing that
\[\begin{aligned}
    \EE[R_{\epsilon}(\tau){\vert \{\AAA_i\}_{i=1}^N, \{\BBB_i\}_{i=1}^N}] &= \EE\left[\frac{\sum_{j\notin\SSS^*}\III({T_j}> \tau, \hat{\text{KL}}_j\leq \epsilon)}{1 + \sum_{j\notin\SSS^*} \III({T_j} < -\tau)} {\bigg\vert \{\AAA_i\}_{i=1}^N, \{\BBB_i\}_{i=1}^N}\right] \\
    &= \sum\limits_{j\notin\SSS^*} \EE\left[\frac{\III({T_j} > \tau, \hat{\text{KL}}_j\leq \epsilon)}{1 + \sum_{k\notin\SSS^*, k\neq j} \III({T_k} < -\tau)}{\bigg\vert \{\AAA_i\}_{i=1}^N, \{\BBB_i\}_{i=1}^N}\right].
\end{aligned}\]
By property (1) in Lemma~\ref{lem:Tjk}, we have 
\begin{equation}
    \begin{aligned}
    &\sum\limits_{j\notin\SSS^*} \EE\left[\frac{\III({T_j} > \tau, \hat{\text{KL}}_j\leq \epsilon)}{1 + \sum_{k\notin\SSS^*, k\neq j} \III({T_k} < -\tau)} {\bigg\vert \{\AAA_i\}_{i=1}^N, \{\BBB_i\}_{i=1}^N}\right] \\
    &= \sum\limits_{j\notin\SSS^*} \EE\left[\frac{\III({T_j > \tau_j}, \hat{\text{KL}}_j\leq \epsilon)}{1 + \sum_{k\notin\SSS^*, k\neq j} \III({T_k < -\tau_j})}{\bigg\vert \{\AAA_i\}_{i=1}^N, \{\BBB_i\}_{i=1}^N}\right] \\
      &= \sum\limits_{j\notin\SSS^*} \EE\left[\frac{\III({T_j} > 0, \hat{\text{KL}}_j\leq \epsilon)\cdot\III({|{T_j}|>\tau_j})}{1 + \sum_{k\notin\SSS^*, k\neq j} \III({T_k < -\tau_j})}{\bigg\vert \{\AAA_i\}_{i=1}^N, \{\BBB_i\}_{i=1}^N}\right]\\
 &= \sum\limits_{j\notin\SSS^*} \EE\left[\frac{\PP\left({T_j} > 0, \hat{\text{KL}}_j\leq \epsilon \big\vert {|T_j|, \TT_{-j}, \{\AAA_i\}_{i=1}^N, \{\BBB_i\}_{i=1}^N}\right)\cdot\III({|T_j|>\tau_j})}{1 + \sum_{k\notin\SSS^*, k\neq j} \III({T_k < -\tau_j}
 )}{\bigg\vert \{\AAA_i\}_{i=1}^N, \{\BBB_i\}_{i=1}^N}\right].
    \end{aligned}
\end{equation}
The last equation holds by taking conditional expectation conditioned on $\left(|T_j|, \TT_{-j}, \{\AAA_i\}_{i=1}^N, \{\BBB_i\}_{i=1}^N\right)$ inside the outer expectation, while noticing that $\tau_j \text{ is a function of } |T_j| \text{ and } \TT_{-j}$.
Then, by equation \eqref{eq:KLcontrol} we have
\begin{equation}
    \begin{aligned}
    &\sum\limits_{j\notin\SSS^*} \EE\left[\frac{\PP\left({T_j} > 0, \hat{\text{KL}}_j\leq \epsilon \big\vert {|T_j|, \TT_{-j}, \{\AAA_i\}_{i=1}^N, \{\BBB_i\}_{i=1}^N}\right)\cdot\III({|T_j|>\tau_j})}{1 + \sum_{k\notin\SSS^*, k\neq j} \III({T_k < -\tau_j}
 )}{\bigg\vert \{\AAA_i\}_{i=1}^N, \{\BBB_i\}_{i=1}^N}\right]\\ 
 &\leq \sum\limits_{j\notin\SSS^*} \EE\left[e^\epsilon\frac{\PP\left({T_j} < 0\big\vert {|T_j|, \TT_{-j}, \{\AAA_i\}_{i=1}^N, \{\BBB_i\}_{i=1}^N}\right)\cdot\III({|T_j|>\tau_j})}{1 + \sum_{k\notin\SSS^*, k\neq j} \III({|T_j|<-\tau_j})}{\bigg\vert \{\AAA_i\}_{i=1}^N, \{\BBB_i\}_{i=1}^N}\right]\\
&= e^\epsilon\cdot\sum\limits_{j\notin\SSS^*} \EE\left[\frac{\III({T_j < -\tau_j})}{1 + \sum_{k\notin\SSS^*, k\neq j} \III({T_k < -\tau_j})}{\bigg\vert \{\AAA_i\}_{i=1}^N, \{\BBB_i\}_{i=1}^N}\right].
    \end{aligned}
\end{equation}
Finally, by property (2) in Lemma~\ref{lem:Tjk}, we have
\begin{equation}
\begin{aligned}
&\sum\limits_{j\notin\SSS^*}\EE\left[\frac{\III({T_j < -\tau_j})}{1 + \sum_{k\notin\SSS^*, k\neq j} \III({T_k < -\tau_j})}{\bigg\vert \{\AAA_i\}_{i=1}^N, \{\BBB_i\}_{i=1}^N}\right] \\&= \sum\limits_{j\notin\SSS^*}\EE\left[\frac{\III({T_j < -\tau_j})}{1 + \sum_{k\notin\SSS^*, k\neq j} \III({T_k < -\tau_k})}{\bigg\vert \{\AAA_i\}_{i=1}^N, \{\BBB_i\}_{i=1}^N}\right]  \\
 &= \sum\limits_{j\notin\SSS^*}\EE\left[\frac{\III({T_j < -\tau_j})}{ \sum_{j\notin\SSS^*} \III({T_k < -\tau_k})}{\bigg\vert \{\AAA_i\}_{i=1}^N, \{\BBB_i\}_{i=1}^N}\right]\\
 &= 1.
\end{aligned}
\end{equation}
This finishes the proof of $\EE[R_{\epsilon}(\tau){\vert \{\AAA_i\}_{i=1}^N, \{\BBB_i\}_{i=1}^N}] \leq e^\epsilon$.

{ For the case when $\PP = \QQ$, we need to show that $\PP_{ij}( \xx_{ij}\vert \xxo_{i,-j}, \yyo_i)\cdot \PP_{ij}(\xxko_{i,-j}, \xxk_{ij}\vert \xxo_{i,-j}, \xx_{ij}, \yyo_i) $ remains the same after switching the roles of $\xx_{ij}$ and $\xxk_{ij}$. It is equivalent to proving that 
\begin{equation}\label{eq:disswap}
    (\XX_{ij}, \XXk_{ij}, \XXo_{i, -j}, \XXko_{i, -j}, \YYo_i) \stackrel{d}{=} (\XXk_{ij},  \XX_{ij}, \XXo_{i, -j}, \XXko_{i, -j}, \YYo_i)
\end{equation}
for $j\notin \SSS^*$ and $i\in\{i:j\in(\mathcal{M}\cup\AAA_i)\}$ under the true model, where $\XX$, $\XXk$, $\XXo$ and $\XXko$ were defined in Section 4.

 By Proposition \ref{prop:knockZmiss}, for each $i = 1, ..., N$, there exists $\ZZ_i^*$ and $\ZZk_i^*$ satisfying Definition \ref{def:knockoff}. 
Note that for any $j\notin \SSS^*$ and $i\in\{i:j\in(\mathcal{M}\cup\AAA_i)\}$, we have
\begin{align}
    p(Z^*_{ij}, \tilde{Z}^*_{ij}, \ZZ^*_{i, -j}, \ZZk^*_{i, -j}, \YYo_i) 
    &= p(Z^*_{ij}, \tilde{Z}^*_{ij}, \ZZ^*_{i, -j}, \ZZk^*_{i, -j}) p(\YYo_i\vert\ZZ^*_i, \ZZk_i^*)\\
    &= p(Z^*_{ij}, \tilde{Z}^*_{ij}, \ZZ^*_{i, -j}, \ZZk^*_{i, -j}) p(\YYo_i\vert\ZZ^*_i)\\
    &= p(Z^*_{ij}, \tilde{Z}^*_{ij}, \ZZ^*_{i, -j}, \ZZk^*_{i, -j}) p(\YYo_i\vert\ZZ_{i,-j})\\
    &= p((Z^*_{ij}, \tilde{Z}^*_{ij})_{\text{swap}},  \ZZ^*_{i, -j}, \ZZk^*_{i, -j}) p(\YYo_i\vert\ZZ_{i,-j}),
\end{align}
where $(A, B)_{\text{swap}}$ means swapping the values of two random variables $A$ and $B$.
Here the first equality holds since $\ZZk^*_i$ and $\YYo_i$ are conditionally independent given $\ZZ^*_i$, the second equality holds since $\YYo_i\vert\ZZ_i$ follows the true latent regression IRT model as well as $j\notin\SSS^*$, and the third equality holds by letting $\SSS = \{j\}$ in the last property of Definition \ref{def:knockoff}.
This implies that
\begin{equation}\label{eq:swapZ}
    (Z^*_{ij}, \tilde{Z}^*_{ij}, \ZZ^*_{i, -j}, \ZZk^*_{i, -j}, \YYo_i) \stackrel{d}{=} (\tilde{Z}^*_{ij}, Z^*_{ij}, \ZZ^*_{i, -j}, \ZZk^*_{i, -j}, \YYo_i).
\end{equation}
Therefore,
\[\begin{aligned}
    &p((\XX_{ij}, \XXk_{ij})_{\text{swap}}, \XXo_{i, -j}, \XXko_{i, -j},\YYo_i) \\
    &= \int p((\XX_{ij}, \XXk_{ij})_{\text{swap}}, \XXo_{i, -j}, \XXko_{i, -j}\vert Z^*_{ij}, \tilde{Z}^*_{ij}, \ZZ^*_{i, -j}, \ZZk^*_{i, -j},\YYo_i) \cdot p(Z^*_{ij}, \tilde{Z}^*_{ij}, \ZZ^*_{i, -j}, \ZZk^*_{i, -j},\YYo_i) d\ZZ^*d\ZZk^*\\
    &= \int p((\XX_{ij}, \XXk_{ij})_{\text{swap}}, \XXo_{i, -j}, \XXko_{i, -j}\vert Z^*_{ij}, \tilde{Z}^*_{ij}, \ZZ^*_{i, -j}, \ZZk^*_{i, -j}) \cdot p(Z^*_{ij}, \tilde{Z}^*_{ij}, \ZZ^*_{i, -j}, \ZZk^*_{i, -j},\YYo_i) d\ZZ^*d\ZZk^*\\
    &=\int p((\XX_{ij}, \XXk_{ij})_{\text{swap}}\vert Z^*_{ij}, \tilde{Z}^*_{ij})\cdot p( \XXo_{i, -j}, \XXko_{i, -j}\vert\ZZ^*_{i, -j}, \ZZk^*_{i, -j})\cdot p(Z^*_{ij}, \tilde{Z}^*_{ij}, \ZZ^*_{i, -j}, \ZZk^*_{i, -j},\YYo_i)  d\ZZ^*d\ZZk^*\\
\end{aligned}
\]
Here the second equality holds since $\XX, \XXk$ are conditionally independent of $\YYo_i$ given $\ZZ^*$ and $\ZZk^*$, which follows from the second property stated in Definition \ref{def:knockoff}, and the last equality holds since all $\XX_{ij}$'s and $\XXk_{ij}$'s are  conditionally independent given $\ZZ^*$ and $\ZZk^*$. Note that by Definition \ref{def:knockoff}, $\XX_{ij}\vert Z_{ij}^*$ and $\XXk_{ij}\vert \tilde Z_{ij}^*$ follow the same conditional distribution, and thus $p((\XX_{ij}, \XXk_{ij})_{\text{swap}}\vert Z^*_{ij}, \tilde{Z}^*_{ij}) = p(\XX_{ij}, \XXk_{ij}\vert (Z^*_{ij}, \tilde{Z}^*_{ij})_{\text{swap}})$. As a result,
\[\begin{aligned}
    &p((\XX_{ij}, \XXk_{ij})_{\text{swap}}, \XXo_{i, -j}, \XXko_{i, -j},\YYo_i) \\
    &= \int p(\XX_{ij}, \XXk_{ij}\vert (Z^*_{ij}, \tilde{Z}^*_{ij})_{\text{swap}})\cdot p( \XXo_{i, -j}, \XXko_{i, -j}\vert\ZZ^*_{i, -j}, \ZZk^*_{i, -j})\cdot p(Z^*_{ij}, \tilde{Z}^*_{ij}, \ZZ^*_{i, -j}, \ZZk^*_{i, -j},\YYo_i)  d\ZZ^*d\ZZk^* \\
    &= \int p(\XX_{ij}, \XXk_{ij}\vert (Z^*_{ij}, \tilde{Z}^*_{ij})_{\text{swap}})\cdot p( \XXo_{i, -j}, \XXko_{i, -j}\vert\ZZ^*_{i, -j}, \ZZk^*_{i, -j})\cdot p((Z^*_{ij}, \tilde{Z}^*_{ij})_{\text{swap}}, \ZZ^*_{i, -j}, \ZZk^*_{i, -j},\YYo_i)  d\ZZ^*d\ZZk^* \\
    &= \int p(\XX_{ij}, \XXk_{ij}\vert Z^*_{ij}, \tilde{Z}^*_{ij})\cdot p( \XXo_{i, -j}, \XXko_{i, -j}\vert\ZZ^*_{i, -j}, \ZZk^*_{i, -j})\cdot p(Z^*_{ij}, \tilde{Z}^*_{ij}, \ZZ^*_{i, -j}, \ZZk^*_{i, -j},\YYo_i)  d\ZZ^*d\ZZk^* \\
    &=p((\XX_{ij}, \XXk_{ij})_{\text{swap}}, \XXo_{i, -j}, \XXko_{i, -j},\YYo_i),
\end{aligned}\]
where we made use of \eqref{eq:swapZ} in the second equation. This finishes the proof of \eqref{eq:disswap}.}

\end{proof}

\begin{proof}[Proof of Theorem \ref{thm:derandpfer}]
By definition, we have
\begin{align}
    &\EE|\{j: j\in \hat{\mathcal S}\setminus \mathcal S^* \mbox{~and~} \hat{\text{KL}}_j^{(m)}  \leq \epsilon,  m =1, ..., M\}| \\
    &= \EE \left[\sum\limits_{j\in \hat{\mathcal S}\setminus \mathcal S^*} \III\left(\Pi_j \geq \eta \text{~and~} \max_{1\leq m\leq M} \hat{\text{KL}}_{j}^{(m)}\leq\epsilon\right)\right]\\
    &= \sum\limits_{j\in \hat{\mathcal S}\setminus \mathcal S^*} \PP \left(\Pi_j \geq \eta \text{~and~} \max_{1\leq m\leq M} \hat{\text{KL}}_{j}^{(m)}\leq\epsilon\right) \\
    &= \sum\limits_{j\in \hat{\mathcal S}\setminus \mathcal S^*} \PP \left(\frac{1}{M} \sum\limits_{m=1}^{M}\III\left(j\in\SSS^{(m)}\right)\geq \eta \text{~and~} \max_{1\leq m\leq M}\hat{\text{KL}}_{j}^{(m)}\leq\epsilon\right).
\end{align}
Note that within the event $\{\max_{1\leq m\leq M}\hat{\text{KL}}_{j}^{(m)}\leq\epsilon\}$, $$\III\left(j\in\SSS^{(m)}\right) = \III\left(j\in\SSS^{(m)} \text{~and~} \hat{\text{KL}}_{j}^{(m)}\leq\epsilon\right).$$
Thus
\begin{align}
    &\EE|\{j: j\in \hat{\mathcal S}\setminus \mathcal S^* \mbox{~and~} \hat{\text{KL}}_j^{(m)}  \leq \epsilon,  m =1, ..., M\}| \\
    &= \sum\limits_{j\in \hat{\mathcal S}\setminus \mathcal S^*} \PP \left(
    \frac{1}{M} \sum\limits_{m=1}^{M}\III\left(j\in\SSS^{(m)} \text{~and~} \hat{\text{KL}}_{j}^{(m)}\leq\epsilon\right) \geq \eta 
    \text{~and~} \max_{1\leq m\leq M}\hat{\text{KL}}_{j}^{(m)}\leq\epsilon\right)\\
    &\leq \sum\limits_{j\in \hat{\mathcal S}\setminus \mathcal S^*} \PP \left(
    \frac{1}{M} \sum\limits_{m=1}^{M}\III\left(j\in\SSS^{(m)} \text{~and~} \hat{\text{KL}}_{j}^{(m)}\leq\epsilon\right) \geq \eta\right) \\
    &= \sum\limits_{j\in \hat{\mathcal S}\setminus \mathcal S^*} \PP \left(
    \Pi_j^\dagger \geq \eta\right) 
\end{align}
By the assumption $\PP \left(
\Pi_j^\dagger \geq \eta\right) \leq \EE \left[\Pi_j^\dagger\right]$ and the result of Theorem \ref{thm:robust},
\begin{align}
    \sum\limits_{j\in \hat{\mathcal S}\setminus \mathcal S^*} \PP \left(
    \Pi_j^\dagger \geq \eta\right) &
    \leq \sum\limits_{j\in \hat{\mathcal S}\setminus \mathcal S^*} \EE \left[\Pi_j^\dagger \right] = \EE\left[ \sum\limits_{j\in \hat{\mathcal S}\setminus \mathcal S^*}\III\left(j\in\SSS^{(1)} \text{~and~} \hat{\text{KL}}_{j}^{(1)}\leq\epsilon\right)\right]\leq \nu e^\epsilon, 
\end{align}
which proved the first part of Theorem \ref{thm:derandpfer}.
The proof of the second part of Theorem \ref{thm:derandpfer} is totally the same as that of Proposition 2 in \cite{ren2021derandomizing}, and we omit it here.
\end{proof}

\section{Proof of Propositions}
\begin{proof}[Proof of Proposition \ref{prop:knockZmiss}] 


Let $\ZZ^{*, \text{Alg1}}_i$ be the underlying variables obtained from Step~1 of Algorithm \ref{alg:knockcons2}.  
Furthermore, let $\ZZalg_i$ be the complete data obtained by imputing $\ZZ^{*, \text{Alg1}}_i$ using the transformation functions within the true Gaussian copula model, i.e., $\Zalg_{ij} = F_j(Z_{ij}^{* , \text{Alg1}}),~j=1,\ldots, p$. For simplicity, we still use $\WWk_i$, $\ZZk^*_i$, $\ZZk_i$, and $\ZZko_i$ to denote the knockoff copies obtained from Algorithm \ref{alg:knockcons2}. We will verify that $(\WWk_i, \ZZko_i)$ satisfies all conditions in Definition~\ref{def:knockoff}. 

The proof for the first condition is straightforward. 
To verify the second condition, we first notice that under the SMAR condition,
\begin{equation}\label{eq:condunder}
    \begin{aligned}
        p(\ZZ^*_i, \WW_i, \ZZo_i, \YYo_i,  \AAA_i) &= p(\ZZ^*_i, \WW_i, \ZZo_i, \YYo_i) p(\AAA_i \vert \ZZ^*_i, \WW_i, \ZZo_i, \YYo_i) \\
        &= p(\ZZ^*_i, \WW_i, \ZZo_i, \YYo_i) p(\AAA_i \vert \WW_i, \ZZo_i),\\
    \end{aligned}
\end{equation}
where $\ZZ^*_i$ represents the true underlying variables in the data generating process. This implies that the missing mechanism is ignorable. Since $\ZZ^{*, \text{Alg1}}_i$ is sampled according to $p(\ZZ^*_i\vert\WW_i, \ZZo_i, \YYo_i)$, the joint distribution of 
$\ZZ^{*, \text{Alg1}}_i$ , $\WW_i, \ZZo_i,$ and $\YYo_i$ is the same as that of $\ZZ^{*}_i$ , $\WW_i, \ZZo_i, $ and $\YYo_i$. Therefore, we have
$$p(\YYo_i, \WW_i\vert\ZZ^{*, \text{Alg1}}_i) = p(\YYo_i\vert\ZZ^{*, \text{Alg1}}_i) p(\WW_i\vert\ZZ^{*, \text{Alg1}}_i).$$
Moreover, it is obvious that $\ZZk_i^*$ is conditionally independent of $(\YYo_i, \WW_i)$ given $\ZZ_i^*$ based on its construction. This concludes the proof for the second condition.

Based on the aforementioned arguments, as the 
joint distribution of $\ZZ^{*, \text{Alg1}}_i$, $\WW_i$, $\ZZo_i$, and $\YYo_i$
is identical to the true distribution, it is easy to see that third and forth conditions holds. Finally, the last condition is also satisfied,  as the covariance matrix of $(\ZZ_i^*, \ZZk_i^*)$ remains unchanged when simultaneously swapping the $j$-th row with the $j+p$-th row and the $j$-th column with the $j+p$-th column for any $j\in\SSS$. 

\end{proof}

\begin{proof}[Proof of Proposition \ref{prop:knockffstat}]
Let $(\hat{\bbb}, \hat{\rrr})$ be the maximum likelihood estimates obtained by solving \eqref{eq:latregest}. For any subset $\SSS \subset \{1, \ldots, p\},$ 
\begin{equation*}
    	{t_{j}}\left(\left\{(\WW, \ZZo), (\WWk, \ZZko)\right\}_{\mathrm{swap}(\mathcal S)}, \YYo\right)= 
    	\begin{cases}
    	\operatorname{sign}(\|\hat \bbb_j^\dagger\|  - \|\hat \rrr_j^\dagger\|)\max\left\{\|\hat \bbb_j^\dagger\|/\sqrt{p_j}, \|\hat \rrr_j^\dagger\|/\sqrt{p_j}\right\} = {T_j}, & j \notin \mathcal S, 
    	\\ \operatorname{sign}(\|\hat \rrr_j^\dagger\| - \|\hat \bbb_j^\dagger\|  )\max\left\{\|\hat \bbb_j^\dagger\|/\sqrt{p_j}, \|\hat \rrr_j^\dagger\|/\sqrt{p_j}\right\} = {-T_j}, & j \in \mathcal S,
    	\end{cases}
\end{equation*}
Hence {$T_j$} given by \eqref{eq:knockoff} satisfies Definition~\ref{def:feature}.
\end{proof}

\begin{proof}[Proof of Proposition \ref{prop:pfer0}]
This proposition is implied by  Theorem \ref{thm:robust} in the case of $\PP = \QQ$.
\end{proof}

\begin{proof}[Proof of Proposition \ref{prop:pfer}]
This proposition is implied by  Theorem \ref{thm:derandpfer} in the case of $\PP = \QQ$.
\end{proof}

\section{Proof of Lemmas}\label{app:lemproof}
\begin{proof}[Proof of Lemma \ref{lem:null}]
    For simplicity, we omit the subscript $i$ in this proof. Due to the linear structure of the latent regression model \eqref{eq:likelihood}, it is easy to show that when any variable in vector $(g_1(Z_{1}), \ldots, g_p(Z_p))^\top$ cannot be perfectly predicted by the others, $\SSS^*$ is the index set of non-null variables. Therefore, we need only to prove that if the underlying correlation matrix $\SSigma$ has full rank, then it is impossible to make a perfect prediction of any variable in the vector $(g_1(Z_{1}), \ldots, g_p(Z_p))^\top$ using the others.
    Suppose for the sake of contradiction that some variable in $g_j(Z_{j})$ can be perfectly predicted by the other variables. 
\begin{itemize}
    \item If $Z_j$ is continuous, note that $Z_j^* = (Z_j - c_j)/d_j = (g_j(Z_j) - c_j)/d_j$ and $\{g_l(Z_{l})\}_{l\neq j}$ are functions of $\{Z_l\}_{l\neq j}$, which implies that $Z_j^*$ can be predicted by $\{Z_l^*\}_{l\neq j}$. This is impossible under the full rank assumption. 
    \item If $Z_j$ is binary, then similarly we know that $\III(Z_j^* > c_j)$ can be predicted by $\{Z_l^*\}_{l\neq j}$. This leads to a contradiction of the full rank assumption, since $Z_j^*$ given $\{Z_l^*\}_{l\neq j}$ follows a non-degenerated normal distribution with support on the whole real line.
    \item Finally, if $Z_j$ is ordinal with $K_j + 1$ categories, then there exit $k \in \{1, \ldots, K_j\}$ such that $\III(Z_j^* > c_{jk})$ can be perfectly predicted by $\{\III(Z_j^* > c_{jk^\prime}): k^\prime \neq k\}$ along with $\{Z_l^*\}_{l\neq j}$.     
    Consider the event
    $$\mathcal{E} = \left\{ Z_j^* \in (c_{j,k-1}, c_{j, k+1})\right\} = \left\{\III(Z_j^* > c_{jk^\prime}) = 0 \text{ for all } k^\prime > k \text{ and } \III(Z_k^* > c_{jk^\prime}) = 1 \text{ for all } k^\prime < k\right\}.
    $$
     According to the assumption of contraction, $Z_j^*$ follows a distribution with support either on $(c_{j,k-1}, c_{j, k})$ or $(c_{j,k}, c_{j, k+1})$, conditioned on event $\mathcal{E}$ and $\{Z_l^*\}_{l\neq j}$. However,
    due to the full rank assumption, $Z_j^*$ should follow a truncated normal distribution with support $(c_{j,k-1}, c_{j, k+1})$, conditioned on event $\mathcal{E}$ along with $\{Z_l^*\}_{l\neq j}$. This leads to a contradiction.
\end{itemize}
To conclude, we have proved our claim by by using an argument of contradiction.
\end{proof}

\begin{proof}[Proof of Lemma \ref{lem:Wcontrol}]
Hereafter in this section, we will denote $\mathcal{M} = \{1, ..., p_1\}.$ Recall the definition of $ \XXo_{i}$ and $\XXko_{i}$ in Subsection~\ref{subsec:robust}.
For each $i=1, ..., N$ and $j = 1, ... ,p$, we also introduce
$$\begin{aligned}        
        \XXo_{i, -j}&= \{\XX_{ik}: k\in(\mathcal{M}\cup\AAA_i)\backslash\{j\}\}, \\
    \XXo_{j} &= \{\XX_{ij}:  j\in (\mathcal{M}\cup\AAA_i)\},\\ 
    \XXo_{-j}&= \{\XX_{ik}: k\in(\mathcal{M}\cup\AAA_i)\backslash\{j\}\}, \\
    \XXko_{i, -j} &= \{\XXk_{ik}: k\in(\mathcal{M}\cup\AAA_i)\backslash\{j\}\},\\
\XXko_{j} &= \{\XXk_{ij}: j\in (\mathcal{M}\cup\AAA_i)\}, \\ \XXko_{-j} &= \{\XXk_{ik}: k\in(\mathcal{M}\cup\AAA_i)\backslash\{j\}
\}.
\end{aligned}$$ 


From now on we fix {$j\notin \SSS^*$}. Label the unordered pair of vectors 
$\{\XXo_{j}, \XXko_{j}\}$ 
as $\{\XX_{j}^{(0)}, \XX_{j}^{(1)}\}$, such that 
\begin{equation}
    \begin{cases}
    \text{If } \XXo_{j} = \XX_{j}^{(0)} \text{ and } \XXko_{j} = \XX_{j}^{(1)}, \text{then } T_j \geq 0, \\
    \text{If } \XXo_{j} = \XX_{j}^{(1)} \text{ and } \XXko_{j} = \XX_{j}^{(0)}, \text{then } T_j \leq 0.
    \end{cases}
\end{equation}
It follows from the flip-sign property in Definition~\ref{def:feature} that the statistic  {$|T_j|$} is a function of $\XX_{j}^{(0)}, \XX_{j}^{(1)}, \XXo_{-j}, \XXko_{-j},$ and $\YYo$. It is worth noting that the sign of $T_j$ cannot be determined if we only have the information of $\{\XX_{j}^{(0)}, \XX_{j}^{(1)}\}$.
It is also easy to see that for any $k\neq j$, $T_k$ is also a function of $\XX_{j}^{(0)}, \XX_{j}^{(1)}, \XXo_{-j}, \XXko_{-j},$ and $\YYo$.

Ignoring the trivial case $|T_j| = 0$, we have
\begin{align*}
&\frac{\PP(T_j > 0 \vert {\XX_{j}^{(0)}, \XX_{j}^{(1)}, \XXo_{-j}, \XXko_{-j}, \YYo}, { \{\AAA_i\}_{i=1}^N, \{\BBB_i\}_{i=1}^N})}{\PP(T_j < 0 \vert {\XX_{j}^{(0)}, \XX_{j}^{(1)}, \XXo_{-j}, \XXko_{-j}, \YYo, { \{\AAA_i\}_{i=1}^N, \{\BBB_i\}_{i=1}^N}} )} 
\\ 
=& \frac{\PP((\XXo_j, \XXko_j) = (\XX_{j}^{(0)}, \XX_{j}^{(1)}) \vert {\XX_{j}^{(0)}, \XX_{j}^{(1)}, \XXo_{-j}, \XXko_{-j}, \YYo, { \{\AAA_i\}_{i=1}^N, \{\BBB_i\}_{i=1}^N}} )}{\PP((\XXo_j, \XXko_j) = (\XX_{j}^{(1)}, \XX_{j}^{(0)}) \vert {\XX_{j}^{(0)}, \XX_{j}^{(1)}, \XXo_{-j}, \XXko_{-j}, \YYo, { \{\AAA_i\}_{i=1}^N, \{\BBB_i\}_{i=1}^N}} )}
\\
=& {\prod\limits_{{i:j\in (\mathcal{M}\cup\AAA_i})}}\frac{\PP((\XX_{ij}, \tilde{\XX}_{ij}) = (\XX_{ij}^{(0)}, \XX_{ij}^{(1)}) \vert {\XX_{ij}^{(0)}, \XX_{ij}^{(1)}, \XXo_{i,-j}, \XXko_{i,-j}, {\YYo_i}, {\AAA_i, \BBB_i}})}{\PP((\XX_{ij}, \XXk_{ij}) = (\XX_{ij}^{(1)}, \XX_{ij}^{(0)}) \vert {\XX_{ij}^{(0)}, \XX_{ij}^{(1)}, \XXo_{i,-j}, \XXko_{i,-j}, {\YYo_i}, {\AAA_i, \BBB_i}})}. 
\end{align*}
Note that in equation~\eqref{eq:Wprop}, $\{{i:j\in (\mathcal{M}\cup\AAA_i})\} = \{i: i = 1, ..., N\}$ if $j \in \mathcal{M}$, and $\{i:j\in (\mathcal{M}\cup\AAA_i)\} = \{i:j\in\AAA_i\}$ if $j \notin \mathcal{M}$.
For each $i$, given that $j \notin \SSS^*$, by the SMAR assumption on $\AAA_i$ and MCAR assumption on $\BBB_i$, we have
{
\[\begin{aligned}
&p\left({\AAA_i, \BBB_i} \vert (\XX_{ij}, \tilde{\XX}_{ij}) = (\XX_{ij}^{(0)}, \XX_{ij}^{(1)}), {\XX_{ij}^{(0)}, \XX_{ij}^{(1)}, \XXo_{i,-j}, \XXko_{i,-j}, {\YYo_i}}\right)\\ 
    &= p({\AAA_i} \vert \XXo_{i,-j}) p({\BBB_i} \vert \YYo_i)\\
    &= p\left({\AAA_i, \BBB_i} \vert {\XX_{ij}^{(0)}, \XX_{ij}^{(1)}, \XXo_{i,-j}, \XXko_{i,-j}, {\YYo_i}}\right).\\
\end{aligned}\]
}
Therefore,
\[\begin{aligned}
    &\PP\left((\XX_{ij}, \tilde{\XX}_{ij}) = (\XX_{ij}^{(0)}, \XX_{ij}^{(1)}) \vert {\XX_{ij}^{(0)}, \XX_{ij}^{(1)}, \XXo_{i,-j}, \XXko_{i,-j}, {\YYo_i}, {\AAA_i, \BBB_i}}\right)\\
    &= \PP\left((\XX_{ij}, \tilde{\XX}_{ij}) = (\XX_{ij}^{(0)}, \XX_{ij}^{(1)}) \vert {\XX_{ij}^{(0)}, \XX_{ij}^{(1)}, \XXo_{i,-j}, \XXko_{i,-j}, {\YYo_i}}\right) \\
    &\quad\cdot \PP\left({\AAA_i, \BBB_i} \vert (\XX_{ij}, \tilde{\XX}_{ij}) = (\XX_{ij}^{(0)}, \XX_{ij}^{(1)}), {\XX_{ij}^{(0)}, \XX_{ij}^{(1)}, \XXo_{i,-j}, \XXko_{i,-j}, {\YYo_i}}\right)\\
    &\quad\cdot
    {\PP\left({\AAA_i, \BBB_i} \vert {\XX_{ij}^{(0)}, \XX_{ij}^{(1)}, \XXo_{i,-j}, \XXko_{i,-j}, {\YYo_i}}\right)^{-1}}\\
    &= \PP\left((\XX_{ij}, \tilde{\XX}_{ij}) = (\XX_{ij}^{(0)}, \XX_{ij}^{(1)}) \vert {\XX_{ij}^{(0)}, \XX_{ij}^{(1)}, \XXo_{i,-j}, \XXko_{i,-j}, {\YYo_i}}\right).
\end{aligned}\]
Similarly, 
\[\begin{aligned}
    &\PP\left((\XX_{ij}, \tilde{\XX}_{ij}) = (\XX_{ij}^{(1)}, \XX_{ij}^{(0)}) \vert {\XX_{ij}^{(0)}, \XX_{ij}^{(1)}, \XXo_{i,-j}, \XXko_{i,-j}, {\YYo_i}, {\AAA_i, \BBB_i}}\right) \\
    &= \PP\left((\XX_{ij}, \tilde{\XX}_{ij}) = (\XX_{ij}^{(0)}, \XX_{ij}^{(1)}) \vert {\XX_{ij}^{(0)}, \XX_{ij}^{(1)}, \XXo_{i,-j}, \XXko_{i,-j}, {\YYo_i}}\right),
\end{aligned}\]
and we can conclude that 
\begin{equation}\label{eq:Wprop}
    \frac{\PP(T_j > 0 \vert {\XX_{j}^{(0)}, \XX_{j}^{(1)}, \XXo_{-j}, \XXko_{-j}, \YYo}, { \{\AAA_i\}_{i=1}^N, \{\BBB_i\}_{i=1}^N})}{\PP(T_j < 0 \vert {\XX_{j}^{(0)}, \XX_{j}^{(1)}, \XXo_{-j}, \XXko_{-j}, \YYo, { \{\AAA_i\}_{i=1}^N, \{\BBB_i\}_{i=1}^N}} )} = \frac{\PP(T_j > 0 \vert {\XX_{j}^{(0)}, \XX_{j}^{(1)}, \XXo_{-j}, \XXko_{-j}, \YYo})}{\PP(T_j < 0 \vert {\XX_{j}^{(0)}, \XX_{j}^{(1)}, \XXo_{-j}, \XXko_{-j}, \YYo} )}
\end{equation}

Let $\xx_{ij}^{(0)}, \xx_{ij}^{(1)}, \xxo_{i,-j}, \xxko_{i,-j},$ and $ \yyo_i$ be the realisations of $\XX_{ij}^{(0)}, \XX_{ij}^{(1)}, \XXo_{i,-j}, \XXko_{i,-j},$ and $ \YYo_i$, respectively, for each $i \in \{i: j \in (\mathcal{M}\cup\AAA_i)\}$.
{Observe that the conditional probability function of $\XX_{ij}, \XXk_{ij}, \XXko_{i,-j}$ given $\XXo_{i,-j}$ and $\YYo_i$ can be decomposed as
\begin{align}
    p(\XX_{ij}, \XXk_{ij}, \XXko_{i,-j}\vert \XXo_{i,-j},  \YYo_i) = 
    \PP_{ij}(\XX_{ij}\vert \XXo_{i,-j},  \YYo_i) \cdot \QQ_{ij}(\XXko_{i,-j}, \XXk_{ij} \vert \XXo_{i,-j}, \XX_{ij}, \YYo_i).\label{eq:swap}
\end{align}}
Therefore, we have
\begin{align}
    &{\prod\limits_{{i:j\in (\mathcal{M}\cup\AAA_i})}}\frac{\PP((\XX_{ij}, \XXk_{ij}) = (\xx_{ij}^{(0)}, \xx_{ij}^{(1)}) \vert {\XX_{ij}^{(0)} = \xx_{ij}^{(0)}, \XX_{ij}^{(1)} = \xx_{ij}^{(1)}, \XXo_{i,-j} = \xxo_{i,-j}, \XXko_{i,-j} = \xxko_{i,-j}, \YYo_i = \yyo_i})}{\PP((\XX_{ij}, \XXk_{ij}) = (\xx_{ij}^{(1)}, \xx_{ij}^{(0)}) \vert {\XX_{ij}^{(0)} = \xx_{ij}^{(0)}, \XX_{ij}^{(1)} = \xx_{ij}^{(1)}, \XXo_{i,-j} = \xxo_{i,-j}, \XXko_{i,-j} = \xxko_{i,-j}, \YYo_i = \yyo_i})}\\
    =&{\prod\limits_{{i:j\in (\mathcal{M}\cup\AAA_i})}}\frac{p(\XX_{ij}= \xx_{ij}^{(0)}, \XXk_{ij} =\xx_{ij}^{(1)}, \XXko_{i,-j} = \xxko_{i,-j} \vert {\XXo_{i,-j} = \xxo_{i,-j}, \YYo_i = \yyo_i})}{p(\XX_{ij}= \xx_{ij}^{(1)}, \XXk_{ij} =\xx_{ij}^{(0)}, \XXko_{i,-j} = \xxko_{i,-j}  \vert {\XXo_{i,-j} = \xxo_{i,-j}, \YYo_i = \yyo_i})}\\
    =& {\prod\limits_{{i:j\in (\mathcal{M}\cup\AAA_i})}}\frac{\PP_{ij}( \xx_{ij}^{(0)}\vert \xxo_{i,-j}, \yyo_i)\cdot \QQ_{ij}(\xxko_{i,-j}, \xx_{ij}^{(1)}\vert \xxko_{i,-j}, \xx_{ij}^{(0)}, \yyo_i)}{\PP_{ij}( \xx_{ij}^{(1)}\vert \xxo_{i,-j}, \yyo_i)\cdot \QQ_{ij}(\xxko_{i,-j}, \xx_{ij}^{(0)}\vert \xxko_{i,-j}, \xx_{ij}^{(1)}, \yyo_i)}. \label{eq:KL}
\end{align}


We now define $\rho_j$ as a function of $\XX_{j}^{(0)}, \XX_{j}^{(1)}, \XXo_{-j}, \XXko_{-j},$ and $\YYo$ such that
\begin{equation}\label{eq:rhodef}
    \begin{aligned}
    \rho_j = \rho_j(\XX_{j}^{(0)}, \XX_{j}^{(1)}, \XXo_{-j}, \XXko_{-j}, \YYo) = \log \frac{\PP(T_j > 0 \vert {\XX_{j}^{(0)}, \XX_{j}^{(1)}, \XXo_{-j}, \XXko_{-j}, \YYo })}{\PP(T_j < 0 \vert {\XX_{j}^{(0)}, \XX_{j}^{(1)}, \XXo_{-j}, \XXko_{-j}, \YYo} )} 
\end{aligned}.
\end{equation}
By \eqref{eq:Wprop}, $\rho_j$ can also be expressed as
\begin{equation}\label{eq:rhoeq}
    \rho_j = \log\left[{\prod\limits_{{i:j\in (\mathcal{M}\cup\AAA_i})}} \frac{\PP((\XX_{ij}, \XXk_{ij}) = (\XX_{ij}^{(0)}, \XX_{ij}^{(1)}) \vert {\XX_{ij}^{(0)}, \XX_{ij}^{(1)}, \XXo_{i,-j}, \XXko_{i,-j}, \YYo_i})}{\PP((\XX_{ij}, \XXk_{ij}) = (\XX_{ij}^{(1)}, \XX_{ij}^{(0)}) \vert {\XX_{ij}^{(0)}, \XX_{ij}^{(1)}, \XXo_{i,-j}, \XXko_{i,-j}, \YYo_i})}\right].
\end{equation}
According to the definition of $\XX_j^{(0)}$ and $\XX_j^{(1)}$, if $T_j > 0$ then we have $(\XXo_j, \XXko_j) = (\XX_j^{(0)}, \XX_j^{(1)})$. Combining this with \eqref{eq:KL} and the definition of $\hat{\text{KL}}_j$, we get
\begin{equation}\label{eq:rho}
   \rho_j = \rho_j(\XXo_j, \XXko_j, \XXo_{-j}, \XXko_{-j}, \YYo) = \hat{\text{KL}}_j  \text{ within the event } \{T_j > 0\}.
\end{equation} 
Putting things together, we have
\begin{align}
    \PP&\left(T_j > 0, \hat{\text{KL}}_j \leq \epsilon \big\vert |T_j|,\TT_{-j}{ \{\AAA_i\}_{i=1}^N, \{\BBB_i\}_{i=1}^N}\right) \\
    =& \EE\left[\PP\left(T_j > 0, \hat{\text{KL}}_j \leq \epsilon \big\vert {\XX_{j}^{(0)}, \XX_{j}^{(1)}, \XXo_{-j}, \XXko_{-j}, \YYo}, { \{\AAA_i\}_{i=1}^N, \{\BBB_i\}_{i=1}^N}\right)  \big\vert |T_j|,\TT_{-j}, { \{\AAA_i\}_{i=1}^N, \{\BBB_i\}_{i=1}^N} \right] \\
    =& \EE\left[\PP\left(T_j > 0, \rho_j \leq \epsilon \big\vert {\XX_{j}^{(0)}, \XX_{j}^{(1)}, \XXo_{-j}, \XXko_{-j}, \YYo}, { \{\AAA_i\}_{i=1}^N, \{\BBB_i\}_{i=1}^N}\right)  \big\vert |T_j|,\TT_{-j}, { \{\AAA_i\}_{i=1}^N, \{\BBB_i\}_{i=1}^N} \right] \\    
    =& \EE\left[\PP\left(T_j > 0 \big\vert {\XX_{j}^{(0)}, \XX_{j}^{(1)}, \XXo_{-j}, \XXko_{-j}, \YYo}, { \{\AAA_i\}_{i=1}^N, \{\BBB_i\}_{i=1}^N}\right)  {\III(\rho_j\leq\epsilon)} \big\vert |T_j|,\TT_{-j}, { \{\AAA_i\}_{i=1}^N, \{\BBB_i\}_{i=1}^N} \right]\\
    =&  \EE\left[\PP\left(T_j < 0 \big\vert {\XX_{j}^{(0)}, \XX_{j}^{(1)}, \XXo_{-j}, \XXko_{-j}, \YYo}, { \{\AAA_i\}_{i=1}^N, \{\BBB_i\}_{i=1}^N}\right) e^{\rho_j}\III(\rho_j \leq\epsilon) \big\vert |T_j|,\TT_{-j}, { \{\AAA_i\}_{i=1}^N, \{\BBB_i\}_{i=1}^N} \right].
\end{align}
Here, the first equation holds 
since $|T_j|$ and $\TT_{-j} = \{T_k\}_{k\neq j}$ are functions of $\XX_{j}^{(0)}, \XX_{j}^{(1)}, \XXo_{-j}, \XXko_{-j},$ and $\YYo$,
the second equation holds by \eqref{eq:rho}, the third equation holds since $\rho_j$ is a function of $\XX_{j}^{(0)}, \XX_{j}^{(1)}, \XXo_{-j}, \XXko_{-j},$ and $\YYo$, and the final equation holds by \eqref{eq:Wprop} and \eqref{eq:rhodef}.
Note that since $e^{\rho_j}\III(\rho_j \leq\epsilon)\leq e^\epsilon$, \begin{align}
    &\EE\left[\PP\left(T_j < 0 \big\vert {\XX_{j}^{(0)}, \XX_{j}^{(1)}, \XXo_{-j}, \XXko_{-j}, \YYo}\right) e^{\rho_j}\III(\rho_j \leq\epsilon) \big\vert |T_j|,\TT_{-j} \right] \\
    \leq& e^\epsilon \EE\left[\PP\left(T_j < 0 \big\vert {\XX_{j}^{(0)}, \XX_{j}^{(1)}, \XXo_{-j}, \XXko_{-j}, \YY}\right)  \big\vert |T_j|,\TT_{-j} \right]\\
    =& e^\epsilon \PP\left(T_j < 0 \big\vert |T_j|,\TT_{-j}\right),
\end{align}
which finishes the proof of \eqref{eq:KLcontrol}.
\end{proof}

\begin{proof}[Proof of Lemma \ref{lem:Tjk}]
The first result is obvious since $T_j = |T_j|$ if $T_j > \tau \geq 0$. To prove the second result, without loss of generality we assume 
$\tau_j \leq \tau_k$, then we have $$T_j < -\tau_j \text{ and } T_k < -\tau_j,$$ and consequently
\[|T_j| > \tau_j \text{ and } |T_k| > \tau_j.\]
Note that we have
\[\begin{aligned}
    1 + \sum_{l\neq k}\III(T_l < -\tau_j) + \III(|T_k| < -\tau_j) &= 1 + \sum_{l\neq k}\III(T_l < -\tau_j) \\
    &= 1 + \sum_{l\neq j, k}\III(T_l < -\tau_j) + \III(T_j < -\tau_j) \\
    &= 1 + \sum_{l\neq j, k}\III(T_l < -\tau_j) + \III(T_k < -\tau_j) \\
    &= 1 + \sum_{l\neq j}\III(T_l < -\tau_j) \\
    &= 1 + \sum_{l\neq j}\III(T_l < -\tau_j) + \III(|T_j| < -\tau_j) \\
    &= v, \\
\end{aligned}\]
where the last equality holds by definition of $\tau_j$. Hence, by definition of $\tau_k$, it  holds  that $\tau_k \leq \tau_j$,
which finish the proof for $\tau_k = \tau_j$. The proof for the corollary is straightforward.
\end{proof}

\section{Computation}\label{app:algorithm}

\subsection{Algorithm for solving optimisation problem \texorpdfstring{\eqref{eq:mmlcopula}}{1}}
\subsubsection{Algorithm}

Recall the constrained optimisation problem \eqref{eq:mmlcopula}
\begin{equation}
\begin{aligned}
(\hat\Xi, \hat\Lambda_1, ..., \hat \Lambda_{p_1}) = \argmax_{\Xi, \Lambda_1, ..., \Lambda_{p_1}}\quad & l_1(\Xi, \Lambda_1, ..., \Lambda_{p_1})\\
 \text{ subject to }\quad&\Sigma_{jj} = 1,~j = 1, \ldots, p,\\
  & {d_j > 0,~j \notin \DD,}\\&c_{j1} < c_{j2} < \ldots < c_{jK_j},~j\in \DD. 
\end{aligned} 
\end{equation}

\subsubsection{Reparameterisation of the problem}
To better deal with the constraints in the optimisation problem, we first reparametrise part of the parameters in the model. For the correlation matrix $\SSigma$, we reparameterize it as $\SSigma = \LL\LL^\top$, where $\LL = (l_{ab})_{p\times p}$ is the Cholesky decomposition of $\SSigma$. Since $\LL$ is a lower triangular matrix, we have $l_{ab} = 0, ~a = j+1, \ldots, p$. By constraining $l_{aa}\neq 0$  and $\sum_{a=1}^b l_{ab}^2 = 1$, we can also guarantee that $\LL\LL^\top$ is positive definite and has unit diagonals.  For the threshold parameters of discrete variable $Z_j$,  we introduce $\uu_j = (u_{j1}, u_{j2}, \ldots, u_{jK_j})^\top (K_j \geq 1)$, such that 
\begin{align*}
	u_{jk} =
	\begin{cases}
		c_{j1},& \text{ if }k=1,\\
		\log(c_{jk} - c_{j, k-1}),& \text{ if }2\leq k \leq K_j.\\		
	\end{cases}
\end{align*}
Note that there are no constraints on  $\uu_j$ anymore. 

{We now define $\tilde\Xi = \{\LL\}\cup\{c_j, d_j, j\notin\DD\}\cup\{\uu_j,~j\in\DD\}$ as the set of parameters in the Gaussian copula model after reparametrisation, and set 
$\Upsilon = \tilde\Xi \cup\{\Lambda_1, ..., \Lambda_{p_1}\}$ as the set of all unknown parameters.
 Then we can rewrite $f_i(\ww_i, \zzo_i\vert \Xi, \Lambda_1, ..., \Lambda_{p_1})$ as
\begin{equation}\label{eq:redensity}
    f_i(\ww_i, \zzo_i\vert\Upsilon) = \int \cdots\int  f(\zz_i\vert \tilde\Xi) \left(\prod_{j=1}^{p_1} q_j(\mathbf w_{ij}\vert {z_{ij}}; \Lambda_j)\right) \left(\prod_{j\notin \AAA_i}d z_{ij}\right),
\end{equation}
where
\begin{equation}\label{eq:zsdensity}
    f(\zz_i\vert \tilde\Xi)= \int\ldots\int  \left(\prod_{j\in \DD } dz_{j}^{*}\right) \left[\phi(\zz^*\vert \LL\LL^\top) \times \left(\prod_{{j\notin\DD}} d_j^{-1}\right) \times  \left( \mathop{\prod}_{j\in \DD }\III (z_{j}^* \in  (c_{j, z_{j}-1},c_{j, z_{j}} ])\right) \right]\Bigg\vert_{z_{j}^* = \frac{z_{j} - c_j}{d_j}, j\notin\DD}.
\end{equation}
Hereafter, we will use the notations $c_{jk} (k = 0, 1, \ldots, K_j, K_j + 1)$ for $j\in \DD$ with slight abuse, to denote functions of $\uu_j$ such that
\begin{equation}
	c_{jk}(\uu_j) = 
	\begin{cases}
		-\infty, & \text{ if } k = 0.\\
		u_{j1}, & \text{ if } k = 1.\\ 
		u_{j1} + \sum\limits_{k=2}^{K_j}\exp(u_{jl}) & \text{ if } 2 \leq k \leq K_j.\\ 
		\infty, &  \text{ if } k = K_j+1.
	\end{cases}
\end{equation}  

The marginal likelihood function can then be written as
\[l_1(\Upsilon)= \sum_{i=1}^N \log f_i(\ww_i, \zzo_i\vert\Upsilon),\]
and constrained optimisation problem after reparametrisation \eqref{eq:mmlcopula} now becomes 
\begin{align}\label{eq:remmlcopula}
	\argmax_{\Upsilon}\quad& l_1(\Upsilon)\\
	\text{ subject to }\quad& l_{aa}\neq 0, a = 1, ..., p,\\
 & \sum_{a=1}^{b}l_{ab}^2 = 1, b = 1, ..., p,\\
 &l_{ab} = 0, 1\leq a < b \leq p,\\ 
 & d_j > 0,~\text{ if } j \notin \DD.\\
\end{align}
For convenience we also denote $\mathbf{\Omega}$ as the feasible set  in the following, where $\mathbf{\Omega} = \{\Upsilon: l_{aa}\neq 0, a = 1, ..., p; \sum_{a=1}^{b}l_{ab}^2 = 1, b = 1, ..., p; l_{ab} = 0, 1\leq a < b \leq p; d_j > 0,\text{ if } j \notin \DD\}$.}

\subsubsection{The stochastic gradient}
We propose to solve \eqref{eq:remmlcopula} by gradient-based method. However, the gradient of $l_1(\Upsilon)$ involves high-dimensional integrals, thus is difficult to compute exactly.  Fortunately, it can be approximated by method of sampling. For convenience, we define $\mathcal{M} = \{1, ..., p_1\}$, and 
$\CC = \{1, ..., {p}\}\backslash (\mathcal{M}\cup\DD)$. 

Hereafter, we will implicitly make use of the assumption that $Z_{ij} = Z_{ij}^*$, $c_j = 0$, and $d_j = 1$ for $j\in \mathcal{M}$.

\paragraph{Sampling distribution for $\ZZ^*$} 
We begin by rewriting $f_i(\ww_i, \zzo_i\vert\Upsilon)$ as 
\begin{equation}
\begin{aligned}    
     f_i(\ww_i, \zzo_i\vert\Upsilon) &= \int \cdots\int  \xi_i(\ww_i, \zz_i,  \zz_i^*\vert\Upsilon) \left(\prod_{j\notin (\mathcal{M}\cup\AAA_i)}d z_{ij}\right) \left(\prod_{j\in (\mathcal{M}\cup\DD)} dz_{ij}^{*}\right)\\
     &= \int \cdots\int \left[\int \cdots\int  \xi_i(\ww_i, \zz_i,  \zz_i^*\vert\Upsilon) \left(\prod_{j\notin (\mathcal{M}\cup\AAA_i)}d z_{ij}\right)\right] \left(\prod_{j\in (\mathcal{M}\cup\DD)} dz_{ij}^{*}\right)\\
\end{aligned}
\end{equation}
where
\begin{align*}
    &\xi_i(\ww_i, \zz_i,  \zz_i^*\vert\Upsilon)  = \\&\qquad\left(\phi(\zz_i^*\vert \LL\LL^\top) \times \prod_{j\in\CC} d_j^{-1} \times \prod_{j=1}^{p_1} q_j(\mathbf w_{ij}\vert {z_{ij}^*}; \Lambda_j) 
    \times \prod_{j\in\DD} \III (z_{ij}^* \in  (c_{j, z_{ij}-1},c_{j, z_{ij}} ])
    \right)\Bigg\vert_{z_{ij}^*= \frac{z_{ij} - c_j}{d_j}, j\in \mathcal{C}}.
\end{align*}
{Note that we have replaced $q_j(\mathbf w_{ij}\vert {z_{ij}}; \Lambda_j)$ by $q_j(\mathbf w_{ij}\vert {z_{ij}^*}; \Lambda_j)$ here.} It can be shown by change of variables that
\[\begin{aligned}
    &\int \cdots\int  \xi_i(\ww_i, \zz_i,  \zz_i^*\vert\Upsilon) \left(\prod_{j\notin \AAA_i}d z_{ij}\right) \\
    &= \left(\phi(\zz_i^*\vert \LL\LL^\top) \times \prod_{j\in\CC\cap\AAA_i} d_j^{-1} \times \prod_{j=1}^{p_1} q_j(\mathbf w_{ij}\vert {z_{ij}^*}; \Lambda_j) 
    \times \prod_{j\in\DD\cap\AAA_i} \III (z_{ij}^* \in  (c_{j, z_{ij}-1},c_{j, z_{ij}}])
    \right)\Bigg\vert_{z_{ij}^*= \frac{z_{ij} - c_j}{d_j}, j\in \CC\cap\AAA_i},
\end{aligned}\]
which depends only on $\ww_i, \zzo_i,$ and $\zz_i^*$. Hence we introduce the notation $\psi_i(\ww_i, \zzo_i, \zz_i^*\vert\Upsilon)$ to denote this integral, that is, let
\[\begin{aligned}
    &\psi_i(\ww_i, \zzo_i, \zz_i^*\vert\Upsilon) \\
    &= \int \cdots\int  \xi_i(\ww_i, \zz_i,  \zz_i^*\vert\Upsilon) \left(\prod_{j\notin \AAA_i}d z_{ij}\right)\\
    &= \left(\phi(\zz_i^*\vert \LL\LL^\top) \times \prod_{j\in\CC\cap\AAA_i} d_j^{-1} \times \prod_{j=1}^{p_1} q_j(\mathbf w_{ij}\vert {z_{ij}^*}; \Lambda_j) 
    \times \prod_{j\in\DD\cap\AAA_i} \III (z_{ij}^* \in  (c_{j, z_{ij}-1},c_{j, z_{ij}} ])
    \right)\Bigg\vert_{z_{ij}^*= \frac{z_{ij} - c_j}{d_j}, j\in \CC\cap\AAA_i}.
\end{aligned}\]

When $\ww_i$ and $\zzo_i$ are fixed, $\psi_i(\ww_i, \zzo_i, \zz_i^*\vert\Upsilon)$
can result in a density function for $\ZZ_i^*$ after normalisation. In the following, we use the notation $\mathcal{P}_i(\ww_i, \zzo_i\vert\Upsilon)$ to denote such a distribution for $\ZZ_i^*$.

\paragraph{Reformulation of derivatives}
Next, we will represent the derivatives of $l_1(\Upsilon)$ with respect to $\Upsilon$ using the prescribed sampling distribution for $\ZZ_i^*$.

For the derivatives with respect to $\LL$, we can use a change of measure to show that:
\[\frac{\partial \log f_i(\ww_i, \zzo_i\vert\Upsilon)}{\partial \LL} =\EE_{\ZZ_i^*\sim\mathcal{P}_i(\ww_i, \zzo_i\vert\Upsilon)}\left[\frac{\partial}{\partial \LL} \log \psi_i(\ww_i, \zzo_i, \ZZ_i^*\vert\Upsilon)\right]\]
for each $i = 1, .., N$. Here $\EE_{\ZZ_i^*\sim\mathcal{P}_i(\ww_i, \zzo_i\vert\Upsilon)}$ indicates that the expectation is taken with respect to $\ZZ_i^*$, which follows the distribution $\mathcal{P}(\ww_i, \zzo_i\vert\Upsilon)$. Note that if $z_{ij}^* \in (c_{j, z_{ij}-1},c_{j, z_{ij}}]$ for each $j\in\DD\cap\AAA_i$, we will have
\[\begin{aligned}
    \log\psi_i(\ww_i, \zzo_i, \zz_i^*\vert\Upsilon) =& -\frac{p}{2}\log(2\pi)-\left[\frac{1}{2} \zz_i^{*\top}(\LL\LL^\top)^{-1} \zz_i^*\right]\Bigg\vert_{z_{ij}^*= \frac{z_{ij} - c_j}{d_j},~j\in\CC\cap \AAA_i} -\sum_{j\in\CC\cap \AAA_i} \log d_j - \log\det(\LL)\\
    &+ \sum_{j=1}^{p_1}\log q_j(\ww_{ij}\vert z_{ij}^*; \Lambda_j).\\
\end{aligned}\]
Therefore, the derivatives $\frac{\partial}{\partial \LL} \log \psi_i(\ww_i, \zzo_i, \zz_i^*\vert\Upsilon)$ have simple expressions.

Similarly, for the derivatives with respect to  $c_j$ and $d_j$ when $j \in \CC$, and for the derivatives with respect to $\Lambda_j$ when $j \in \mathcal{M}$, we have
\[
\begin{aligned}
    &\frac{\partial \log f_i(\ww_i, \zzo_i\vert\Upsilon)}{\partial c_j} =\EE_{\ZZ_i^*\sim\mathcal{P}_i(\ww_i, \zzo_i\vert\Upsilon)}\left[\frac{\partial}{\partial c_{j}} \log\psi_i(\ww_i, \zzo_i, \ZZ_i^*\vert\Upsilon)\right],\\    
    &\frac{\partial \log f_i(\ww_i, \zzo_i\vert\Upsilon)}{\partial d_j} =\EE_{\ZZ_i^*\sim\mathcal{P}_i(\ww_i, \zzo_i\vert\Upsilon)}\left[\frac{\partial}{\partial d_{j}} \log \psi_i(\ww_i, \zzo_i, \ZZ_i^*\vert\Upsilon)\right],\\
    &\frac{\partial \log f_i(\ww_i, \zzo_i\vert\Upsilon)}{\partial \Lambda_j} =\EE_{\ZZ_i^*\sim\mathcal{P}_i(\ww_i, \zzo_i\vert\Upsilon)}\left[\frac{\partial}{\partial \Lambda_{j}} \log \psi_i(\ww_i, \zzo_i, \ZZ_i^*\vert\Upsilon)\right].
\end{aligned}
\]
Note that when $j \in \CC\backslash\AAA_i$, it always holds that $\frac{\partial \log f_i(\ww_i, \zzo_i\vert\Upsilon)}{\partial c_j} = \frac{\partial \log f_i(\ww_i, \zzo_i\vert\Upsilon)}{\partial d_j} = 0$.

For the derivatives with respect to $\uu_j$ when $j \in \DD\cap \AAA_i$, we can also perform change of measure and obtain:
\[\frac{\partial \log f_i(\ww_i, \zzo_i\vert\Upsilon)}{\partial  \uu_{j}} =\EE_{\ZZ_i^*\sim\mathcal{P}_i(\ww_i, \zzo_i\vert\Upsilon)}\left[\frac{\partial}{ \partial \uu_{j}} \log \left(\int_{c_{j, z_{ij}-1}}^{c_{j, z_{ij}}} \psi_i(\ww_i, \zzo_i, \zz_i^*\vert\Upsilon)\big\vert_{z_{ij}=v; z_{ik}^* = Z_{ik}^*, k \neq j}~dv \right)\right]\]
for each $i = 1, ..., N$. It is easy to check that the expression inside the expectation can be simplified as:
\[\begin{aligned}
    \frac{\partial}{ \partial \uu_{j}} \log \left(\int_{c_{j, z_{ij}-1}}^{c_{j, z_{ij}}} \psi_i(\ww_i, \zzo_i, \zz_i^*\vert\Upsilon)\big\vert_{z_{ij}=v; z_{ik}^* = Z_{ik}^*, k \neq j}~dv \right)
    = \frac{\partial}{ \partial \uu_{j}} \log \left(\int_{c_{j, z_{ij}-1}}^{c_{j, z_{ij}}} \phi(\zz_i^*\vert\LL\LL^\top)\big\vert_{z_{ij}^*=v; z_{ik}^* = Z_{ik}^*, k \neq j}~dv \right).
\end{aligned}\]
However, if $j \in \DD\backslash\AAA_i$, we always have $\frac{\partial \log f_i(\ww_i, \zzo_i\vert\Upsilon)}{\partial  \uu_{j}}  = \mathbf{0}.$

\paragraph{Approximation formulas} After obtaining a realisation $\zz_i^*$ of $\ZZ_i^*$ sampled from the distribution $\mathcal{P}_i(\ww_i, \zzo_i\vert\Upsilon)$, we can use the expression above to derive approximations for the derivatives of $l_1(\Upsilon)$ with respect to $\Upsilon$.

Specifically, for the derivatives with respect to $\LL$ we have
\begin{equation}\label{eq:derivcor}
    \begin{aligned}
    \frac{\partial \log f_i(\ww_i, \zzo_i\vert\Upsilon)}{\partial \LL} &\approx \frac{\partial}{\partial \LL} \log \psi_i(\ww_i, \zzo_i, \zz_i^*\vert\Upsilon) \\
 &= \left(\LL\LL^\top\right)^{-1}\zz_{i}^*(\zz_{i}^*)^\top\left(\LL^\top\right)^{-1}-\operatorname{diag}\left(\left(\frac{1}{l_{11}}, \ldots, \frac{1}{l_{pp}}\right)^\top\right),
\end{aligned}
\end{equation}
where $\operatorname{diag}(\vv)$ denotes the diagonal matrix whose diagonal is the given vector $\vv$.

For the derivatives with respect to $c_j$ and $d_j$ when $j \in \CC\cap\AAA_i$, we have
\begin{equation}\label{eq:derivcon}
    \begin{aligned}
         \frac{\partial \log f_i(\ww_i, \zzo_i\vert\Upsilon)}{\partial c_j} \approx \frac{\partial}{\partial c_j} \log \psi_i(\ww_i, \zzo_i, \zz_i^*\vert\Upsilon) &= \frac{1}{d_j}\left[(\LL\LL^\top)^{-1}\zz_{i}^*\right]_j,  \text{ and }\\
    \frac{\partial \log f_i(\ww_i, \zzo_i\vert\Upsilon)}{\partial d_j} \approx \frac{\partial}{\partial d_j} \log \psi_i(\ww_i, \zzo_i, \zz_i^*\vert\Upsilon) &= \frac{z_{ij} - 1}{d_j}\left[(\LL\LL^\top)^{-1}\zz_{i}^*\right]_j,
\end{aligned}
\end{equation}
where $\left[(\LL\LL^\top)^{-1}\zz_{i}^*\right]_j$ denote the $j-th$ entry of vector $(\LL\LL^\top)^{-1}\zz_{i}^*$.

For the derivatives with respect to $\Lambda_j$ when $j \in \mathcal{M}$, we have
\begin{equation}\label{eq:derivmea}
    \begin{aligned}
         \frac{\partial \log f_i(\ww_i, \zzo_i\vert\Upsilon)}{\partial \Lambda_j} \approx \frac{\partial}{\partial \Lambda_j} \log \psi_i(\ww_i, \zzo_i, \zz_i^*\vert\Upsilon) =  \frac{\partial}{\partial \Lambda_j} \log q_j(\ww_{ij}\vert z_{ij}^*; \Lambda_j).
\end{aligned}
\end{equation}

For the derivatives with respect to $\uu_j$ when $j \in \DD\cap\AAA_i$, there are two cases to consider:
\begin{enumerate}
    \item If $Z_{ij}$ is a binary variable, there is only one parameter, $u_{j1}$, whose derivative can be approximated by
    \begin{equation}\label{eq:derivbin}
        \begin{aligned}
	\frac{\partial \log f_i(\ww_i, \zzo_i\vert\Upsilon)}{\partial u_{j1}}  & \approx \frac{\partial}{\partial u_{j1}} \log \left(\int_{c_{j, z_{ij}-1}}^{c_{j, z_{ij}}} \phi(\zz_i^*\vert\LL\LL^\top)\big\vert_{z_{ij}^*=v}~dv \right) \\
    &= \begin{cases}
	-{\phi(\zz_{i}^*\vert\LL\LL^\top)\big\vert_{z_{ij}^*=u_{j1}}}\Big/{\int_{u_{j1}}^{\infty} \phi(\zz_i^*\vert\LL\LL^\top)\big\vert_{z_{ij}^*=v}~dv} , & \text{ if } z_{ij} = 1. \\		{\phi(\zz_{i}^*\vert\LL\LL^\top)\big\vert_{z_{ij}^*=u_{j1}}}\Big/{\int_{-\infty}^{u_{j1}} \phi(\zz_i^*\vert\LL\LL^\top)\big\vert_{z_{ij}^*=v}~dv} , & \text{ if } z_{ij} = 0.
	\end{cases}
\end{aligned}  
    \end{equation} 
    
    \item If $Z_{ij}$ is an ordinal variable, for each $k = 1, ..., K_j$, the derivative with respect to $u_{jk}$ can be approximated by
    \begin{equation}
        \label{eq:derivord}\frac{\partial \log f_i(\ww_i, \zzo_i\vert\Upsilon)}{\partial u_{jk}} \approx \sum_{l=1}^{K_j}  \frac{\partial}{\partial c_{jl}} \left[\log \left(\int_{c_{j, z_{ij}-1}}^{c_{j, z_{ij}}} \phi(\zz_i^*\vert\LL\LL^\top)\big\vert_{z_{ij}^*=v}~dv \right)\right] \cdot \frac{\partial c_{jl}}{\partial u_{jk}}.
    \end{equation}
    Here, we have used the chain rule, and for each $l = 1, ..., K_j$, we have
    \[\begin{aligned}
        \frac{\partial c_{jl}}{\partial u_{jk}} = \begin{cases}
		1, & \text{ if } k = 1,\\
		\exp(u_{jk}), & \text{ if } 2\leq k \leq l.\\
		0, &\text{ others,}
	\end{cases}
    \end{aligned}\]
    and
    \[\begin{aligned}
    &\frac{\partial}{\partial c_{jl}} \left[\log \left(\int_{c_{j, z_{ij}-1}}^{c_{j, z_{ij}}} \phi(\zz_i^*\vert\LL\LL^\top)\big\vert_{z_{ij}^*=v}~dv \right)\right]  \\&
    = 
        \begin{cases}
		{\phi(\zz_{i}^*\vert\LL\LL^\top)\vert_{z_{ij}^* = c_{jl}}}\Big/{\int_{c_{j,l-1}}^{c_{jl}} \phi(\zz_{i}^*\vert\LL\LL^\top)\vert_{z_{ij}^*= v}~dv} , & \text{ if } z_{ij} < K_j  \text{ and } l = z_{ij} + 1,\\
 		-{\phi(\zz_{i}^*\vert\LL\LL^\top)\vert_{z_{ij}^* = c_{jl}}}\Big/{\int_{c_{jl}}^{c_{j,l+1}} \phi(\zz_{i}^*\vert\LL\LL^\top)\vert_{z_{ij}^*(\tilde{\Xi}) = v}~dv} , & \text{ if } z_{ij} > 0 \text{ and } l = z_{ij}, \\		 		
		0, &\text{ others.}
	\end{cases}
    \end{aligned}\]    
\end{enumerate}

Finally, it is worth noting that any derivatives not explicitly mentioned in the previous equations are equal to zero.

\subsubsection{Gibbs sampler for sampling from \texorpdfstring{$\mathcal{P}_i(\ww_i, \zzo_i\vert\Upsilon)$}{1}}\label{sec:GibbsZ}
 The stochastic gradient algorithm requires sampling $\ZZ_i^*$ from $\mathcal{P}_i(\ww_i, \zzo_i\vert\Upsilon)$, which is actually the conditional distribution of $\ZZ_i^*$ given $\WW_i = \ww_i$ and $\ZZo_i = \zzo_i$, that is, 
 \[ p(\ZZ_{i}^* \vert \WW_i = \ww_i, \ZZo_i = \zzo_i; \Upsilon).\]
 We consider using a Gibbs sampler, which involves sequentially sampling $Z_{ij}^*$, for each $j = 1, ..., p$, from
 \[p(Z_{ij}^* \vert \ZZ_{i, -j}^*, \ZZo_i, \WW_i; \Upsilon)\]
If $j \in (\mathcal{C}\cap \AAA_i)$, it is easy to seed that 
$p(Z_{ij}^* \vert \ZZ_{i, -j}^*, \ZZo_i, \WW_i; \Upsilon) = \III(Z_{ij}^* = \frac{Z_{ij}-c_j}{d_j}).$
Thus we need only to discuss $j \in \mathcal{M}\cup (\mathcal{C}\backslash \AAA_i) \cup \mathcal{D} $, where $\mathcal{M} = \{1, ..., p_1\}$. We break them down into three cases:
\begin{enumerate}
    \item If $j\in\mathcal{M}$, then 
    \[\begin{aligned}
        p(Z_{ij}^* \vert \ZZ_{i, -j}^*, \ZZo_i, \WW_{i}; \Upsilon)&=p(Z_{ij}^* \vert \ZZ_{i, -j}^*, \WW_{ij}; \Upsilon) \\
        &~\propto~ p(Z_{ij}^* \vert \ZZ_{i, -j}^*; \LL\LL^\top) p(\WW_{ij} \vert Z_{ij}^*; \Lambda_j).
    \end{aligned}\]
    Note that $p(Z_{ij}^* \vert \ZZ_{i, -j}^*; \LL\LL^\top)$
    results in a Gaussian density for $Z_{ij}^*$ (after normalisation). If $\WW_{ij} \vert Z_{ij}$ follows a unidimensional linear factor model, then $p(\WW_{ij}\vert  Z_{ij}^*; \Lambda_j)$ will also  yield a  Gaussian density with respect to $Z_{ij}^*$,  thereby making $p(Z_{ij}^* \vert \ZZ_{i, -j}^*, \ZZo_i, \WW_{ij}; \Upsilon)$ a Gaussian distribution, which is easy to sample from. On the other hand, if $\WW_{ij} \vert Z_{ij}$ follows a unidimensional IRT model, then $p(Z_{ij}^* \vert \ZZ_{i, -j}^*, \ZZo_i, \WW_{ij}; \Upsilon)$ will have a log-concave density. In such case, we can employ adaptive sampling methods to sample from this distribution.

    \item If $j\in \AAA_i$, $Z_{ij}$ missing and we have \[\begin{aligned}
        p(Z_{ij}^* \vert \ZZ_{i, -j}^*, \ZZo_i, \WW_{i}; \Upsilon)&~\propto~p(Z_{ij}^* \vert \ZZ_{i, -j}^*; \LL\LL^\top),
    \end{aligned}\]
    which is a Gaussian distribution for $Z_{ij}^*$.

    \item If $j\in (\DD \cap \AAA_i)$, $Z_{ij}$ is discrete and missing. In this case we have 
     \[\begin{aligned}
        p(Z_{ij}^* \vert \ZZ_{i, -j}^*, \ZZo_i, \WW_{i}; \Upsilon)&~\propto~p(Z_{ij}^* \vert \ZZ_{i, -j}^*, Z_{ij}; \Upsilon)\\
        &~\propto~ p(Z_{ij}^* \vert \ZZ_{i, -j}^*;\LL\LL^\top)\times \mathbb{I}(Z_{ij}^* \in (c_{j, Z_{ij}}, c_{j, Z_{ij}+1}]),
    \end{aligned}\]
    which is a truncated normal distribution.
\end{enumerate}

\subsubsection{The stochastic proximal gradient algorithm}
We summarize the algorithm for solving optimisation problem \eqref{eq:remmlcopula} in Algorithm \ref{alg:SGD}, following the framework in \cite{zhang2022computation}.
\setcounter{algorithm}{3}
\begin{algorithm}[Stochastic proximal gradient algorithm for solving problem \eqref{eq:remmlcopula}]\label{alg:SGD}~
	\begin{itemize}
			\item[]{\bf Input:} Observed data $\ww_i$ and $\zzo_i,~i=1, \ldots, N$, initial value of all unknown parameters $\Upsilon^{(0)}$, initial value of underlying variables $\zz_i^{*(0)},~i=1, \ldots, N$, burn-in size $M$, number of follow-up iterations $T$, sequence of step sizes $\{\gamma_t\}_{t=1}^{M+T}$.
		\item[]{\bf Update:} At $t$-th iteration where $1 \leq t \leq M +T$, perform the following two steps:
		\begin{enumerate}
  \item{(\bf Sampling step)} For $i=1, 2,  \ldots, N$, sample $\zz_i^{*(t)}$  from $\mathcal{P}_i(\ww_i, \zzo_i\vert\Upsilon^{(t-1)})$ using the Gibbs sampler described in \ref{sec:GibbsZ}, starting from initial value $\zz_i^{*(t-1)}$.
				
			\item{(\bf Proximal gradient ascent step)} For each $i=1, ..., N$, compute the approximated gradient $$\hat{\nabla}_{\Upsilon}\log f_i(\ww_i, \zzo_i\vert \Upsilon)\vert_{\Upsilon = \Upsilon^{(t-1)}}$$
   using equations \eqref{eq:derivcor}, \eqref{eq:derivcon}, \eqref{eq:derivmea}, \eqref{eq:derivbin}, and \eqref{eq:derivord} with $\zz_i^* = \zz_i^{*(t)}$.  
			Compute $$\mathbf{G}^{(t)} = \sum_{i=1}^{N}  \hat{\nabla}_{\Upsilon}\log f_i(\ww_i, \zzo_i\vert \Upsilon)\vert_{\Upsilon = \Upsilon^{(t-1)}}$$
            as well as a positive definite matrix $\mathbf{H}^{(t)}$ as a preconditioning matrix. Then update $\Upsilon$ by
			\[\Upsilon^{(t)} = \operatorname{proj}_\mathbf{\Omega} \left(\Upsilon^{(t-1)} + \gamma_t  \left(\mathbf H^{(t)}\right)^{-1}\mathbf{G}^{(t)}\right).\]
			 Here $ \operatorname{proj}_\mathbf{\Omega}$ is the projection operator onto the feasible set $\mathbf{\Omega}$.
		\end{enumerate}
	\item[]{\bf Output:} $\hat{\Upsilon}= \frac{1}{T} \sum\limits_{t=B+1}^{B+T} \Upsilon^{(t)}.$
\end{itemize}
\end{algorithm}

\begin{remark}
In practice, we run only one scan of Gibbs sampling in the sampling step, which is enough for the algorithm to converge. 
{Also we can employ a random scan Gibbs sampler for a better mixing rate. See \cite{givens2012computational} for more details.}
\end{remark}

\begin{remark}
    The stepsize can be chosen as $\gamma_t~\propto~t^{-0.51}$, as is recommended in \cite{zhang2022computation}. In practice we use $\gamma_t = (t+100)^{-0.51}$.
\end{remark}
\begin{remark}
	In practice, initial value of $\{c_j, d_j, j\in \CC\}$ and $\{\uu_j, j\in \DD\}$ can be obtained by marginal distribution of the corresponding variables. Initial values of $\Lambda_1$, ..., $\Lambda_{p_1}$ can be obtained by performing factor analysis with a single factor. Initial value of $\SSigma = \LL\LL^\top$ can be obtained by the mixed correlation matrix composed of the {Pearson/ polychoric/ tetrachoric/ biserial/ polyserial} correlation of each pair of variables, where the type of correlation depends on the combination of variables types of the corresponding pair. This mixed correlation matrix can be computed by the  \textit{mixedCor} function in R package \textit{`psych'}.  If the raw mixed correlation matrix is not positive definite, we shall project it to its nearest positive semi-definite matrix.
\end{remark}
\begin{remark}
As long as the  preconditioning matrix $\mathbf H^{(t)}$ converges to a positive definite diagonal matrix, theoretical results in \cite{zhang2022computation} guarantee that the output of the algorithm {converges} as $T \to \infty$. In our study, during the burn-in period we set $\mathbf H^{(t)}$ as an diagonal matrix, whose diagonal elements are composed of
\begin{itemize}
    \item[]{(1)} $-\sum_{i=1}^{N}\frac{\partial^2}{\partial l_{ab}^2}\log \psi_i(\ww_i, \zzo_i, \zz_i^{*(t)}\vert\Upsilon^{(t)})$ for $a, b = 1, \ldots, p$,
    \item[]{(2)} $-\sum_{i: j \in \AAA_i}\frac{\partial^2}{\partial c_j^2}\log \psi_i(\ww_i, \zzo_i, \zz_i^{*(t)}\vert\Upsilon^{(t)})$ and $-\sum_{i: j \in \AAA_i}\frac{\partial^2}{\partial d_j^2}\log \psi_i(\ww_i, \zzo_i, \zz_i^{*(t)}\vert\Upsilon)$ for $j \in \CC$,
    \item[]{(3)} $-\sum_{i: j \in \AAA_i}\frac{\partial^2}{\partial u_{jk}^2} \log \left(\int_{c_{j, z_{ij}-1}}^{c_{j, z_{ij}}} \phi(\zz_i^{*(t)}\vert\LL^{(t)}(\LL^{(t)})^\top)\big\vert_{z_{ij}^{*(t)}=v}~dv \right) $ for $j \in \DD$ and $k = 1, \ldots, K_j$,
    \item[]{(4)} $-N$ for other parameters (if any).
\end{itemize}
After the burn-in period, we simply fix $\mathbf H^{(t)}$ to $\mathbf H^{(B)}$.
\end{remark}


\subsection{Gibbs sampler for constructing knockoff copies}\label{sec:Gibbs}
The most difficult step in constructing knockoff copies is to sample $\ZZ_i^*$ from their conditional distribution given $\WW_i, \ZZo_i$, and $\YYo_i$ (Step 1 in Algorithm \ref{alg:knockcons2}). To accomplish this step, we employ a Gibbs sampler combined with the data augmentation technique.

Specifically, instead of sampling $\ZZ_{i}^*$ from $p(\ZZ_{i}^* \vert \ZZo_i, \WW_i, \YYo_i; \bbb, \beta_0, \sigma^2, \Delta, \Xi, \Lambda_1, ..., \Lambda_{p_1}),$ 
 our goal is to  sample $(\theta_i, \ZZ_{i}^*)$ from their conditional distribution given $\YYo_i, \WW_i$, and $\ZZo_i$, which is written as
 \begin{equation}\label{eq:jointpost}
     p(\theta_i, \ZZ_{i}^* \vert \ZZo_i, \WW_i, \YYo_i; \bbb, \beta_0, \sigma^2, \Delta, \Xi, \Lambda_1, ..., \Lambda_{p_1}).
 \end{equation}
The Gibbs sampler used to sample from \eqref{eq:jointpost} can be divided into two parts:
\begin{enumerate}
    \item Sampling $\theta_i$ from $p(\theta_i \vert \ZZ_{i}^*, \ZZo_i, \WW_i, \YYo_i; \bbb, \beta_0, \sigma^2, \Delta, \Xi, \Lambda_1, ..., \Lambda_{p_1})$ 
    \item Sampling $Z_{ij}^*$ from $p(Z_{ij}^* \vert \theta_i, \ZZ_{i, -j}^*, \ZZo_i, \WW_i, \YYo_i; \bbb, \beta_0, \sigma^2, \Delta, \Xi, \Lambda_1, ..., \Lambda_{p_1})$ for each $j=1, ..., p$ sequentially.
\end{enumerate}

\subsubsection{Sampling \texorpdfstring{$\theta_i$}{1}}\label{sec:sampletheta}  First we note that since $Z_{ij}$ is a function of $Z^*_{ij}$, and $\WW_i$ is conditionally independent of $\theta_i$ given $\ZZ_i$, we need only to sample $\theta_i$ from
$$p(\theta_i \vert \ZZ_{i}, \YYo_i; \bbb, \beta_0, \sigma^2, \Delta),$$
where $\ZZ_i$ is obtained from $\ZZ_i^*$ based on the given Gaussian copula model.
The density function of this distribution can be written as
\[f_i(\theta_i \vert  \zz_i, \yyo_i; \bbb, \beta_0, \sigma^2, \Delta) = \frac{1}{\sqrt{2\pi\sigma^2}}  \exp\left(-\frac{1}{2\sigma^2}{\left(\theta_i - \beta_0 - \sum\limits_{j=1}^p \bbb_j^\top g_j(z_{ij})\right)^2} \right) \prod_{j\in \BBB_i}h_j(y_{ij}\vert \theta_i; \Delta),\]
where for each $j = 1, ..., J$, and $ h_j(y_{ij} \vert \theta_i; \Delta)$ denotes the the conditional probability density/mass function of $y_{ij}$ given $\theta_i$, which takes the form of \eqref{eq:2pl} or \eqref{eq:gpcm} depending on whether item $j$ is dichotomous or polytomous. It is easy to verify that the density function is log-concave with respect to $\theta_i$. Therefore, we can use the adaptive rejection sampling method \citep{gilks1992adaptive} to sample $\theta_i$ from $p(\theta_i \vert \ZZ_{i}, \YYo_i; \bbb, \beta_0, \sigma^2, \Delta)$.

\subsubsection{Sampling \texorpdfstring{$Z_{ij}^*$}{1}}\label{sec:sampleZ}
Since $\ZZ_i^*$ is conditionally independent of $\YYo_i$ given $\theta_i$,  our focus is on sampling $Z_{ij}^*$ from the distribution
\begin{equation}
    p(Z_{ij}^* \vert \theta_i, \ZZ_{i, -j}^*, \ZZo_i, \WW_i; \bbb, \beta_0, \sigma^2,  \Xi, \Lambda_1, ..., \Lambda_{p_1}).
\end{equation}
It is also worth noting that we need only to consider sampling $Z_{ij}^*$ for $j \in \mathcal{M}\cup (\mathcal{C}\backslash \AAA_i) \cup \mathcal{D} $, where $\mathcal{M} = \{1, ..., p_1\}$.
We break them down into four cases:
\begin{enumerate}
    \item If $j \in \mathcal{M}$, we have $Z_{ij}^* \perp \!\!\! \perp  \ZZo_i \vert \ZZ_{i,-j}^*$ and $Z_{ij}^* \perp \!\!\! \perp  \WW_{ik} \vert \ZZ_{i,-j}^*$ for $k\neq j$. Thus the target distribution will become
    \[p(Z_{ij}^* \vert \theta_i, \ZZ_{i, -j}^*, \WW_{ij}; \bbb, \beta_0, \sigma^2,  \Xi, \Lambda_j).\]
    Furthermore, since $\WW_{ij}^* \perp \!\!\! \perp  \theta_i \vert \ZZ_{i}^*$ and  $\WW_{ij}^* \perp \!\!\! \perp  \ZZ_{i, -j} \vert Z_{ij}^*$, we have
    \[\begin{aligned}
        &p(Z_{ij}^* \vert \theta_i, \ZZ_{i, -j}^*, \WW_{ij}; \bbb, \beta_0, \sigma^2,  \Xi, \Lambda_j) \\ &~\propto~p(Z_{ij}^*\vert \ZZ_{i, -j}^*; \SSigma) p(\theta_i\vert  Z_{ij}^*, \ZZ_{i, -j}^*; \bbb, \beta_0, \sigma^2, \Xi) p(\WW_{ij}\vert \theta_i, Z_{ij}^*, \ZZ_{i, -j}^*; \bbb, \beta_0, \sigma^2, \Xi, \Lambda_j) \\
        &=p(Z_{ij}^*\vert \ZZ_{i, -j}^*; \SSigma) p(\theta_i\vert  Z_{ij}^*, \ZZ_{i, -j}^*; \bbb, \beta_0, \sigma^2, \Xi) p(\WW_{ij}\vert  Z_{ij}^*; \Lambda_j)
    \end{aligned}\]
    Obviously, $p(Z_{ij}^*\vert \ZZ_{i, -j}^*; \SSigma)$
    result in a Gaussian density for $Z_{ij}^*$ (after normalisation). Note that 
    \[p(\theta_i\vert  Z_{ij}^*, \ZZ_{i, -j}^*; \bbb, \beta_0, \sigma^2, \Xi) ~\propto~\exp\left(-\frac{1}{2\sigma^2}\left[\bbb_{j}Z_{ij}^* - \left(\theta_i - \beta_0 - \sum\limits_{k\neq j}\bbb_{k}^\top g_k(F_j(Z^*_{ik}))\right)\right]^2\right),\]
   which also results in a Gaussian density with respect to $Z_{ij}^*$. If $\WW_{ij} \vert Z_{ij}$ follows a unidimensional linear factor model, then $p(\WW_{ij}\vert  Z_{ij}^*; \Lambda_j)$ will also  yield a  Gaussian density with respect to $Z_{ij}^*$,  thereby making $p(Z_{ij}^* \vert \theta_i, \ZZ_{i, -j}^*, \WW_{ij}; \bbb, \beta_0, \sigma^2,  \Xi, \Lambda_j)$ a Gaussian distribution, which is easy to sample from. On the other hand, if $\WW_{ij} \vert Z_{ij}$ follows a unidimensional IRT model, then $p(Z_{ij}^* \vert \theta_i, \ZZ_{i, -j}^*, \WW_{ij}; \bbb, \beta_0, \sigma^2,  \Xi, \Lambda_j)$ will have a log-concave density. In such case, we can employ adaptive sampling methods to sample from this distribution.
    
\item If $j\in \CC \backslash \AAA_i$, $Z_{ij}$ is continuous and missing. We have $Z_{ij}^* \perp \!\!\! \perp  \ZZo_i \vert \ZZ_{i,-j}^*$ and $Z_{ij}^* \perp \!\!\! \perp  \WW_i \vert \ZZ_{i,-j}^*$. 
Therefore, the target distribution becomes
    \[p(Z_{ij}^* \vert \theta_i, \ZZ_{i, -j}^*; \bbb, \beta_0, \sigma^2,  \Xi).\]
Since $g_j(Z_{ij}) = Z_{ij} =  d_jZ_{ij}^* + c_j$, 
	\[\begin{aligned}
		p&(Z_{ij}^*\big\vert  \theta_i, \ZZ_{i,-j}^*; \bbb, \beta_0, \sigma^2, \Xi)\\
  \propto& p(Z_{ij}^*\vert \ZZ_{i,-j}^*; \SSigma) p(\theta_i \vert Z_{ij}^*, \ZZ_{i,-j}^*; \bbb, \beta_0, \sigma^2, \Xi)  \\
		\propto &p(Z_{ij}^*\vert \ZZ_{i,-j}^*; \SSigma) 
 \times\exp\left(-\frac{1}{2\sigma^2}\left[\bbb_{j}(d_jZ_{ij}^* + c_j) - \left(\theta_i - \beta_0 - \sum\limits_{k\neq j}\bbb_{k}^\top g_k(F_k(Z_{ik}^*))\right)\right]^2\right).	\end{aligned}\]
Thus $p(Z_{ij}^*\big\vert  \theta_i, \ZZ_{i,-j}^*; \bbb, \beta_0, \sigma^2, \Xi)$ is a Gaussian distribution.

	\item If $j\in \DD \backslash \AAA_i$, $Z_{ij}$ is discrete and missing. Similar to the previous case, we have  $Z_{ij}^* \perp \!\!\! \perp  \ZZo_i \vert \ZZ_{i,-j}^*$ and $Z_{ij}^* \perp \!\!\! \perp  \WW_i \vert \ZZ_{i,-j}^*$, and the target distribution becomes
    \[p(Z_{ij}^* \vert \theta_i, \ZZ_{i, -j}^*; \bbb, \beta_0, \sigma^2,  \Xi).\]
 For convenience, we denote $\bbb_j = (\beta_{j,1}, \ldots, \beta_{j,K_j})^\top$ and  set $\beta_{j,0} = 0$. Then
	\[\begin{aligned}
		&p(Z_{ij}^* \vert \theta_i, \ZZ_{i, -j}^*; \bbb, \beta_0, \sigma^2,  \Xi) \\ \propto& p(Z_{ij}^*\vert \ZZ_{i,-j}^*; \SSigma) p(\theta_i \vert Z_{ij}^*, \ZZ_{i,-j}^*; \bbb, \beta_0, \sigma^2, \Xi) \\
		\propto&p(Z_{ij}^*\vert \ZZ_{i,-j}^*; \SSigma)\\
  &\times\left\{\sum_{k=0}^{K_j}\exp\left(-\frac{1}{2\sigma^2}\left[\sum_{l=0}^{k}\beta_{j,l}  - \left(\theta_i - \beta_0 - \sum\limits_{m\neq j}\bbb_{m}^\top g_m(F_m(Z_{im}^*))\right)\right]^2\right)\mathbb{I}(Z_{ij}^*\in (c_{j, k}, c_{j, k+1}])\right\}.
	\end{aligned}\]
	This distribution can be viewed as a mixture of truncated normal distribution. To sample from this distribution, we can first sample a category $k$ from the set $\{0, 1, \ldots, K_j\}$ with corresponding weights $$\left\{\sum_{k=0}^{K_j}\exp\left(-\frac{1}{2\sigma^2}\left[\sum_{l=0}^{k}\beta_{j,l}  - \left(\theta_i - \beta_0 - \sum\limits_{m\neq j}\bbb_{m}^\top g_m(F_m(Z_{im}^*))\right)\right]^2\right){\cdot \int_{c_{j,k}}^{c_{j, k+1}} p(Z_{ij}^* = v\vert \ZZ_{i,-j}^*; \SSigma) dv }  \right\}_{k=0}^{K_j},$$
	and then  sample $Z_{ij}^*$ from the truncated normal distribution $$p(Z_{ij}^*\vert \ZZ_{i,-j}^*; \SSigma)
 \times\mathbb{I}(Z_{ij}^* \in (c_{jk}, c_{j, k+1}]).$$
	\item If $j\in \DD \cap \AAA_i$, $Z_{ij}$ is discrete and observed. In this case, the target distribution will become
    \[p(Z_{ij}^* \vert \theta_i, \ZZ_{i, -j}^*, Z_{ij}; \bbb, \beta_0, \sigma^2,  \Xi).\]
    Observe that    
	\[\begin{aligned}
		p(Z_{ij}^* \vert \theta_i, \ZZ_{i, -j}^*, Z_{ij}; \bbb, \beta_0, \sigma^2,  \Xi)
		\propto& p(Z_{ij}^*\vert \ZZ_{i,-j}^*; \SSigma) p(\theta_i \vert Z_{ij}^*, \ZZ_{i,-j}^*, Z_{ij}; \bbb, \beta_0, \sigma^2, \Xi)\\
		\propto& p(Z_{ij}^*\vert \ZZ_{i,-j}^*; \SSigma) \times\mathbb{I}(Z_{ij}^*\in (c_{j, Z_{ij}}, c_{j, Z_{ij}+1}]) \\&\times\exp\left(-\frac{1}{2\sigma^2}\left(\theta_i - \beta_0 - \bbb_j^\top g_j(Z_{ij}) - \sum\limits_{k\neq j}\bbb_k^\top g_k(F_k({Z}^*_{ik}))\right)^2\right)\\ 
		\propto& p(Z_{ij}^* \vert \ZZ_{i,-j}^*; \SSigma)\times\mathbb{I}(Z_{ij}^*\in (c_{j, Z_{ij}}, c_{j, Z_{ij}+1}]).
	\end{aligned}\]
	Hence $p(Z_{ij}^* \vert \theta_i, \ZZ_{i, -j}^*, Z_{ij}; \bbb, \beta_0, \sigma^2,  \Xi)$ is a univariate truncated normal distribution, which can be efficiently sampled from.
\end{enumerate}

\subsection{Algorithm for computing MLE of  latent regression model and extended latent regression model}

In this section, we will give a detailed description of the stochastic EM algorithm used to compute the MLE of latent regression model and extended latent regression model. 

Given $\Xi$ and $\Lambda_1$, ..., $\Lambda_{p_1}$,  the likelihood function $l_2(\bbb, \beta_0, \sigma^2, \Delta)$ of the latent regression model takes the form:
\begin{equation}
l_2(\bbb, \beta_0, \sigma^2, \Delta) = \sum_{i=1}^N \log \left(\int \cdots \int \left(\prod_{j\notin \AAA_i }~dz_{ij}\right) f(\zz_i\vert  \Xi) \left(\prod_{j=1}^{p_1} q_j(\mathbf w_{ij}\vert z_{ij};  \Lambda_j)\right) f_i(\yyo_i \vert \zz_i; \bbb, \beta_0, \sigma^2, \Delta)\right),
\end{equation}
where
\[
f(\zz\vert \Xi)= \int\ldots\int  \left(\prod_{j\in \DD }~dz_{j}^{*}\right) \left[\phi(\zz^*\vert \SSigma) \times \left(\prod_{j\notin\DD} d_j^{-1}\right) \times  \left( \mathop{\prod}_{j\in \DD }\III (z_{j}^* \in  (c_{j, z_{j}-1},c_{j, z_{j}} ])\right) \right]\Bigg\vert_{z_{j}^* = \frac{z_{j} - c_j}{d_j}, j\notin\DD},
\]
and
{\small
\[
\begin{aligned}
    &f_i(\yyo_i \vert \zz_i; \bbb, \beta_0, \sigma^2, \Delta) = \\
    &\left( \frac{1}{\sqrt{2\pi\sigma^2}} \int \cdots \int d\theta_i (\prod_{j \notin \BBB_i}dy_{ij})  h(\yy_i\vert \theta_i; \Delta) \exp\left(-\frac{(\theta_i - (\beta_0 + \bbb_1^\top g_1(z_{i1}) + \cdots + \bbb_p^\top g_p(z_{ip}))^2}{2\sigma^2}\right) \right).
\end{aligned}
\]}
The likelihood function $\tilde l_2  (\bbb, \rrr) $ of the extended latent regression model is in a similar form:
{\small
\begin{equation}
\begin{aligned}
    \tilde l_2  (\bbb, \rrr) =&\sum_{i=1}^N \log \left\{\int \cdots \int \left[\left(\prod_{j\notin \AAA_i }~dz_{ij}\right) \left(\prod_{j\notin \AAA_i }~d\tilde z_{ij}\right) \right.\right.\times\\
    &\quad \left.\left.
    f(\zz_i, \zzk_i \vert \Xi) { \left(\prod_{j=1}^{p_1} q_j(\mathbf w_{ij}\vert z_{ij};  \Lambda_j)\right)\left(\prod_{j=1}^{p_1} q_j(\tilde{\mathbf w}_{ij}\vert \tilde{z}_{ij};  \Lambda_j)\right)} f_i(\yyo_i \vert \zz_i, \zzk_i; \bbb,\rrr, \beta_0, \sigma^2, \Delta)\right]\right\},
\end{aligned}
\end{equation}}
where
\begin{equation}\label{eq:jointzz}
    \begin{aligned}
    f(\zz, \zzk \vert \Xi) = &\int\ldots\int  \left(\prod_{j\in \DD }~dz_{j}^{*}\right) 
\left(\prod_{j\in \DD }~d\tilde{z}_{j}^{*}\right) \times \phi(\zz^*, \zzk^*\vert \GG) \times \left(\prod_{j\notin\DD} d_j^{-2}\right) \times  \\
&\left[\left( \mathop{\prod}_{j\in \DD }\III (z_{j}^* \in  (c_{j, z_{j}-1},c_{j, z_{j}} ])\right) \times  \left( \mathop{\prod}_{j\in \DD }\III (\tilde{z}_{j}^* \in  (c_{j, \tilde{z}_{j}-1},c_{j, \tilde{z}_{j}} ])\right)\right]\Bigg\vert_{z_{j}^* = \frac{z_{j} - c_j}{d_j}, \tilde{z}_{j}^* = \frac{\tilde{z}_{j} - c_j}{d_j},  j\notin\DD},
\end{aligned}
\end{equation}
and
\begin{equation*}
    \begin{aligned}
      &f(\yyo_i \vert \zz_i, \zzk_i; \bbb,\rrr, \beta_0, \sigma^2,\Delta) = \\
      &\left( \frac{1}{\sqrt{2\pi\sigma^2}} \int \cdots \int d\theta_i  (\prod_{j \notin \BBB_i}dy_{ij})  h(\yy_i\vert \theta_i; \Delta) \exp\left(-\frac{(\theta_i - (\beta_0 + \bbb_1^\top g_1(z_{i1}) + ... + \rrr_1^\top g_1(\tilde z_{i1}) + ... + \rrr_p^\top g_p(\tilde z_{ip}))^2}{2\sigma^2}\right) \right).  
    \end{aligned}
\end{equation*}
Note that in \eqref{eq:jointzz}, $\GG$ is the covariance matrix for $(\ZZ_i^*, \ZZk_i)$ as defined in Algorithm \ref{alg:knockcons2}, given by
\[\GG = 	\begin{pmatrix} \SSigma &  \SSigma-\SS\\  \SSigma-\SS & \SSigma\end{pmatrix},\]
where $\mathbf{S}$ is a diagonal matrix constructed by MVR procedure \citep{spector2022powerful}.

The similarity in structure between both likelihood functions is quite apparent. Hence, for the sake of simplicity, our focus will be on describing the algorithm to maximise the simper one, namely $l_2(\bbb, \beta_0, \sigma^2, \Delta)$.
In Algorithm \ref{alg:StEM}, we will provide the comprehensive details of the stochastic EM algorithm for maximizing $l_2(\bbb, \beta_0, \sigma^2, \Delta)$.
The stochastic EM algorithm used to maximise the $\tilde l_2(\bbb, \rrr)$  requires only minor modifications.

Before presenting the algorithm, we introduce the complete data log-likelihood function for clarity of exposition. For each $i=1, ..., N$, we define
\[cl_i(\bbb, \beta_0, \sigma^2, \Delta \vert \theta_i, \zz^*_i, \yyo_i) = -\frac{1}{2}\log\sigma^2 - \frac{1}{2\sigma^2} \left(\theta_i - \beta_0 - \sum\limits_{j=1}^p \bbb_j^\top g_j(F_j(z_{ij}^*)) \right)^2 - \sum\limits_{j\in \BBB_i} \log h_j(y_{ij} \vert \theta_i; \Delta), \]
where for each $j = 1, ..., J$, and $ h_j(y_{ij} \vert \theta_i; \Delta)$ denotes the the conditional probability density/mass function of $y_{ij}$ given $\theta_i$, which takes the form of \eqref{eq:2pl} or \eqref{eq:gpcm} depending on whether item $j$ is dichotomous or polytomous.

Now we can present the proposed optimisation algorithm:
\begin{algorithm}[Stochastic EM algorithm for maximizing $l_2(\bbb, \beta_0, \sigma^2, \Delta)$]\label{alg:StEM}~
	\begin{itemize}
        \item[]{\bf Input:} Item responses data $\yyo_i, i = 1, ..., N$, observed data $\ww_i$ and $\zzo_i,~i=1, \ldots, N$, initial value of unknown parameters $\bbb^{(0)}$, $\beta_0^{(0)}$, $(\sigma^2)^{(0)}$, and $\Delta^{(0)}$, initial value of underlying variables $\zz_i^{*(0)},~i=1, \ldots, N$,  parameters of the Gaussian copula model ${\Xi}$, parameters $\Lambda_1, ..., \Lambda_{p_1}$ in measurement models for $Z_{ij}$, burn-in size $M$, number of follow-up iterations $T$,

		
		\item[]{\bf Update:}  At the $t$-th iteration where $1 \leq t \leq M +T$, perform the following steps:
		\begin{enumerate}
			\item{(\bf Sampling step)} For $i = 1, 2, \ldots N$, sample
            $(\theta_i^{(t)}, \zz_i^{*(t)})$ from the conditional distribution
            $$p(\theta_i, \ZZ_{i}^* \vert \ZZo_i=\zzo_i, \WW_i=\ww_i, \YYo_i=\yyo_i; \bbb^{(t-1)}, \beta_0^{(t-1)}, (\sigma^2)^{(t-1)}, \Delta^{(t-1)}, \Xi, \Lambda_1, ..., \Lambda_{p_1})$$
            using the Gibbs sampler described in Section \ref{sec:Gibbs}, starting from the initial value 
 $(\theta_i^{(t-1)}, \zz_i^{*(t-1)})$.
			 
			\item{(\bf Maximisation step)}
            Update $(\bbb, \beta_0, \sigma^2, \Delta)$  by 
            \begin{equation}
                (\bbb^{(t)}, \beta_0^{(t)}, (\sigma^2)^{(t)}, \Delta^{(t)}) = \argmax_{\bbb, \beta_0, \sigma^2, \Delta} \left\{\sum_{i=1}^N cl_i(\bbb, \beta_0, \sigma^2, \Delta \vert \theta_i^{(t)}, \zz_i^{*(t)}, \yyo_i)\right\}.
            \end{equation}            
		\end{enumerate}
		\item[]{\bf Output:} $\hat\bbb = \frac{1}{T}\sum\limits_{t=B}^{B+T} \bbb^{(t)}, 
  \hat\beta_0 = \frac{1}{T}\sum\limits_{t=B}^{B+T} \beta_0^{(t)},
  \hat\sigma^2  = \frac{1}{T}\sum\limits_{t=B}^{B+T} (\sigma^2)^{(t)}$,  and $\hat\Delta  = \frac{1}{T}\sum\limits_{t=B}^{B+T} \Delta^{(t)}$.
	\end{itemize}
\end{algorithm}

\begin{remark}
    Similar to Algorithm~\ref{alg:SGD}, within the sampling step, we can perform only a single round of random scanning using the Gibbs sampler described in \ref{sec:Gibbs}. 
\end{remark}
\begin{remark}
    In the maximisation step, the optimization problem can be equivalently formulated as solving a series of (generalized) linear regression problems. For example, at the $t$-th iteration, $(\beta_0^{(t)}, \bbb^{(t)})$ can be obtained by regressing $\{\theta_i^{(t)}\}_{i=1}^N$ on $\{g_j(F_j(z_{ij}^{*(t)}))\}_{i=1}^N$, and the iterated values of the parameters in the measurement model for the $j$-th item, denoted as $\Delta_j^{(t)}$, can be obtained by regressing $\{y_{ij}\}_{i:j\in\BBB_i}$ on $\{\theta_i^{(t)}\}_{i:j\in\BBB_i}$ for each $j = 1, ..., J$.
\end{remark}

\subsection{The \texorpdfstring{$l_2$}{1}-penalisation}
In the simulation settings described in Section~\ref{sec:sim} and Appendix~\ref{app:sim} below, some cases arise where the number of predictors is comparable to the sample size (e.g., $p=100$ and $N=1000$). Such a scenario may result in overfitting during the M-step of the stochastic EM algorithm, causing the algorithm to become unstable.
To mitigate this problem, we propose to incorporate an $l_2$-penalty into the likelihood functions \eqref{eq:loglik} and\eqref{eq:likelihood}
during the aforementioned simulations.

Specifically, for \eqref{eq:loglik} we apply an $l_2$-penalty on $\bbb$, and the resulting penalised log-likelihood is denoted as $l_2^{\lambda}(\bbb, {\beta_0, } \sigma^2, \Delta)$,  where $\lambda > 0$ represents the tuning parameter. The $l_2^{\lambda}(\bbb, {\beta_0, } \sigma^2, \Delta)$ is formulated as
\begin{equation}\label{eq:penloglik}
l_2^{\lambda}(\bbb, {\beta_0, } \sigma^2, \Delta) = \frac{1}{N}\sum_{i=1}^N l_2(\bbb, {\beta_0, } \sigma^2, \Delta) - \lambda\|\bbb\|_2^2.
\end{equation}
Similarly, for \eqref{eq:likelihood}  we impose an $l_2$-penalty on the vector $(\bbb^\top, \rrr^\top)^\top$ and the resulting penalised log-likelihood is denoted as $\tilde l_2^{\lambda}(\bbb, \rrr)$, which is given by
\begin{equation}\label{eq:penlikelihood}
\tilde l_2^{\lambda}(\bbb, \rrr) = \frac{1}{N}\tilde l_2(\bbb, \rrr) - \lambda(\|\bbb\|_2^2 + \|\rrr\|_2^2).
\end{equation}
In all of our simulation studies, we set $\lambda = \sqrt{\frac{1}{N}}$.

To optimize \eqref{eq:penloglik} and \eqref{eq:penlikelihood}, we need only to modify the M-step within the stochastic EM algorithm. Efficient solutions to the corresponding sub-problems can be achieved using standard algorithms such as the \textit{glmnet} algorithm \citep{friedman2009glmnet}.

\section{Real Data}\label{app:realdata}
\subsection{Data Preprocessing}
%

We perform the following pre-processing steps:
\begin{enumerate}
	\item We discard samples and variables that include too much missing values.
	\item For ordinal variables, we merge the categories with too few samples (below 5\% of observed samples) into their adjacent categories.
	\item For highly correlated variables (correlations higer than 0.7), we discard some of them, or combine them into less correlated variables. This step is necessary since technically we don't want $[\XX, \XXk]$ to suffer from issue of multicollinearity.
	\item We log-transform the variable OUTHOU.
\end{enumerate}

\subsection{Candidate Variables}

\begin{center}
	\begin{longtable}{p{2cm}p{1cm}p{10.5cm}}
		\caption{Full description of candidate variables used in real data analysis.}
		\label{tab:var}\\
		\toprule
		Name & Type & Description \\		
		\midrule				
		OUT.FRI &B& Whether the student meet friends or talk to friends on the phone outside the school.  This variable is obtained by combining questionnaire items ST076Q07NA (whether meet or talk to friends before school) and ST078Q07NA (whether meet or talk to friends after school). Coding: 1 = yes, 0 = no.\\ \hline
		OUT.GAM &B& Whether the student  play video-games outside the school.  This variable is obtained by combining questionnaire items ST076Q06NA (whether play video-games before school) and ST078Q06NA (whether play video-games after school). Coding: 1 = yes, 0 = no.\\ \hline
		OUT.HOL &B&  Whether the student work in the household outside the school.  This variable is obtained by combining questionnaire items ST076Q09NA (whether   work in the household  before school) and ST078Q09NA (whether  work in the household  after school). Coding: 1 = yes, 0 = no.\\ \hline
		OUT.NET &B& Whether the student use Internet outside the school.  This variable is obtained by combining questionnaire items ST076Q05NA (whether  use Internet  before school) and ST078Q05NA (whether use Internet   after school). Coding: 1 = yes, 0 = no.\\ \hline
		OUT.JOB &B& Whether the student work for pay outside the school.  This variable is obtained by combining questionnaire items ST076Q10NA (whether work for pay before school) and ST078Q10NA (whether work for pay after school). Coding: 1 = yes, 0 = no.\\ \hline
		OUT.MEA &B& Whether the  student have meals before shool or after school.  This variable is obtained by combining questionnaire items ST076Q01NA (whether have breakfast before school) and ST078Q01NA (whether have dinner after school). Coding: 1 = yes, 0 = no.\\ \hline
		OUT.PAR &B& Whether the student talk to parents outside the school.  This variable is obtained by combining questionnaire items ST076Q08NA (whether  talk to parents  before school) and ST078Q08NA (whether  talk to parents  after school). Coding: 1 = yes, 0 = no.\\ \hline
		OUT.REA &B& Whether the student read a book/newspaper/magazine outside the school.  This variable is obtained by combining questionnaire items ST076Q04NA (whether read a book/newspaper/magazine before school) and ST078Q04NA (whether read a book/newspaper/magazine after school). Coding: 1 = yes, 0 = no.\\ \hline
		OUT.SPO &B& Whether the student exercise or do a sport outside the school.  This variable is obtained by combining questionnaire items ST076Q11NA (whether exercise or do a sport before school) and ST078Q11NA (whether exercise or do a sport after school). Coding: 1 = yes, 0 = no.\\ \hline
		OUT.STU &B&  Whether the student study for shool or homework outside the school.  This variable is obtained by combining questionnaire items ST076Q02NA (whether study before school) and ST078Q02NA (whether study after school). Coding: 1 = yes, 0 = no.\\ \hline
		OUT.VED &B& Whether the student watch TV/DVD/Video outside the school.  This variable is obtained by combining questionnaire items ST076Q03NA (whether watch TV/DVD/Video before school) and ST078Q03NA (whether watch TV/DVD/Video after school). Coding: 1 = yes, 0 = no.\\ \hline				
		REPEAT &B& Has the student ever repeated a grade. Coding: 1 = yes, 0 = no. \\ \hline	
		SCI.APP &B& Whether the student attend applied sciences and technology courses in this school year or last school year. This variable is obtained by combining questionnaire items ST063Q05NA (whether attend this year) and ST063Q05NB (whether attend last year).  Coding: 1 = yes, 0 = no.\\ \hline
		SCI.BIO &B& Whether the student attend biology courses in this school year or last school year. This variable is obtained by combining questionnaire items ST063Q03NA (whether attend this year) and ST063Q03NB (whether attend last year).  Coding: 1 = yes, 0 = no.\\ \hline
		SCI.CHE &B& Whether the student attend chemistry courses in this school year or last school year. This variable is obtained by combining questionnaire items ST063Q02NA (whether attend this year) and ST063Q02NB (whether attend last year).  Coding: 1 = yes, 0 = no.\\ \hline		
		SCI.EAR &B& Whether the student attend earth and space courses in this school year or last school year. This variable is obtained by combining questionnaire items ST063Q04NA (whether attend this year) and ST063Q04NB (whether attend last year).  Coding: 1 = yes, 0 = no.\\ \hline		
		SCI.GEN &B& Whether the student attend general, integrated, or comprehen science courses in this school year or last school year. This variable is obtained by combining questionnaire items ST063Q06NA (whether attend this year) and ST063Q06NB (whether attend last year).  Coding: 1 = yes, 0 = no.\\ \hline		
		SCI.PHY &B& Whether the student attend physics courses in this school year or last school year. This variable is obtained by combining questionnaire items ST063Q01NA (whether attend this year) and ST063Q01NB (whether attend last year).  Coding: 1 = yes, 0 = no.\\ \hline
		GENDER &B& Student's gender. The original item is named `ST004D01T'. Coding: 1 = male, 0 = female.\\ \hline
		LANGAH &B& Is student's language at home different from the test language. The original item is named `ST022Q01TA'. Coding:  1 = yes, 0 = no.\\ \hline
		DUECEC &O& Duration in ECEC (Early Childhood Education and Care) of student. The original variable is named `DURECEC' in the original dataset, which has 9 categories. We merged them into  4 categories. Coding: 0 = attended ECEC for less than two years, 1 = attended ECEC for at least two but less than three years, 2 = attended ECEC for at least three but less than four years, 3 = attended ECEC for at least four years.\\ \hline
		FISCED &O& Father's education in ISCED (International Standard Classification of Education) level. The original variable has 6 categories, we merged them into 5 categories. Coding: 0 = none or ISCED 1 (primary education), 1 = ISCED 2 (lower secondary), 2 = ISCED 3B or 3C (vocational/pre-vocational upper secondary), 3 = ISCED 3A (general upper secondary) or 4 (non-tertiary
		post-secondary), 4 = ISCED 5B (vocational tertiary), 5 = ISCED 5A  (theoretically oriented tertiary) or 6 (post-graduate). \\ \hline
		GRADE &O& Student's grade. 0 = lower than modal grade, 1 = not lower than modal grade, 2 = higher than modal grade. Here the modal grade is the grade level that most 15-year-old students in the country attend.\\ \hline
		MISCED &O& Mother's education in ISCED (International Standard Classification of Education) level. The original variable has 6 categories, we merged them into 5 categories. Coding: 0 = none or ISCED 1 (primary education), 1 = ISCED 2 (lower secondary), 2 = ISCED 3B or 3C (vocational/pre-vocational upper secondary), 3 = ISCED 3A (general upper secondary) or 4 (non-tertiary
		post-secondary), 4 = ISCED 5B (vocational tertiary), 5 = ISCED 5A or ISCED 6 (theoretically oriented tertiary
		and post-graduate).\\ \hline
		DAYPEC &O& Averaged days that student attends physical education classes each week. The original item is named `ST031Q01NA', which has 8 categories. We merged them into 4 categories. Coding: 0 = 0 days, 1 = 1 or 2 days, 2 = 3 or 4 days, 3 = 5 days or more.\\ \hline
		DAYMPA &O& Number of days with moderate physical activities for a total of at least 60 minutes per each week. The original item is named `ST032Q01NA', which has 8 categories. Coding:  0 = 0 days, 1 = 1 day, 2 = 2 days, 3 = 3 days, 4 = 4 days, 5 =5 days, 6 = 6 days, 7 = 7 days, 8 = 8 days. \\ 
		\hline
		SKIDAY &O& How often did student skipped a whole school day in the last two full weeks of school. The original item is named `ST062Q01TA', which has 4 categories. We merged them into 3 categories. Coding: 0 = none, 1 = one or two times, 2 = three or more times. \\ 
		\hline
		SKICAL &O& How often did student skipped some classes in the last two full weeks of school. The original item is named `ST062Q02TA', which has 4 categories. We merged them into 3 categories. Coding: 0 = none, 1 = one or two times, 2 = three or more times.\\ 
		\hline
		ARRLAT &O& How often did student arrived late for school in the last two full weeks of school. The original item is named `ST062Q03TA', which has 4 categories. We merged them into 3 categories.. Coding: 0 = none, 1 = one or two times, 2 = three or more times.\\ \hline		
		CHOCOU &O&  Can student choose the school science course(s) he or she study.  The original item is named `ST064Q01NA', which has 3 categories. Coding: 0 = no, not at all; 1 = yes, to a certain degree; 2 = yes, can choose freely. \\ \hline
		CHODIF &O&  Can student choose the level of difficulty for school science course(s) he or she study.  The original item is named `ST064Q02NA', which has 3 categories. Coding: 0 = no, not at all; 1 = yes, to a certain degree; 2 = yes, can choose freely.\\ 
		\hline
		CHONUM &O& Can student choose the number of school science course(s) he or she study. The original item is named `ST064Q03NA', which has 3 categories. Coding: 0 = no, not at all; 1 = yes, to a certain degree; 2 = yes, can choose freely.\\ 
		\hline
		EISCED &O& The ISCED level that student expects to complete. The original item is named `ST111Q01TA', which has 6 categories. We merged them into 3 categories. Coding: 0 = level 2 or 3A, 1 = level 4 or 5B, 2 = level 5A or 6. \\ \hline
		
		ADINST &C& Adaption of instruction. This variable is derived based on IRT scaling. The observed values range from -1.97 to 2.04.\\ \hline
		ANXTES &C& Personality: test anxiety. This variable is named `ANXTEST' in the original dataset, which is derived based on IRT scaling. The observed values range from -2.51 to 2.55.\\ \hline
		BELONG &C& Subjective well-being: sense of belonging to school. This variable is derived based on IRT scaling. The observed values range from -3.13 to 2.61.\\ \hline
		FISEIO &C& ISEI (International Socio-economic Index) of occupational status of father. This variable is named `BFMJ2' in the original dataset. The observed values range from 12 to 89.\\ \hline
		MISEIO &C& ISEI (International Socio-economic Index) of occupational status of mother. This variable is named `BMMJ1' in the original dataset. The observed values range from 12 to 89.\\ \hline
		EISEIO &C& Student's expected ISEI of occupational status. This variable is named `BSMJ' in the original dataset. The observed values range from 16 to 89.\\ \hline
		COOPER &C& Collaboration and teamwork dispositions: enjoy cooperation. This variable is named `COOPERATE' in the original dataset, which is derived based on IRT scaling. The observed values range from -3.33 to 2.29.\\ \hline
		CPSVAL &C& Collaboration and teamwork dispositions: value cooperation. This variable is named `CPSVALUE' in the original dataset, which is based on IRT scaling. The observed values range from -2.83 to 2.14.\\ \hline
		CULTPO &C& Cultural possessions at home. This variable is named `CULTPOSS' in the original dataset, which is derived based on IRT scaling. The observed values range from -1.71 to 2.63.\\ \hline
		DISCLI &C& Disciplinary climate in science classes. This variable is named `DISCLISCI' in the original dataset, which is derived based on IRT scaling. The observed values range from -2.41 to 1.88.\\ \hline
		EMOSUP &C& Parents emotional support. This variable is named `EMOSUPS' in the original dataset, which is derived based on IRT scaling. The observed values range from -3.08 to 1.10.\\ \hline
		ENVAWA &C& Environmental awareness. This variable is named `ENVAWARE' in the original dataset, which is derived based on IRT scaling. The observed values range from -3.38 to 3.29.\\ \hline
		ENVOPT &C& Environmental optimism. This variable is derived based on IRT scaling. The observed values range from -1.80 to 3.01.\\ \hline		
		EPIST &C& Epistemological beliefs. This variable is derived based on IRT scaling. The observed values range from -2.79 to 2.16.\\ \hline
		HEDRES &C& Home educational resources. This variable is derived based on IRT scaling. The observed values range from -4.41 to 1.78. \\ \hline
		EBSCIT &C& Enquiry-based science teaching and learning practices. This variable is named `IBTEACH' in the original dataset, which is derived based on IRT scaling. The observed values range from -3.34 to 3.18. \\ \hline
		INSTSC &C& Instrumental motivation. This variable is named `INSTSCIE' in the original dataset, which is derived based on IRT scaling. The observed values range from -1.93 to 1.74. \\ \hline
		INTBRS &C& Interest in broad science topics. This variable is named `INTBRSCI' in the original dataset, which is derived based on IRT scaling. The observed values range from -2.55 to 2.73.\\ \hline
		JOYSCI &C& Enjoyment of science. This variable is named `JOYSCIE' in the original dataset, which is derived based on IRT scaling. The observed values range from -2.12 to 2.16.\\ \hline
		MOTIVA &C& Achievement motivation. This variable is named `MOTIVAT' in the original dataset, which is derived based on IRT scaling. The observed values range from -3.09 to 1.85. \\ \hline
		OUTHOU &C& Out-of-school study time per week (hours).This variable is obtained from the original variable named `OUTHOURS' by setting 0 as missing and performing a log-transformation. The observed values range from 0 to 4.25. \\ \hline
		PERFEE &C& Perceived feedback. This variable is named `PERFEED' in the original dataset, which is derived based on IRT scaling. The observed values range from -1.53 to 2.50. \\ \hline				
		SCIACT &C& Index science activities. This variable is named `SCIEACT' in the original dataset, which is derived based on IRT scaling. The observed values range from -1.76 to 3.36. \\ \hline
		SCIEEF &C& Science self-efficacy. This variable is named `SCIEEFF' in the original dataset, which is derived based on IRT scaling. The observed values range from -3.76 to 3.28. \\ \hline		
		TDSCIT &C& Teacher-directed science instruction.  This variable is named `TDTEACH' in the original dataset, which is derived based on IRT scaling. The observed values range from -2.45 to 2.08.\\ \hline
		TEASUP &C& Teacher support in a science classes of students choice. This variable is named `TEACHSUP' in the original dataset, which is derived based on IRT scaling. The observed values range from -2.72 to 1.45.\\ \hline
		TMINS &C& Learning time in class per week (minutes). The observed values range from 0 to 3000. \\ \hline
	    UNFAIR &C& Teacher unfairness. This variable is named `unfairteacher' in the original dataset, which is derived by taking the sum of the responses to questionnaire items ST039Q01NA to ST039Q06NA (questions that ask students about how
		often in the past 12 months they had experienced unfair treatment by teachers). The observed values range from 1 to 24.\\ \hline		
		WEALTH &C& Family wealth. This variable is derived based on IRT scaling. The observed values range from -7.01 to 4.27. \\
		\bottomrule
	\end{longtable}
\end{center}

\section{Additional Simulation}\label{app:sim}

\subsection{When IRT parameters are unknown}
{When the IRT parameters are unknown, they need to be jointly estimated within Algorithm \ref{alg:StEM}. To avoid identifiability issue, we fix 
$\beta_0$ and $\sigma^2$ to $0$ and $1$ 
respectively within the algorithm. 

The simulation results are presented in Table \ref{tab:simA1}. It can be seen that the performance of the revised procedure is nearly the same as the that of the procedure when IRT are known. The baseline algorithm controls PFER slightly over the
nominal level, while the derandomised knockoff method strictly controls the PFER  below the
nominal level and has a greater power in most cases.
}


\begin{table}[H]
	\caption{Simulation results when  IRT parameters are unknown. Here, ``Baseline" refers to the baseline algorithm, Algorithm 2, and ``DRM" refers to derandomised knockoffs, Algorithm 3. 
 $\nu$ refers to the nominal PFER level.} 
	\label{tab:simA1}
	\centering
	\adjustbox{max height=0.5\textheight}
	{
	\begin{tabular}{cccccccc}			
		\toprule
		&&& $\nu=1$  &  $\nu=2$ & $\nu=3$ & $\nu=4$ & $\nu=5$ \\
		\midrule
		& &  Baseline &  1.03&  2.06&  3.12&   3.90& 4.72 \\
		\cmidrule{3-8}
		&\multirow{-2}{*}{PFER} &  DRM &   0.02&   0.11&   0.29&   0.50& 0.79   \\	
		\cmidrule{2-8}
		& & Baseline &  55.8\% &  65.4\% &  71.0\% &  74.1\% &  76.4\%\\
		\cmidrule{3-8} 
		\multirow{-5}{*}{N = 1000}&\multirow{-2}{*}{TPR}& DRM &  60.8\% &  68.8\% &  74.0\% &  77.5\% &  80.5\%\\
            
		\midrule
		& &  Baseline  & 1.08 & 2.16 & 3.31 & 4.23 & 5.18\\
		\cmidrule{3-8}
		&\multirow{-2}{*}{PFER}&  DRM  & 0.04 & 0.32 & 0.57 & 0.98 & 1.37 \\		
		\cmidrule{2-8}
		& & Baseline &  81.7\% & 88.3\% & 91.1\% & 92.6\% & 93.7\%\\
		\cmidrule{3-8} 
		\multirow{-5}{*}{N = 2000}&\multirow{-2}{*}{TPR}& DRM &  84.4\% & 90.2\% & 93.5\% & 95.2\% & 96.0\%\\              
		\midrule

		& &  Baseline &   1.19 & 2.19 &3.31 & 4.36 &5.24 \\
		\cmidrule{3-8}
		&\multirow{-2}{*}{PFER}&
            DRM & 0.15 & 0.52 & 0.96 & 1.43& 2.08\\			
		\cmidrule{2-8}
		& & Baseline & 96.5 \%& 98.5\%&99.0\%&99.4\%&99.6\%\\
		\cmidrule{3-8} 
		\multirow{-5}{*}{N = 4000}&  \multirow{-2}{*}{TPR}& DRM & 98.1\%&99.3\%&99.7\%&99.7\%&99.7\%\\
		\bottomrule
	\end{tabular}}
\end{table}

\subsection{Comparison with imputation-based stepwise AIC/BIC selection}
{To compare performance, we test an imputation-based stepwise model selection procedure, using either AIC or BIC, on the same simulation data as in Section~\ref{sec:sim}. The procedure is as follows:

First, we use the Gibbs sampler in \ref{sec:Gibbs} to sample $(\ttt_i, \ZZ_i^*)$ from their conditional distribution given $\ZZo_i, \WW_i,$ and $\YYo_i$.
These samples were obtained based on the estimated latent regression IRT model, where the IRT parameters are assumed to be known. 

Second, we initiate the model selection procedure by starting with the full model (i.e., $\theta \sim 1 + Z$), and iteratively add or remove one covariate in a greedy manner at each iteration, such that the updated model exhibits the most significant decrease in AIC or BIC after each step. This procedure continues until the model does not change any more.

The simulation results of such procedure are reported in Table \ref{tab:simA2}. 
These results indicate that the stepwise AIC/BIC procedure  tends to select a model that is significantly larger than the true model.
 Although the TPR of the selected model remains consistently high, there is no guarantee regarding the control of the PFER.
}

\begin{table}[H]
	\caption{Simulation results of the stepwise AIC/BIC procedure. Here, ``AIC" refers to the stepwise algorithm based on AIC, and ``BIC" refers to the stepwise algorithm based on BIC.} 
	\label{tab:simA2}
	\centering
	\adjustbox{max height=0.5\textheight}
	{
	\begin{tabular}{ccccc}			
		\toprule
		&& $N=1000$  &  $N=2000$ & $N=4000$\\
		\midrule
		&AIC &  43.75 &  40.18 &  33.78 \\
		\cmidrule{2-5}
		\multirow{-2}{*}{PFER} &  BIC &  14.92  &    10.46 &  5.85 \\	\midrule
		&AIC &  99.2\% &  100\% &  98.3\%  \\
		\cmidrule{2-5}
		\multirow{-2}{*}{TPR} &  BIC &  96.9\% &  99.7\% &  94.9\% \\

            

		\bottomrule
	\end{tabular}}
\end{table}

\subsection{Including latent predictors}
In this simulation study, we extend our setting to include latent predictors. The simulation setup is similar to that described in Section \ref{sec:sim}, with some modifications. 

Specifically, in this simulation study, we designate the 1st, 3rd, 5th, 7th, and 9th continuous variables within each block to be latent. In other words, for each $j$ in the form of $10u+v$ where $u, v \in \{1, 3, 5, 7, 9\}$, $Z_{ij}$ becomes a latent construct. Additionally, we introduce five continuous indicators, denoted as $\mathbf{W}_{ij} = (W_{ij1}, \ldots, W_{ij5})^\top \in \mathbb{R}^5$, associated with these $Z_{ij}$, such that
\[W_{ijl} = Z_{ij} + \varepsilon_{ijl},~l = 1, ..., 5.\]
where $\varepsilon_{ijl}$'s are independent stabdard Gaussian noise.
For the sake of simplicity, in this simulation study, we assume that all parameters in the conditional model of $\mathbf{W}_{ij}$ given $Z_{ij}$ are known.

Besides, we make modifications to the missing mechanism described in Section \ref{sec:sim}, allowing it to potentially depend on $\mathbf{W}_{ij}$. Similar to Section \ref{sec:sim}, the following steps are taken:

For each observation $i$, we generate a random variable $R_i$ from a categorical distribution with support ${1, 2, ..., 5}$, where $P(R_i = k) = 0.2$ for all $k = 1, ..., 5$. Let $\mathcal{S}_k^*$ denote the set of non-null variables in the $k$th block. For observation $i$, when $R_i = k$, we ensure that all the variables in $\mathcal{S}_k^*$ which are not latent are observed.
We also define $\bar{X}_{ij}$ for each $j$ as follows:
\[\bar{X}_{ij} = \begin{cases}
    \frac{1}{5}\sum_{l=1}^5 W_{ijl}, & \text{if } Z_{ij} \text{ is latent, }\\
    Z_{ij}, & \text{otherwise}.
\end{cases}\]
Then, for each of variables $j\notin\SSS^*$ that are not latent, we set its probability of being missing as
$(1+\exp(1-(\sum_{j'\in \mathcal{S}_k^*} \bar{X}_{ij'})/2))^{-1}.$
Under this setting, around 35\% of the entries of the data matrix (not including $\WW_{ij}$'s) for predictors are missing. 

In Table \ref{tab:simA3}, we present the simulation results under the described setting. Both the baseline algorithm the derandomised knockoff method 
 control PFER below the
nominal level. However, there seems to be a decrease in power compared to the results in Section~\ref{sec:sim}. We conjecture that there may be more powerful knockoff statistics in this setting, which we leave for future investigation.
\begin{table}[H]
	\caption{Simulation results involving latent predictors. Here, ``Baseline" refers to the baseline algorithm, Algorithm 2, and ``DRM" refers to derandomised knockoffs, Algorithm 3. 
 $\nu$ refers to the nominal PFER level.} 
	\label{tab:simA3}
	\centering
	\adjustbox{max height=0.5\textheight}
	{
	\begin{tabular}{cccccccc}			
		\toprule
		&&& $\nu=1$  &  $\nu=2$ & $\nu=3$ & $\nu=4$ & $\nu=5$ \\
		\midrule
		& &  Baseline &  0.38&  0.82&  1.63&   2.62& 3.56 \\
		\cmidrule{3-8}
		&\multirow{-2}{*}{PFER} &  DRM &   0.01&   0.04&   0.09&   0.20& 0.40   \\	
		\cmidrule{2-8}
		& & Baseline &  33.2\% &  42.4\% &  48.7\% &  55.0\% &  58.4\%\\
		\cmidrule{3-8} 
		\multirow{-5}{*}{N = 1000}&\multirow{-2}{*}{TPR}& DRM &  27.1\% &  36.0\% &  43.9\% &  49.8\% &  55.0\%\\
            
		\midrule
		& &  Baseline   & 0.15 & 0.59 & 1.60 & 2.58 & 3.46\\
		\cmidrule{3-8}
		&\multirow{-2}{*}{PFER}&  DRM  & 0.00 & 0.00 & 0.10 & 0.32 & 0.72 \\		
		\cmidrule{2-8}
		& & Baseline   & 40.2 \%& 54.4\%& 66.4\%& 72.6\%& 74.9\%\\
		\cmidrule{3-8} 
		\multirow{-5}{*}{N = 2000}&\multirow{-2}{*}{TPR}& DRM &  34.3\% & 49.2\% & 62.2\% & 69.0\% & 71.8\%\\              
		\midrule

		& &  Baseline &   0.05 & 0.27 & 1.35 & 2.24 & 3.37 \\
		\cmidrule{3-8}
		&\multirow{-2}{*}{PFER}&
            DRM & 0.00 & 0.00 & 0.15 & 0.62 & 1.02\\			
		\cmidrule{2-8}
		& & Baseline & 35.4 \%& 51.7\%& 71.3\%& 79.2\%& 79.5\%\\
		\cmidrule{3-8} 
		\multirow{-5}{*}{N = 4000}&  \multirow{-2}{*}{TPR}& DRM & 31.3\%& 49.1\%& 73.4\%&77.3\%&77.4\%\\
		\bottomrule
	\end{tabular}}
\end{table}

\bibliography{ref}

@article{janson2016familywise,
	title={Familywise error rate control via knockoffs},
	author={Janson, Lucas and Su, Weijie},
	journal={Electronic Journal of Statistics},
	volume={10},
	number={1},
	pages={960--975},
	year={2016},
	publisher={Institute of Mathematical Statistics and Bernoulli Society}
}

@article{ren2021derandomizing,
author = {Ren, Zhimei  and Wei, Yuting  and Cand{\`e}s, Emmanuel},
title = {Derandomizing Knockoffs},
journal = {Journal of the American Statistical Association},
year  = {2023},
volume = {118}, 
pages = {948-958},
publisher = {Taylor & Francis}
}

@book{little2019statistical,
	title={Statistical analysis with missing data},
	author={Little, Roderick JA and Rubin, Donald B},
	year={2019},
	publisher={Wiley},
	address={Hoboken, NJ}
}

@article{zhang2022computation,
  title={Computation for latent variable model estimation: A unified stochastic proximal framework},
  author={Zhang, Siliang and Chen, Yunxiao},
  journal={Psychometrika},
  volume={87}, 
  pages={1473--1502},
  year={2022}
}

@article{fan2017high,
  title={High dimensional semiparametric latent graphical model for mixed data},
  journal={{Journal of the Royal Statistical Society: Series B (Methodological)}},
  author={Fan, Jianqing and Liu, Han and Ning, Yang and Zou, Hui},
  volume={79},
  number={2},
  pages={405--421},
  year={2017}
}

@article{barber2020robust,
  title={Robust inference with knockoffs},
  author={Barber, Rina Foygel and Cand{\`e}s, Emmanuel J and Samworth, Richard J},
  journal={The Annals of Statistics},
  volume={48},
  number={3},
  pages={1409--1431},
  year={2020},
  publisher={Institute of Mathematical Statistics}
}

@article{spector2022powerful,
  title={Powerful knockoffs via minimizing reconstructability},
  author={Spector, Asher and Janson, Lucas},
  journal={The Annals of Statistics},
  volume={50},
  number={1},
  pages={252--276},
  year={2022},
  publisher={Institute of Mathematical Statistics}
}

@book{skrondal2004generalized,
  title={Generalized latent variable modeling: Multilevel, longitudinal, and structural equation models},
  author={Skrondal, Anders and Rabe-Hesketh, Sophia},
  year={2004},
  publisher={Chapman and Hall/CRC},
address = {Boca Raton, FL}
}

@article{grund2021treatment,
  title={On the treatment of missing data in background questionnaires in educational large-scale assessments: {An} evaluation of different procedures},
  author={Grund, Simon and L{\"u}dtke, Oliver and Robitzsch, Alexander},
  journal={Journal of Educational and Behavioral Statistics},
  volume={46},
  number={4},
  pages={430--465},
  year={2021},
  publisher={SAGE Publications Sage CA: Los Angeles, CA}
}

@book{singer2018international,
  title={International education assessments: Cautions, conundrums, and common sense.},
  author={Singer, Judith D and Braun, Henry I and Chudowsky, Naomi},
  year={2018},
  publisher={National Academy of Education},
  address={Washington, DC}
}

@book{von2012role,
  title={The role of international large-scale assessments: Perspectives from technology, economy, and educational research},
  author={{von Davier}, Matthias and Gonzalez, Eugenio and Kirsch, Irwin and Yamamoto, Kentaro},
  year={2012},
  publisher={Springer},
  address={New Yor, NY}
}

@book{national2012framework,
  title={A framework for {K-12} science education: Practices, crosscutting concepts, and core ideas},
  author={{National Research Council}},
  year={2012},
  publisher={National Academies Press},
  address={Washington, DC}
}

@misc{uk2015national,
  title={National curriculum in {England}: Science programmes of study},
  author={{Gov.UK}},
  year={2015},
  howpublished = "\url{https://www.gov.uk/government/publications/national-curriculum-in-england-science-programmes-of-study/national-curriculum-in-england-science-programmes-of-study}"
}

@article{gonzalez2010principles,
  title={Principles of multiple matrix booklet designs and parameter recovery in large-scale assessments},
  author={Gonzalez, E and Rutkowski, L},
  journal={IEA-ETS Research Institute Monograph},
  volume={3},
  pages={125--156},
  year={2010}
}

@article{jacobucci2019practical,
  title={A practical guide to variable selection in structural equation modeling by using regularized multiple-indicators, multiple-causes models},
  author={Jacobucci, Ross and Brandmaier, Andreas M and Kievit, Rogier A},
  journal={Advances in Methods and Practices in Psychological Science},
  volume={2}, 
  pages={55--76},
  year={2019},
  publisher={Sage Publications Sage CA: Los Angeles, CA}
}

@article{von2010stochastic,
  title={Stochastic approximation methods for latent regression item response models},
  author={{von Davier}, Matthias and Sinharay, Sandip},
  journal={Journal of Educational and Behavioral Statistics},
  volume={35},
  number={2},
  pages={174--193},
  year={2010},
  publisher={SAGE Publications Sage CA: Los Angeles, CA}
}

@article{mislevy1984estimating,
  title={Estimating latent distributions},
  author={Mislevy, Robert J},
  journal={Psychometrika},
  volume={49},
  number={3},
  pages={359--381},
  year={1984},
  publisher={Springer}
}

@article{candes2018panning,
  title={Panning for gold: `model-{X}' knockoffs for high dimensional controlled variable selection},
  author={Cand{\`e}s, Emmanuel and Fan, Yingying and Janson, Lucas and Lv, Jinchi},
  journal={Journal of the Royal Statistical Society: Series B (Statistical Methodology)},
  volume={80},
  number={3},
  pages={551--577},
  year={2018},
  publisher={Wiley Online Library}
}

@article{barber2015controlling,
  title={Controlling the false discovery rate via knockoffs},
  author={Barber, Rina Foygel and Cand{\`e}s, Emmanuel J},
  journal={The Annals of Statistics},
  volume={43},
  number={5},
  pages={2055--2085},
  year={2015},
  publisher={Institute of Mathematical Statistics}
}

@article{duckworth2015measurement,
  title={Measurement matters: Assessing personal qualities other than cognitive ability for educational purposes},
  author={Duckworth, Angela L and Yeager, David Scott},
  journal={Educational Researcher},
  volume={44},
  number={4},
  pages={237--251},
  year={2015},
  publisher={Sage Publications Sage CA: Los Angeles, CA}
}

@article{richardson2012psychological,
  title={Psychological correlates of university students' academic performance: A systematic review and meta-analysis.},
  author={Richardson, Michelle and Abraham, Charles and Bond, Rod},
  journal={Psychological Bulletin},
  volume={138},
  number={2},
  pages={353 - 387},
  year={2012},
  publisher={American Psychological Association}
}

@article{han2012composite,
  title={A composite likelihood approach to latent multivariate {Gaussian} modeling of {SNP} data with application to genetic association testing},
  author={Han, Fang and Pan, Wei},
  journal={Biometrics},
  volume={68},
  number={1},
  pages={307--315},
  year={2012},
  publisher={Wiley Online Library}
}

@article{farkas2003cognitive,
  title={Cognitive skills and noncognitive traits and behaviors in stratification processes},
  author={Farkas, George},
  journal={Annual Review of Sociology},
  volume={29},
  pages={541--562},
  year={2003},
  publisher={JSTOR}
}

@article{lee2018non,
  title={Non-cognitive predictors of academic achievement: Evidence from {TIMSS} and {PISA}},
  author={Lee, Jihyun and Stankov, Lazar},
  journal={Learning and Individual Differences},
  volume={65},
  pages={50--64},
  year={2018},
  publisher={Elsevier}
}

@book{organisation2016pisa,
  title={{PISA} 2015 {Technical Report}},
  author={OECD},
  publisher={OECD publishing},
  year={2016},
  address={Paris, France}
}

@article{von2009plausible,
  title={What are plausible values and why are they useful?},
  author={{von Davier}, Matthias and Gonzalez, Eugenio and Mislevy, Robert},
  journal={IERI Monograph Series},
  volume={2},
  number={1},
  pages={9--36},
  year={2009}
}

@book{embretson2000item,
  title={{Item Response Theory for Psychologists}},
  author={Embretson, Susan E and Reise, Steven P},
  year={2000},
  publisher={Lawrence Erlbaum},
  address={Mahwah, NJ}
}

@incollection{birnbaum1968some,
  title={Some latent trait models},
author={Birnbaum, A.},
  booktitle={Statistical theories of mental test scores},
  year={1968},
  editor = {Lord, Frederic M.  and Novick, Melvin Robert },
  publisher={Addison-Wesley}, 
address = {Reading, MA}
}

@article{muraki1992generalized,
  title={A generalized partial credit model: Application of an {EM} algorithm},
  author={Muraki, Eiji},
  journal={Applied Psychological Measurement},
  volume={16},
  number={2},
  pages={159--176},
  year={1992},
  publisher={Wiley Online Library}
}

@book{van2018flexible,
  title={Flexible imputation of missing data},
  author={Van Buuren, Stef},
  year={2018},
  publisher={CRC Press},
  address={London, England}
}

@article{nielsen2000stochastic,
  title={The stochastic {EM} algorithm: Estimation and asymptotic results},
  author={Nielsen, S{\o}ren Feodor},
  journal={Bernoulli},
  volume={6},
  number={3},
  pages={457--489},
  year={2000},
  publisher={JSTOR}
}

@article{zhang2020improved,
  title={An improved stochastic {EM} algorithm for large-scale full-information item factor analysis},
  author={Zhang, Siliang and Chen, Yunxiao and Liu, Yang},
  journal={British Journal of Mathematical and Statistical Psychology},
  volume={73},
  number={1},
  pages={44--71},
  year={2020},
  publisher={Wiley Online Library}
}

@article{sesia2019gene,
  title={Gene hunting with hidden {Markov} model knockoffs},
  author={Sesia, Matteo and Sabatti, Chiara and Cand{\`e}s, Emmanuel J},
  journal={Biometrika},
  volume={106},
  number={1},
  pages={1--18},
  year={2019},
  publisher={Oxford University Press}
}

@article{romano2020deep,
  title={Deep knockoffs},
  author={Romano, Yaniv and Sesia, Matteo and Cand{\`e}s, Emmanuel},
  journal={Journal of the American Statistical Association},
  volume={115},
  number={532},
  pages={1861--1872},
  year={2020},
  publisher={Taylor \& Francis}
}

@article{fan2020ipad,
  title={{IPAD: stable interpretable forecasting with knockoffs inference}},
  author={Fan, Yingying and Lv, Jinchi and Sharifvaghefi, Mahrad and Uematsu, Yoshimasa},
  journal={Journal of the American Statistical Association},
  volume={115},
  number={532},
  pages={1822--1834},
  year={2020},
  publisher={Taylor \& Francis}
}

@article{chen2023IRT,
  title={Item Response Theory -- A Statistical Framework for Educational and Psychological Measurement},
  author={Chen, Yunxiao and Li, Xiaoou and Liu, Jingchen and Ying, Zhiliang},
  journal={Statistical Science},
  year={2023},
note = {To appear}
}

@article{barber2019knockoff,
  title={A knockoff filter for high-dimensional selective inference},
  author={Barber, Rina Foygel and Cand{\`e}s, Emmanuel J},
  journal={The Annals of Statistics},
  volume={47},
  number={5},
  pages={2504--2537},
  year={2019},
  publisher={Institute of Mathematical Statistics}
}

@article{fan2019rank,
  title={{RANK: Large-scale inference with graphical nonlinear knockoffs}},
  author={Fan, Yingying and Demirkaya, Emre and Li, Gaorong and Lv, Jinchi},
  journal={Journal of the American Statistical Association},
  year={2019},
  volume={115},
  number={529},
  pages={362--379},
  publisher={Taylor \& Francis}
}

@article{liu2014stationary,
  title={On the stationary distribution of iterative imputations},
  author={Liu, Jingchen and Gelman, Andrew and Hill, Jennifer and Su, Yu-Sung and Kropko, Jonathan},
  journal={Biometrika},
  volume={101},
  number={1},
  pages={155--173},
  year={2014},
  publisher={Oxford University Press}
}

@article{ren2022derandomized,
  title={Derandomized knockoffs: Leveraging e-values for false discovery rate control},
  author={Ren, Zhimei and Barber, Rina Foygel},
  journal={arXiv preprint arXiv:2205.15461},
  year={2022}
}

@book{pena2016pisa,
  title={PISA 2015 {Results (Volume II). Policies and Practices for Successful Schools}},
  author={OECD},
  year={2016},
  publisher={OECD Publishing},
  address={Paris, France},
  howpublished="\url{http://dx.doi.org/10.1787/9789264267510-en}"
}

@article{gilks1992adaptive,
  title={Adaptive rejection sampling for {Gibbs} sampling},
  author={Gilks, Walter R and Wild, Pascal},
  journal={Journal of the Royal Statistical Society: Series C (Applied Statistics)},
  volume={41},
  number={2},
  pages={337--348},
  year={1992},
  publisher={Wiley Online Library}
}

@book{givens2012computational,
  title={Computational statistics},
  author={Givens, Geof H and Hoeting, Jennifer A},
  year={2012},
  publisher={Wiley},
  address={Hoboken, NJ}
}

@article{kuha2018latent,
  title={Latent variable modelling with non-ignorable item non-response: Multigroup response propensity models for cross-national analysis},
  author={Kuha, Jouni and Katsikatsou, Myrsini and Moustaki, Irini},
  journal={Journal of the Royal Statistical Society: Series A (Statistics in Society)},
  volume={181},
  number={4},
  pages={1169--1192},
  year={2018},
  publisher={Wiley Online Library}
}

@article{akaike1974new,
  title={A new look at the statistical model identification},
  author={Akaike, Hirotugu},
  journal={IEEE Transactions on Automatic Control},
  volume={19},
  number={6},
  pages={716--723},
  year={1974},
  publisher={Ieee}
}

@article{schwarz1978estimating,
  title={Estimating the dimension of a model},
  author={Schwarz, Gideon},
  journal={The Annals of Statistics},
  volume={6}, 
  pages={461--464},
  year={1978},
  publisher={JSTOR}
}

@book{cohen1975,
  title={ Applied multiple regression/correlation analysis for the behavioral sciences},
  author={Cohen, J and Cohen, P},
  year={1975},
  publisher = {Lawrence Erlbaum Associates},
  address = {Hillsdale, NJ}
}

@article{dardanoni2015model,
  title={Model averaging estimation of generalized linear models with imputed covariates},
  author={Dardanoni, Valentino and De Luca, Giuseppe and Modica, Salvatore and Peracchi, Franco},
  journal={Journal of Econometrics},
  volume={184},
  number={2},
  pages={452--463},
  year={2015},
  publisher={Elsevier}
}

@article{dardanoni2011regression,
  title={Regression with imputed covariates: A generalized missing-indicator approach},
  author={Dardanoni, Valentino and Modica, Salvatore and Peracchi, Franco},
  journal={Journal of Econometrics},
  volume={162},
  number={2},
  pages={362--368},
  year={2011},
  publisher={Elsevier}
}

@article{friedman2009glmnet,
  title={glmnet: Lasso and elastic-net regularized generalized linear models},
  author={Friedman, Jerome},
  journal={R package version},
  volume={1},
  number={4},
  year={2009}
}

\end{document}